%% file: main.tex
     \documentclass[sigplan,authorversion,nonacm,10pt]{acmart}
\startPage{1}

\setcopyright{none}

\usepackage{subcaption}
\usepackage{multirow}
\usepackage{amsmath}
\usepackage{array}
\newcommand{\parops}{\textsc{Par-bin-ops}}
\newcommand{\quantlib}{\textsc{QuantLib}}
\newcommand{\highlight}[1]{\textit{\textbf{#1}}}

\newcommand{\para}[1]{\vspace{0.1cm}\noindent{\textbf{#1}}}

\newcommand{\floor}[1]      {\left\lfloor #1 \right\rfloor}
\newcommand{\ceil}[1]       {\left\lceil #1 \right\rceil}

\newcommand{\Oh}[1]{{\mathcal O}\left({#1}\right)}

\newcommand{\oh}[1]{{o}\left({#1}\right)}
\newcommand{\Om}[1]{{\Omega}\left({#1}\right)}

\newcommand{\Th}[1]{{\Theta}\left({#1}\right)}

\newcommand{\hide}[1]{}

\newcommand{\xfor}{{\bf{{for~}}}}

\newcommand{\xto}{{\bf{{to~}}}}
\newcommand{\xdownto}{{\bf{{downto~}}}}

\newcommand{\xdo}{{\bf{{do~}}}}

\newcommand{\xreturn}{{\bf{{return~}}}}

\newcommand{\xparallelfor}{{\bf{{parallel for~}}}}

\newcommand{\T}{\hspace{0.3cm}}
\newcommand{\m}{\mathcal}

\def\probm{o}

\definecolor{gray}{rgb}{0.3,0.3,0.3}

\def\func#1{\textrm{\bf{\sc{#1}}}}

\begin{document}

\setlength{\abovedisplayskip}{5pt}
\setlength{\belowdisplayskip}{5pt}

\title{Fast American Option Pricing using Nonlinear Stencils}

\author{Zafar Ahmad}
\affiliation{%
  \institution{Stony Brook University}
  \city{Stony Brook}
  \country{USA}}
\email{zafahmad@cs.stonybrook.edu}

\author{Reilly Browne}
\affiliation{%
  \institution{Stony Brook University}
  \city{Stony Brook}
  \country{USA}}
\email{rjbrowne@cs.stonybrook.edu}

\author{Rezaul Chowdhury}
\affiliation{%
  \institution{Stony Brook University}
  \city{Stony Brook}
  \country{USA}}
\email{rezaul@cs.stonybrook.edu}

\author{Rathish Das}
\affiliation{%
  \institution{University of Houston}
  \city{Houston}
  \country{USA}}
\email{rathish@uh.edu}

\author{Yushen Huang}
\affiliation{%
  \institution{Stony Brook University}
  \city{Stony Brook}
  \country{USA}}
\email{yushuang@cs.stonybrook.edu}

\author{Yimin Zhu}
\affiliation{%
  \institution{Stony Brook University}
  \city{Stony Brook}
  \country{USA}}
\email{yimzhu@cs.stonybrook.edu}

\renewcommand{\shortauthors}{Ahmad, Browne, Chowdhury, Das, Huang, Zhu}

\begin{abstract}
We study the binomial, trinomial, and Black-Scholes-Merton models of option pricing. We present fast parallel discrete-time finite-difference algorithms for American call option pricing under the binomial and trinomial models and American put option pricing under the Black-Scholes-Merton model. For $T$-step finite differences, each algorithm runs in $\Oh{{\left(T\log^2{T}\right)}/{p} + T}$ time under a greedy scheduler on $p$ processing cores, which is a significant improvement over the $\Th{{T^2}/{p}} + \Om{T\log{T}}$ time taken by the corresponding state-of-the-art parallel algorithm. Even when run on a single core, the $\Oh{T\log^2{T}}$ time taken by our algorithms is asymptotically much smaller than the $\Th{T^2}$ running time of the fastest known serial algorithms. Implementations of our algorithms significantly outperform the fastest implementations of existing algorithms in practice, e.g., when run for $T \approx 1000$ steps on a 48-core machine, our algorithm for the binomial model runs at least $15\times$ faster than the fastest existing parallel program for the same model with the speed-up factor gradually reaching beyond $500\times$ for $T \approx 0.5 \times 10^6$. It saves more than 80\% energy when $T \approx 4000$, and more than 99\% energy for $T > 60,000$.

Our option pricing algorithms can be viewed as solving a class of nonlinear 1D stencil (i.e., finite-difference) computation problems efficiently using the Fast Fourier Transform (FFT). To our knowledge, ours are the first algorithms to handle such stencils in $\oh{T^2}$ time. These contributions are of independent interest as stencil computations have a wide range of applications beyond quantitative finance.

\end{abstract}



\received{20 February 2007}
\received[revised]{12 March 2009}
\received[accepted]{5 June 2009}

\maketitle

\input{1-Introduction.tex}
\input{3-BOPM.tex}
\input{4-TOPM.tex}
\input{5-BSM.tex}

\input{6-experiment.tex}

\vspace{0.2in}
\noindent
\textbf{Acknowledgements} We want to thank Julian Shun, Yuanhao Wei, Shangdi Yu, Pranali Vani, and Sabiyyah Ali for valuable discussions and a careful review of the initial draft of the paper which helped enhance its quality and ensure the correctness.




\bibliography{main}

\appendix
\input{appendix_arxiv}

\end{document}

%% file: 1-Introduction.tex

\vspace{-0.2cm}
\section{Introduction}
\vspace{-0.1cm}

Option pricing or computing the value of a contract giving one the right to buy/sell an asset under some given constraints is one of the most important computational problems in quantitative finance \cite{ames2014numerical}. Rapid changes in financial markets often lead to rapid changes in asset prices which makes the ability to quickly estimate option prices essential in avoiding potential financial losses \cite{peter2007understanding}.

An \highlight{option} is a two-party financial contract that gives one party (called the \highlight{holder}) the right (but not an obligation) to buy/sell (i.e., exercise) an asset from/to the other party (called the \highlight{writer}) at a fixed price (called the \highlight{strike/exercise price}) on or before an expiration date (called the \highlight{exercise/maturity date}). A \highlight{call option} gives the right to buy whereas a \highlight{put option} gives the right to sell. Also, based on the expiration date and the settlement rule, there are two major styles of options: \highlight{European} and \highlight{American}. A European option can only be exercised at the expiration date while an American option can be exercised at any time before that.

The \highlight{option pricing} problem asks for assigning a \highlight{value} or \highlight{price} 
to an options contract based on the calculated probability that the contract will be exercised at expiration.
%
%
The theoretical value of an option \cite{merton1973theory,smith1976option,bensoussan1984theory,galai1976option,black1973pricing,kumar2012analytical} is determined by its \highlight{stock price} $S$ (i.e., its current market price), \highlight{strike price} $K$, \highlight{risk-free rate of return} $R$ (i.e., the theoretical rate of return assuming zero risk), \highlight{dividend yield} $Y$ (i.e., a ratio that shows how much dividend/year is paid relative to $S$), \highlight{volatility} $V$ (i.e., how much the trading price varies over time), and \highlight{time to expiration} $E$ (e.g., in days).%

\begin{table}[t!]
\caption{Notations}
\vspace{-0.35cm}
\scalebox{0.85}{
\begin{tabular}{cl|cl}
{\bf{Symbol}} & {\bf{Meaning}} & {\bf{Symbol}} & {\bf{Meaning}}\\ \hline
$S$ & stock price & $K$ & strike price\\
$R$ & risk-free rate of return & $V$ & volatility\\ 
$Y$ & dividend yield & $E$ & time to expire\\
$T$ & number of time steps & & (in days)\\\hline\\[-0.3cm]
\end{tabular}
}
\label{fig:notation}
\vspace{-0.7cm}
\end{table}

The earliest option pricing model \cite{bachelier1900theorie} was based on the assumption of the geometric Brownian motion for asset pricing and the no-arbitrage idea. Many improved models were developed later
\cite{samuelson2011louis,black1973pricing, merton1973theory,derman1998stochastic, dupire1994pricing,stein1991stock, heston1993closed, ball1994stochastic,bergomi2015stochastic,merton1976option, kou2002jump, bakshi2000spanning, duffie2000transform, dempster2002spread,matsuda2004introduction,runggaldier2003jump,madan1998variance,geman2001time, jackson2008fourier,derman1994riding, dupire1994pricing}. 

Analytical solutions to the option pricing problem are sometimes available, particularly for European options \cite{shreve2004stochastic,shreve2005stochastic,hull2003options,prathumwan2020solution}. But they are not available for American options except for a limited number of cases with significant constraints (e.g., American call options with zero or one dividend \cite{shreve2004stochastic} and perpetual American put options \cite{mackean1965free,shreve2004stochastic}). This difficulty in finding closed-form analytical solutions for most option pricing problems makes \textit{computational} approaches the only path forward. The main computational approaches to solving the option pricing problem include the binomial tree method \cite{karatzas1998methods}, the finite difference method \cite{thomas2013numerical,sewell2005numerical,holmes2007introduction,leveque2007finite,ames2014numerical,strikwerda2004finite}, and the Monte Carlo method \cite{duan1998empirical,longstaff2001valuing,tsitsiklis2001regression}.

The binomial tree method works by tracing the option's value at discrete time steps over the life of the option. For a given number of time steps $T$ between the valuation and expiration dates of the option, a binomial tree of height $T$ is created with the leaves storing the potential prices of the asset at the time of expiration. Then one works backward to compute for each $t \in [0, T - 1]$ the value of the nodes at depth $t$ of the tree (each giving a possible price at time step $t$) from the values of the nodes at depth $t + 1$ using a simple formula. The value computed for the root node is the required option value. Straightforward iterative implementation of this method runs in $\Th{T^2}$ time on a single processing core and $\Th{T^2 / p + T\log{T}}$ time on $p$ cores (see Figure \ref{fig:notations-and-BOPM-loops} and Table \ref{tbl:work_span}). 
It provides a discrete-time approximation of the continuous-time option pricing in the Black–Scholes model and is widely used by professional option traders.

\begin{figure}[t!]
\framebox{
\scalebox{0.92}{
 \begin{minipage}{0.48\textwidth}
 {\footnotesize
 \medskip\noindent\func{BOPM-American-Call}$\left(~S,~K,~R,~V,~Y,~E,~T~\right)$



 \noindent
 \begin{enumerate}

 \item $\Delta t \gets {E}/{T}$,~ $u \gets e^{V\sqrt{\Delta t}}$,~ $d \gets {1}/{u}$,~ $p \gets {(e^{(R - Y)\Delta t} - d)}/{(u - d)}$

 \item[] $m \gets e^{-R\Delta t}$,~ $s_0 \gets m p$,~ $s_1 \gets m(1 - p)$

 \item \xfor $j \gets 0$ \xto $T$ \xdo $G_{T,j} \gets \max{\left( 0, S u^{2j - T} - K\right)}$

 \item\label{step:nestedloop} \xfor $i \gets T - 1$ \xdownto $0$ \xdo

 \item[] \T \xparallelfor $j \gets 0$ \xto $i$ \xdo

 \item[] \T\T $G_{i,j} \gets \max{\left( s_0 G_{i+1, j} + s_1 G_{i+ 1, j + 1},~ S u^{2j - i} - K\right)}$

 \item \xreturn $G_{0, 0}$
 
 \end{enumerate}

 }
 \end{minipage}
 }
}
 \vspace{-0.3cm}  
\caption{Standard looping code for pricing
American Call options under the
Binomial Option Pricing Model.}
\label{fig:notations-and-BOPM-loops}
 \vspace{-0.5cm}  
\end{figure}

The trinomial tree method extends the binomial method by accounting for the possibility that an asset value remains the same after a time step \cite{boyle1986option}. With only a constant factor increase in run-time it provides more precise predictions than the binomial model.

The finite-difference method approximates the continuous-time differential equations describing the evolution of an option price over time by a set of discrete-time difference equations and then solves them iteratively under appropriate boundary conditions. The explicit finite difference method divides the lifetime of the option into $T$ discrete time steps and then uses the potential values of the asset at time step $T$ (the time of expiration) to compute the asset values at each time step $t \in [0, T - 1]$ from the asset values at time step $t + 1$ based on the difference equations (i.e., update equations or stencils). The final option value is found at time $t = 0$. Similar to the binomial tree method, the iterative implementation of this method runs in $\Th{T^2}$ time on a single processing core and $\Th{T^2 / p + T\log{T}}$ time on $p$ cores.
Other finite difference methods used for option pricing include implicit finite difference and the Crank–Nicolson method \cite{crank1947practical}.

The Monte Carlo method works by generating random backward paths the asset price may follow starting from the time of expiration and ending at the time of valuation. Each of these paths leads to a payoff value for the option and the average of these payoff values can be viewed as an expected value of the option. This method is used for pricing options with complicated features and/or multiple sources of uncertainty that other methods (analytical, tree-based, finite difference) cannot handle \cite{duan1998empirical}, but is usually not competitive when those methods apply as the convergence rate of Monte Carlo method is sublinear \cite{james1980monte,metropolis1949monte}. They are hard to develop for some options, such as Black-Scholes Model for the American put option, but still many results exist on Monte Carlo methods \cite{ibanez2004monte,longstaff2001valuing,tsitsiklis2001regression,boyle1977options}.

\para{Our Contributions.}
We present three shared-memory parallel algorithms for American option pricing -- call option under the binomial and trinomial option pricing models and put option under the Black-Scholes-Merton model. All three algorithms run in $\Th{\left(T \log^2{T}\right)/p + T}$ time on $p$ processing cores which is a significant improvement over the $\Om{T^2 / p + T\log{T}}$ time taken by the state-of-the-art parallel algorithms, where $T$ is the number of time steps. When run on a single processing core they run in $\Th{T \log^2{T}}$ time compared to the $\Th{T^2}$ time taken by the fastest existing serial algorithms. 
We use the \highlight{Fast Fourier Transform (FFT)} to speed up the computation. Table \ref{tbl:work_span} summarizes the results.

We use the \textit{work-span} model \cite{CormenLeRiSt2009} to analyze the performance of parallel programs. Let $\m{T}_p$ be the running time on a $p$-processor machine under a greedy scheduler. Then $\m{T}_1$ and $\m{T}_\infty$ are called \textit{work} and \textit{span}, respectively. The \textit{parallel running time} $\m{T}_p = \Th{\m{T}_1 / p + \m{T}_{\infty}}$.

\begin{table}[t]
    \caption{Parallel Algorithms
for American Option Pricing: The bounds hold for \textbf{call option pricing under the binomial and trinomial option pricing
models}, and \textbf{put option pricing under the Black-Scholes-Merton model}. Here, $T = $ number of time
steps, $p = $ number of processing cores, and $M = $ cache size. Also, $\m{T}_p = $ running time on $p$ cores,
and thus $\m{T}_1$ (Work) and $\m{T}_{\infty}$ (Span) represent run-times on one core and an unbounded number
of cores, respectively. Under a greedy scheduler, $\m{T}_p = \Th{{\m{T}_1}/{p} + \m{T}_\infty}$, which is asymptotically optimal.}
\vspace{-0.4cm}
    \label{tbl:work_span}
    \centering
    \scalebox{0.65}{
    \begin{tabular}{||>{\centering\arraybackslash}m{1.3in}|>{\centering\arraybackslash}m{0.6in}|>{\centering\arraybackslash}m{1in}|>{\centering\arraybackslash}m{1.5in}||}
        \hline\hline
         \multirow{2}{*}{\textbf{Algorithm}} & \textbf{Work} & \textbf{Span} & \textbf{Parallel Running Time}\\ 
         & \textbf{$\m{T}_{1} (T)$} & \textbf{$\m{T}_{\infty} (T)$} & \textbf{$\m{T}_{p} (T)$} \\\hline\hline
         Nested Loop (standard, see Figure \ref{fig:notations-and-BOPM-loops}) & \multirow{6}{*}{$\Th{T^2}$} & $\Th{T\log T}$ & $\Th{\frac{T^2}{p} + T\log{T}}$ \\ \cline{1-1} \cline{3-4} 
         Tiled Loop (cache-aware) \cite{brunelle2022parallelizing}  & & $\Th{TM + \frac{T}{M} \log{\frac{T}{M}} }$ & $\Th{\frac{T^2}{p} + T M + \frac{T}{M} \log{\frac{T}{M}}}$ \\  \cline{1-1} \cline{3-4}
         Recursive Tiling (cache-oblivious) \cite{brunelle2022parallelizing,tang2011pochoir,frigo2005cache, FrigoSt2009} & & $\Th{T^{\log_2 3}}$ & $\Th{\frac{T^2}{p} + T^{\log_2 3}}$ \\ \hline\hline
         Our Algorithms & $\Th{T \log^2 T}$ & $\Th{T}$ & $\Th{\frac{T\log^2{T}}{p} + T}$ \\ \hline\hline
    \end{tabular}
    }
\vspace{-0.6cm}
\end{table}

\hide{
Our major contributions in this paper are three $\Oh{T \log^2{T}}$ time algorithms for American option pricing improving over the $\Th{T^2}$ time taken by existing discrete-time algorithms, where $T$ is the number of time steps. Our first algorithm evaluates an American call option under the binomial option pricing model while the second one is an American put option under the Black-Scholes-Merton model. Both algorithms use \highlight{Fast Fourier Transforms (FFTs)} for speeding up the computation.
}

The following proposition, which follows easily from the complexities given in Table \ref{tbl:work_span},
notes that each of our parallel algorithms runs asymptotically faster than the corresponding fastest existing
parallel algorithm for every value of $p$.
  
\begin{proposition} 
\label{prop:Tp}
Let $\m{T}^{(old)}_p(T)$ be the running time of any existing algorithm from Table \ref{tbl:work_span} on $p$ cores, and let $\m{T}^{(new)}_p(T)$ denote the same for our algorithm. Then for every (positive) value of $p$ under a greedy scheduler: $\m{T}^{(new)}_p(T) = \oh{\m{T}^{(old)}_p(T)}$.
\end{proposition}

We have implemented our algorithms and compared their running times, energy consumption, and cache performance with those of the option pricing implementations available in the \parops{} framework \cite{brunelle2022parallelizing} developed recently in 2022. Implementations of our algorithms run orders of magnitude faster, consume significantly less energy, and usually incur far fewer L1 cache misses than those implementations.

\para{How Our Algorithms Differ from Existing FFT-based Option Pricing Algorithms.}
FFTs have been used for European option pricing before. European option pricing is simpler than American option pricing, e.g., the European version of the American option pricing algorithm shown in Figure \ref{fig:notations-and-BOPM-loops} can be obtained by replacing the assignment $G_{i,j} \gets \max{\left( s_0 G_{i+1, j} + s_1 G_{i+ 1, j + 1},~ S u^{2j - i} - K\right)}$ in Step \ref{step:nestedloop} with the simpler assignment $G_{i,j} \gets s_0 G_{i+1, j} + s_1 G_{i+ 1, j + 1}$. The absence of the `$\textbf{max}$' operator in this assignment makes an efficient evaluation of the doubly-nested loop in Step \ref{step:nestedloop} easier.

Black, Scholes, and Merton \cite{black1973pricing,merton1973theory} showed that the European option can be calculated using
a Parabolic PDE with infinite domain constraint. By using the Fourier transform, one gets an integral form for European options. To obtain the numerical value from the integral form, one uses numerical integration \cite{hildebrand1987introduction,burden2015numerical} which can be sped up using FFTs. 

There are also approximation results \cite{chang2007richardson,oliveira2014convolution,zhylyevskyy2010fast,lord2008fast} based on repeated Richardson extrapolation \cite{richardson1911approximate,richardson1927viii} and FFT for numerical integration in American options. However, even if the extrapolation is repeated only for a constant number of times for an option that expires in $E$ days, the approximation takes $\Om{\left({E}/{\Delta t} \right) N \log{N}}$ time when $N$ grid points are used to discretize the price of the underlying asset and ${E}/{\Delta t}$ exercise points are placed with every pair of consecutive exercise points being $\Delta t$ days apart. Observing that $T = {E}/{\Delta t}$ corresponds to the number of discrete time steps in the finite difference formulation of the problem, we can rewrite the complexity as $\Om{T N \log{N}}$. Usually, $N \geq T$ is used in practice \cite{oliveira2014convolution}, which reduces the complexity to $\Om{T^2 \log{T}}$.

A major weakness of the existing FFT-based numerical integration approach above is that a closed-form expression for the characteristic function of the log-price must be known for the technique to work. However, our approach does not need to know such a closed-form expression as we apply FFT to speed up stencil/finite-difference computations and not numerical integration. Thus, our approach will work on a larger set of option pricing problems. We are not restricted to infinite-domain problems either \cite{cont2005finite}.

\para{Implications for Nonlinear Stencil Computations.}
Our option pricing algorithms can be viewed as solving a class of \highlight{nonlinear 1D stencil computation problems} asking to evolve a grid of size $\Th{T}$ for $T$ time steps, in $\Th{T \log^2{T}}$ work (i.e., serial time) or $\Th{(T \log^2{T}) / p + T}$ parallel time on $p$ processing cores. As stencil algorithms, they are of independent interest. 

A \highlight{stencil} is a pattern (equation) used to update the values of cells in a spatial grid and evolve the grid over a given number of time steps. The process of evolving cell values in the spatial grid according to a stencil is called a \highlight{stencil computation} \cite{frigo2005cache}. The finite-difference method performs a stencil computation with an update equation derived from the differential equations used as a stencil. 
Stencils are widely used in various fields, including mechanical engineering \cite{paoli2002numerical,zhang2006weighted, renson2016numerical, szilard2004theories, rappaz2010numerical}, meteorology \cite{murray1991numerical,johnson2010numerical, robert1981stable, robert1982semi, avissar1989parameterization, kalnay1990global}, stochastic and fractional differential equations \cite{zhang2004numerical, yuan2002numerical,lord2004numerical}, chemistry \cite{najm1998semi,snider2010heterogeneous,long2008numerical,aubin2004modeling, han2015second}, electromagnetics \cite{komissarov2002time, taflove2005computational, van2012numerical, atangana2015numerical}, finance \cite{chen2008numerical}, and physics \cite{gammie2003harm,vijayaraghavan1990efficient,mangeney2002numerical,hirouchi2009development,cundall1992numerical, pang1999introduction, barth2013high, vesely1994computational, thijssen2007computational}, image processing \cite{weickert2000applications, peyre2011numerical, vese2002numerical}. 

Standard stencil algorithms perform $\Th{NT}$ work to evolve a grid of size $N$ for $T$ time steps, they include looping algorithms, tiled looping algorithms \cite{Bandishti2012, bondhugula2017,Wolfe1987,Wolf1991,Wolf1996,Wonnacott2002,bondhugula2016plutoplus,Andonov1997,Hogstedt1999, Zhang2015A3D, Malas2014TowardsEE,Malas2015}, and recursive divide-and-conquer algorithms \cite{frigo2005cache, FrigoSt2009, tang2011pochoir, Ekanathan2017, Katsuto2010, Sriram2007}. 

A stencil is called \highlight{linear} if it computes the value of a cell at time step $t$ as a fixed linear combination of cell values at time steps before $t$, otherwise it is called \highlight{nonlinear}. For \highlight{1D linear stencils} Ahmad et al. \cite{ahmad2021fast} provide FFT-based algorithms that take $\Oh{T\log{T}}$ serial time for periodic grids and $\Oh{T\log^{2}{T}}$ serial time for aperiodic grids, assuming that the input grid is of size $\Th{T}$. 

The \highlight{stencils we encounter in our current work are nonlinear} because they do not use a linear combination of cell values from prior time steps for updating a target cell, provided that the resulting value is smaller than the value of a function computed solely based on the spatial coordinates of the target cell and other option pricing parameters (e.g., see Step \ref{step:nestedloop} of Figure \ref{fig:notations-and-BOPM-loops}). Such a stencil divides the space-time grid into two disjoint regions -- in one region only the linear combination applies, while in the other only the function value applies. However, the problem is that the boundary between these two regions is not known ahead of time and the location of the boundary may move as the time step $t$ progresses. As a result, \highlight{Ahmad et al.'s \cite{ahmad2021fast} results for linear stencils do not apply}. To the best of our knowledge, ours are the first algorithms for handling such stencils running in time subquadratic in $T$.

%% file: 3-BOPM.tex
\begin{figure}
\begin{minipage}{0.23\textwidth}
    \centering
    \includegraphics[width=0.82\textwidth]{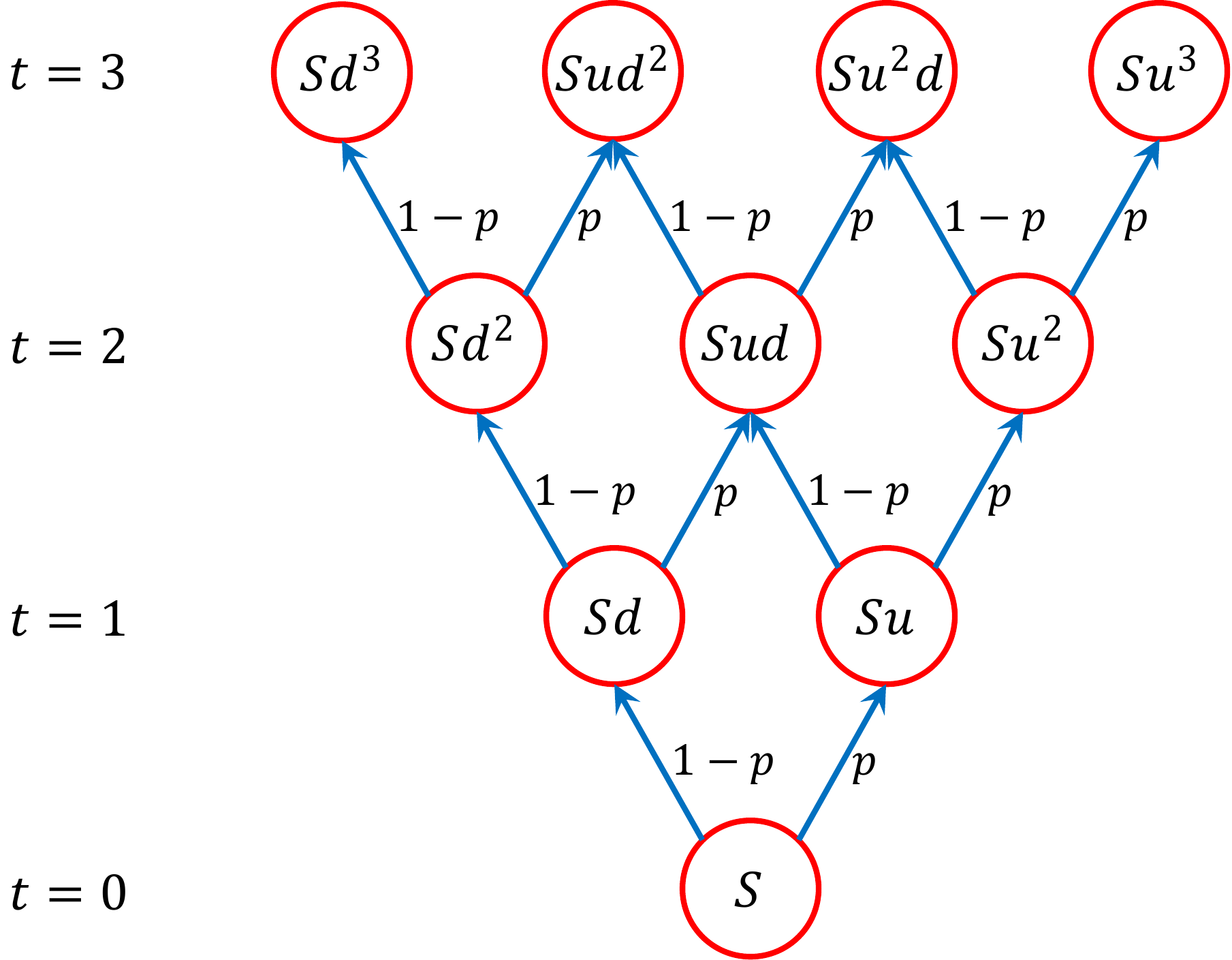}
    {\boldmath $(a)$}
\end{minipage}
\begin{minipage}{0.23\textwidth}
    \centering
    \includegraphics[width=0.6\textwidth]{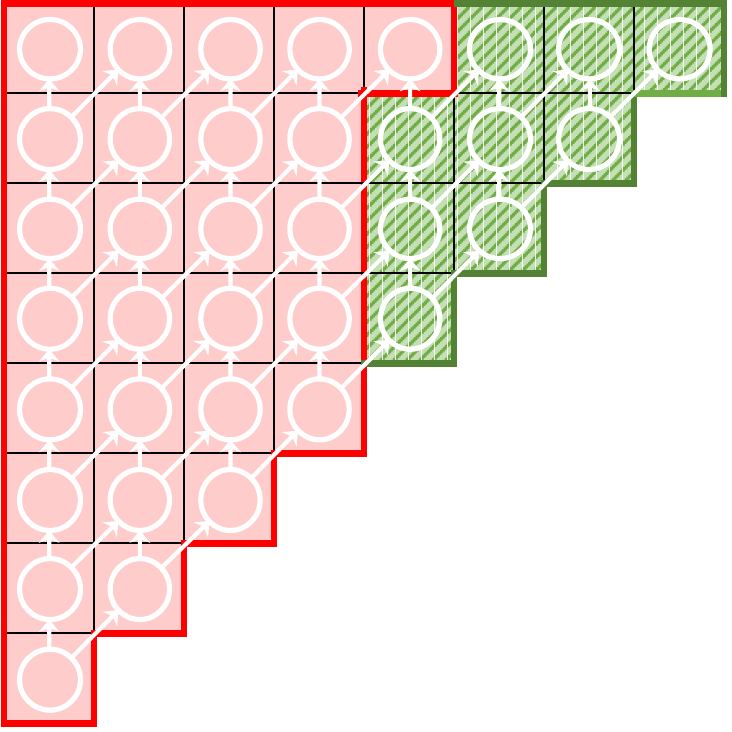}
    {\boldmath $(b)$}
\end{minipage}

\vspace{-0.2cm}
\caption{{\boldmath $(a)$} A 3 time-step binomial tree. {\boldmath $(b)$} A binomial tree of $7$ time steps embedded in an $8 \times 8$ grid. An upward arrow has a price change factor of ${1}/{u}$ while a rightward arrow has a price change factor of $u$. }
\label{fig:trees}
\vspace{-0.6cm}
\end{figure}

\vspace{-0.3cm}
\section{American Call Option under the Binomial Option Pricing Model}
\label{ssec:ACO-BOPM}

\subsection{Binomial Option Pricing Model (BOPM)}
\label{ssec:BOPM}
BOPM \cite{sharpe1978, cox1979option, rendleman1979two} is a simple discrete-time option pricing model without using advanced mathematical 
tools. It is a paradigm of practice.
%


\hide{
\begin{figure}[ht]
    \centering
        \includegraphics[width=0.7\linewidth]{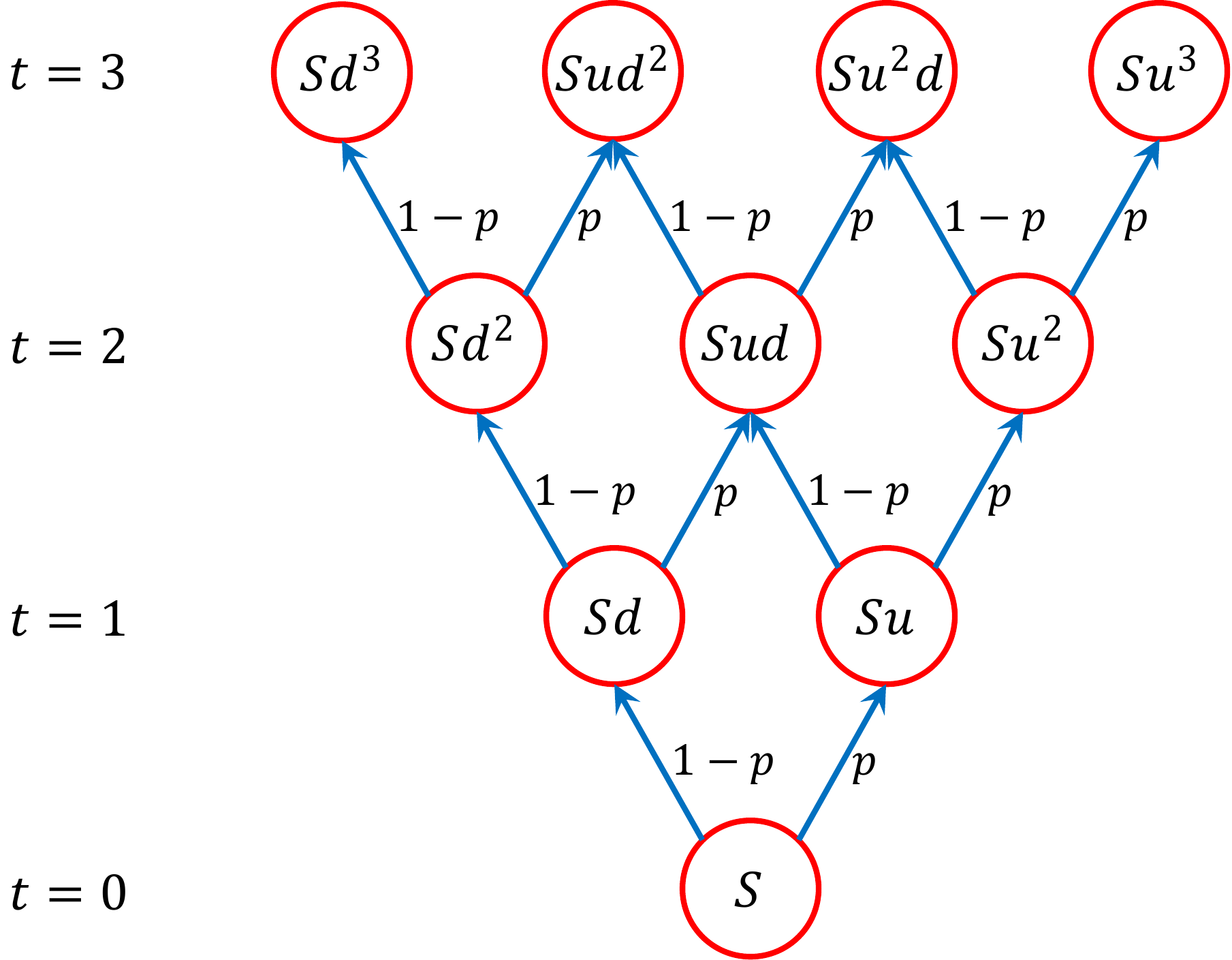}
        \caption{A 3 time-step binomial tree.}
        \label{img:bt2}
\end{figure}
\begin{figure}[ht]
        \includegraphics[width=0.6\linewidth]{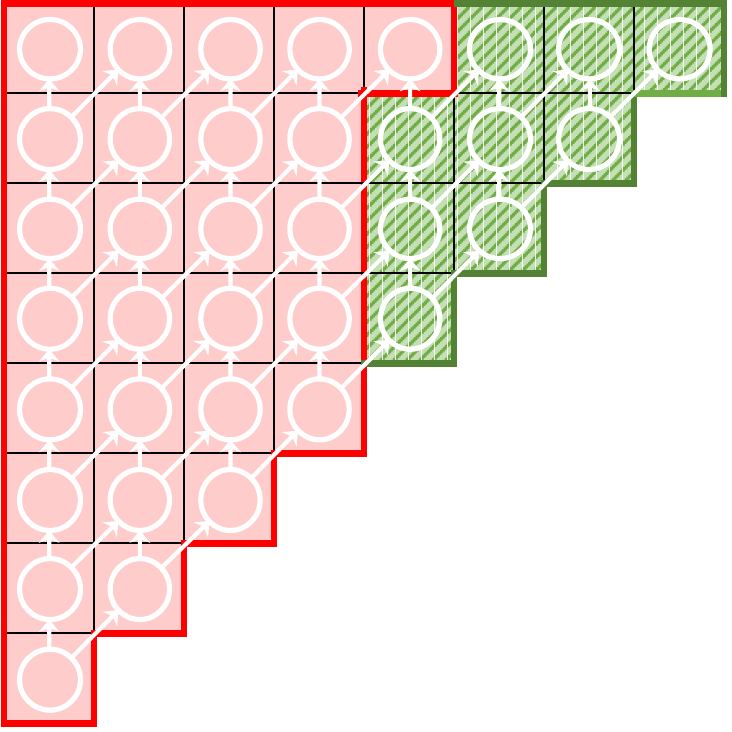}
        \caption{A binomial tree of $7$ time steps embedded in an $8 \times 8$ grid. In the binomial tree, an upward arrow has a price change factor of $\frac{1}{u}$ while a rightward arrow has a price change factor of $u$.}
        \label{fig:tritree}
\end{figure}
}

BOPM encodes the various sequences of prices the asset might take as paths in a binomial tree. Each node in the tree represents a possible price at a certain time and the nodes at two successive layers in the tree represent prices at times that are apart by some fixed time step $\Delta t$. The prices increase or decrease by some factor after every $\Delta t$ time.
Figure \ref{fig:trees}$(a)$ gives an example of a 3-time-step binomial price tree that is produced by moving from the valuation day to the expiration day. Denote the initial price by $S$. The price in the next time step (i.e., after $\Delta t$ time) can go up to $f_u = S\cdot u$ or go down to $f_d = S\cdot d$, where the up factor $u=e^{V\sqrt{\Delta t}}$ and the down factor $d={1}/{u}$ are determined by $\Delta t$ and volatility $V$.
\hide{
\begin{displaymath}
    u=e^{V\sqrt{\Delta t}}
\end{displaymath}
\begin{displaymath}d=e^{-V \sqrt{\Delta t}}=\frac{1}{u}
\end{displaymath}
}

Denote the node value as $X_{node} = S \times u^{N_u-N_d}$, where $N_u$ and $N_d$ are the numbers of ticks up and down, respectively. The final nodes of the tree represent the prices on the expiration date. Given the strike price of $K$, the price one can call or put before the contract expires, i.e., the \highlight{exercise value} of each node will be $\max(X_{node}-K, 0)$ for a call option and $\max(K-X_{node}, 0)$ for a put option.

\hide{
\begin{figure}[ht]
    \includegraphics[width=0.6\linewidth]{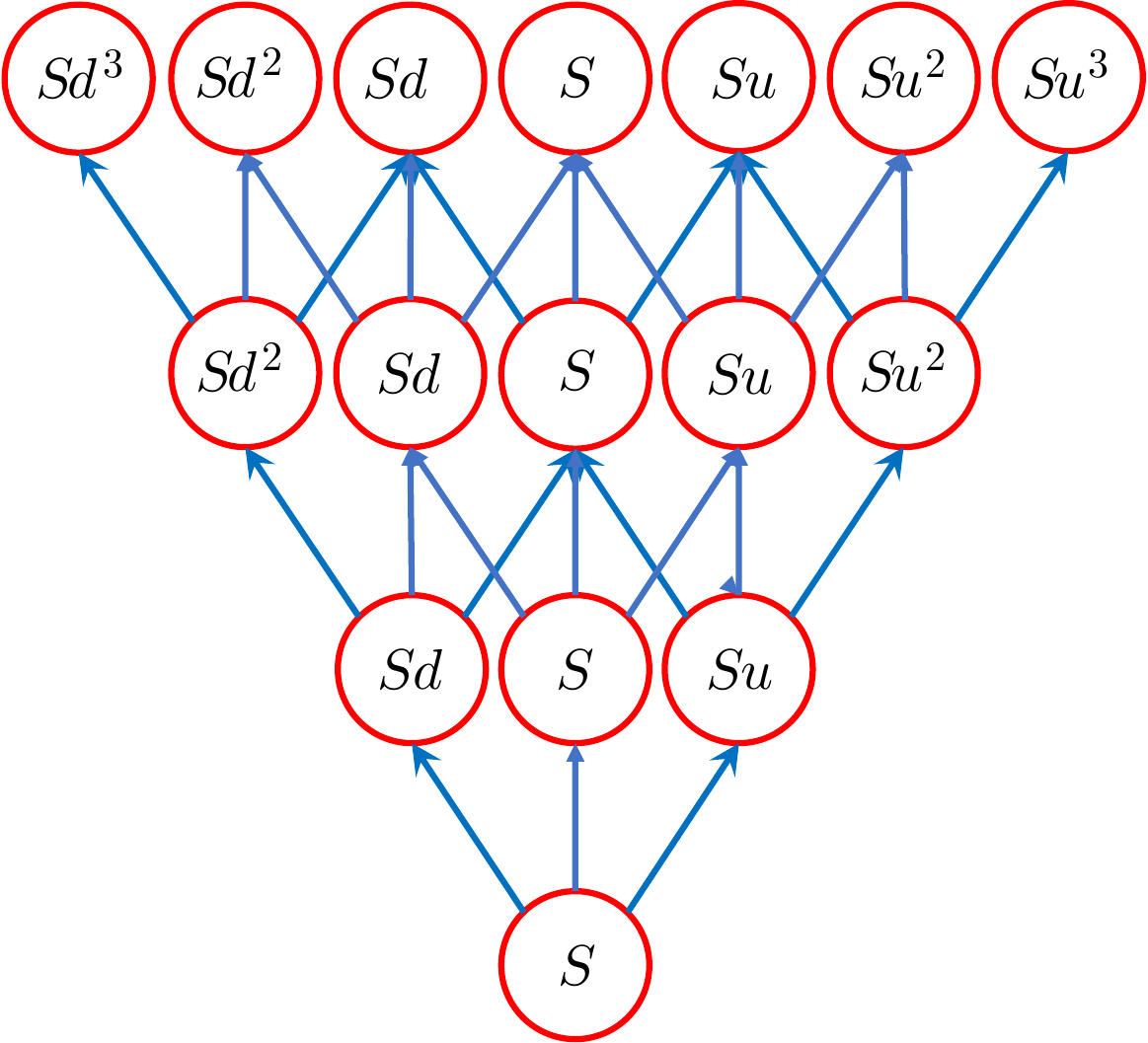}
    \caption{A 3 time-step trinomial tree.}
    \label{img:tree_in_grid}
\end{figure}
}

The risk-neutral valuation of the binomial value is performed iteratively backward. Under the assumption of risk neutrality, the value of the option today is its expected future payoff discounted at the risk-free rate of $R$. Let us number the nodes in each layer of Figure \ref{fig:trees}$(a)$ from top to bottom by successive integers starting from 1. Then the node values $X_{t, j}$ and $X_{t, j + 1}$ of a layer representing some time $t$ can be used to compute the \highlight{binomial value} of the $j$-th node in the layer representing time $t - \Delta t$ as follows: $e^{-R\Delta t}(p \cdot X_{t, j} + (1-p) \cdot X_{t, j+1})$, where, $p={(e^{R\Delta t}-d)}/{(u-d)}$ \cite{hull2003options}. Denote $m = e^{- R \cdot \Delta t }$, $s_0 = m(1-p)$, and $s_1 = mp$. Then the binomial value of that node is: $s_0 \cdot X_{t, j} + s_1 \cdot X_{t, j+1}$. For options on stocks paying a continuous dividend yield $Y$, $p={(e^{(R-Y)\Delta t}-d)}/{(u-d)}$.

The value at each node will be equal to its binomial value for European options and the larger of its binomial value and exercise value for American options.

The binomial tree of $T$ time steps can be embedded in a $(T+1) \times (T+1)$ grid. Figure \ref{fig:trees}$(b)$ shows an example. 

\begin{definition}
\label{def:v}
Let $G_{i, j}$ be the grid value in row $i \in [0, T]$ and column $j \in [0, T]$ of the $(T+1) \times (T+1)$ grid $G$. Let $G^{green}_{i, j} = S \cdot u^{2j-i}-K$, and let $G^{red}_{i, j} = s_0 G_{i+1,j} + s_1 G_{i+1,j+1}$ if $i \in [0, T)$, and $0$ otherwise. Then
%
%
$$G_{i, j} = 
        \begin{cases}
           G^{red}_{i, j},   &\text{if }G^{red}_{i, j} \ge G^{green}_{i, j}\\
           G^{green}_{i, j}, &\text{otherwise.}
       \end{cases}$$

%
\hide{
\begin{table*}[ht]
    \centering
    \begin{tabular}{ccccc}
         $G_{i, j} = 
        \begin{cases}
           G^{red}_{i, j}, &\text{if }G^{red}_{i, j} \ge G^{green}_{i, j}\\
           G^{green}_{i, j}, &\text{otherwise.}
       \end{cases}$ & where & $
        \begin{aligned}
        &G^{red}_{i, j} = 
            \begin{cases}
               0, &\text{if }i=T\\
               s_0 G_{i+1,j} + s_1 G_{i+1,j+1}, &0 \le i < T
           \end{cases}
        \end{aligned}$
        & and &
        $G^{green}_{i, j} = S \cdot u^{2j-i}-K$
    \end{tabular}
\end{table*}
}
%
\end{definition}

We say that cell $(i, j)$ of $G$ is \highlight{red} provided $G_{i, j} = G^{red}_{i, j}$, and \highlight{green} otherwise. We show in Section \ref{ssec:BOPM-boundary} that all red cells in $G$ form a single contiguous region and all green cells form another. A single boundary divides the two regions. We analyze the properties of this \highlight{red-green divider} in Section \ref{ssec:BOPM} which we will exploit to design an efficient algorithm for American call options in Section \ref{ssec:BOPM-algo}. 

\subsection{Properties of the Red-Green Divider}
\label{ssec:BOPM-boundary}
As shown in the example in Figure \ref{fig:trees}$(b)$, we assume that a binomial tree for $T$ time steps is embedded in a $(T + 1) \times (T + 1)$ grid $G$ with the root at the bottom-left corner $G[0, 0]$ and the leaves in the top row $G[T, 0 .. T]$. 
For $0 \leq j \leq i \leq T$, the two children of the binomial tree node at $G[i, j]$ are stored at $G[i + 1, j]$ and $G[i + 1, j + 1]$. The arrow from $G[i, j]$ to $G[i + 1, j + 1]$ represents a price change factor of $u$ while the one from $G[i, j]$ to $G[i + 1, j]$ represents a price change factor of $d = 1/u$. So, the entire tree occupies only the upper-left triangular part of the grid. 

\hide{
The following theorem of inequality is very useful in our proof of boundary.

\begin{theorem}
    \label{thm:leqInMax}
    Let $a_1 \leq a_2 $ and $b_1 \leq b_2$, we must have:
    \begin{align*}
        \max(a_1,b_1) \leq \max(a_2,b_2)
    \end{align*}
\end{theorem}
}


\hide{
Lemma \ref{lma:VNonDecAtRow} below shows that cell values in $G$ are nondecreasing from left to right in every row.


\begin{lemma} \label{lma:VNonDecAtRow}
    $G_{i, j-1} \le G_{i, j}$ for any $0 \le i \le T$ and $1\le j \le i$
\end{lemma}

\begin{proof}
See Appendix Lemma \ref{lma:VNonDecAtRow2}.
\end{proof}

Lemma \ref{lma:diffRowwise} below shows a lower bound on the amount of increase in cell values as one goes from left to right along a row in $G$. 

\begin{lemma} \label{lma:diffRowwise}
    $\Delta^{row}_{i,j} =G_{i, j+1}-G_{i, j}\ge \left(u^2-1\right) G_{i, j}$ for $i \in [0, T]$, $j \in [0, i - 1]$, and $G_{i,j} > 0$.
\end{lemma}
\begin{proof}
See Appendix Lemma \ref{lma:diffRowwise2}.
\end{proof}
}

Lemma \ref{lma:rightOfGreenIsGreen} shows that within the upper-left triangle of $G$, if a cell is green then the cell to its right is also green.

\begin{lemma}
    \label{lma:rightOfGreenIsGreen}
    $\left( G_{i,j+1} = G^{green}_{i,j+1} \right) \implies $ for $i \in [0, T]$ and $j \in [0, i - 1]$, $\left( G_{i,j} = G^{green}_{i,j} \right)$.
\end{lemma}

\begin{proof}
    
    Observe that $\frac{s_0}{u} + s_1 u 
    = mp \left(u - \frac{1}{u}\right) + \frac{m}{u} = e^{-Y\Delta t}$. 
    
    Let $\delta_{i,j} = G_{i,j}^{red} - G_{i,j}^{green}$, $\Delta_{i,j} = \delta_{i,j+1} - \delta_{i,j}$, $\tilde{\delta}_{i,j} = G_{i,j} - G_{i,j}^{green}$ and $\tilde{\Delta}_{i,j} = \tilde{\delta}_{i,j+1} - \tilde{\delta}_{i,j}$. 
    
    We will use mathematical induction to show that $\Delta_{i,j}\leq 0$ for all $0 \leq i \leq T$ and $j < i$. 

    As $\delta_{T, j} = K - S u^{2j - T}$ and thus $\Delta_{T,j} = S u^{2j - T}\left( 1 - u^2 \right) \leq 0$, the claim holds for $i = T$.
    
    Now suppose that the claim holds for some given $i + 1 \leq T$, i.e., $\Delta_{i+1,j} \leq 0$ for $0 \leq j < i+1 \leq T$. Since $\Delta_{i+1,j} \leq 0$, there exists a $j_{i+1}$ such that $G_{i + 1,j} = G_{i + 1,j}^{green}$ when $j > j_{i+1}$ and $G_{i + 1,j} = G_{i + 1,j}^{red}$ when $ j \leq j_i$, where $j_{i+1}$ is the largest $ j$ such that $\delta_{i+1,j} > 0$. Then,
    \begin{itemize}
     \setlength{\itemindent}{0.1cm}
     \item[$-$] for $j > j_{i+1}$, $\tilde{\Delta}_{i+1,j} = \tilde{\delta}_{i+1,j+1} - \tilde{\delta}_{i+1,j} 
         = 0 - 0 = 0$;
    \item[$-$] for $j  = j_{i+1}$, $\tilde{\Delta}_{i+1,j} = \tilde{\delta}_{i+1,j+1} - \tilde{\delta}_{i+1,j} 
         = -\tilde{\delta}_{i+1,j} \leq 0$; and
    \item[$-$] for $j < j_{i+1}$, $\tilde{\Delta}_{i+1,j} = \Delta_{i+1,j} \leq 0$.    
    \end{itemize}
    Thus, $\tilde{\Delta}_{i+1,j}  \leq 0$ for all $j \in [0, i + 1)$. 

    \noindent
    Hence, $\Delta_{i,j} = \delta_{i,j+1} - \delta_{i,j}$

    \vspace{-0.3cm}
    \begin{align*}
       &= \sum_{k \in \left\{ 0, 1\right\}}{s_k \left( \left(G_{i+1,j+k+1}-G_{i+1,j+k+1}^{green}\right) - \left(G_{i+1,j+k}-G_{i+1,j+k}^{green}\right)\right)}\\
       &\phantom{=} + \sum_{k \in \left\{ 0, 1\right\}}{s_k \left(G_{i+1,j+k+1}^{green}-G_{i+1,j+k}^{green}\right)} + G_{i,j}^{green}  - G_{i,j+1}^{green}\\
\hide{      \phantom{\Delta_{i,j}} &= s_0 \left(G_{i+1,j+1}-G_{i+1,j+1}^{green}\right) - s_0\left(G_{i+1,j} - G_{i+1,j}^{green}\right)  \\
       &\phantom{=} + s_1 \left(G_{i+1,j+2} - G_{i+1,j+2}^{green}\right) 
        - s_1\left(G_{i+1,j+1} - G_{i+1,j+1}^{green}\right)\\
       &\phantom{=} + s_0\left(G_{i+1,j+1}^{green} -G_{i+1,j}^{green}\right) 
        + s_1\left(G_{i+1,j+2}^{green} -G_{i+1,j+1}^{green}\right) + G_{i,j}^{green}  - G_{i,j+1}^{green}\\
}
       &= s_0 \tilde{\Delta}_{i+1,j} + s_1 \tilde{\Delta}_{i+1,j+1}  
       + S u^{2j-i}\left(e^{-Y \Delta t}-1\right)\left(u^2-1\right) \leq 0        
    \end{align*}
    \vspace{-0.3cm}
    
    Therefore, $\Delta_{i,j} \leq 0$ for all $0 \leq i < T$ and $j < i$. 
    
    Because $\Delta_{i,j}  \leq 0$, there exists a $j_i$ such that when $j > j_i$, $G_{i,j} = G_{i,j}^{green}$ and when $ j \leq j_i$, $G_{i,j} = G_{i,j}^{red}$, where $j_i$ is the largest $j$ such that $\delta_{i,j} > 0$. Now if $G_{i,j} = G_{i,j}^{green}$, we have $j > j_i$, and thus $j + 1 > j_i$, which implies that $G_{i,j+1} = G_{i,j+1}^{green}$.
\end{proof}

\begin{lemma}
    \label{lma:rangeOfGreen}
    $\left( G_{i, j}=G^{green}_{i, j} \right) \hspace{-0.1cm}\implies\hspace{-0.1cm}\Big(j \ge \frac{i}{2} + \frac{1}{2V\sqrt{\Delta t}}\ln{\Big(\frac{K}{S}\frac{1-e^{-R\Delta t}}{1-e^{-Y\Delta t}}\Big)}\Big)$
\end{lemma}




\begin{proof}
    By $G_{i, j} = G^{green}_{i, j}$ and Def. \ref{def:v}, we know $S\cdot u^{2j-i} - K > s_0 G_{i+1,  j} + s_1 G_{i+1, j+1}$, $G_{i+1, j} \ge S\cdot u^{2j-(i+1)} - K$, and $G_{i+1, j+1} \ge S\cdot u^{2(j+1)-(i+1)} - K\label{eq7rangeOfGreenInAppendix}$.

\hide{
\begin{align}
    G_{i+1, j} \ge S\cdot u^{2j-(i+1)} - K\label{eq7hidden}
\end{align}

\begin{align}
    G_{i+1, j+1} \ge S\cdot u^{2(j+1)-(i+1)} - K\label{eq8hidden}
\end{align}
}

Denote $z=S\cdot u^{2j-i}$ and use above inequalities: 
%
\begin{align*}
    &z - K \ge s_0 (S\cdot u^{2j-(i+1)} - K) + s_1 (S\cdot u^{2(j+1)-(i+1)} - K) \\
    &= m (S\cdot u^{2j-(i+1)} - K) + mp ((S\cdot u^{2(j+1)-(i+1)} - K) \\
&\quad - (S\cdot u^{2j-(i+1)} - K)) \\
    &= (1/u) \left(mz - mKu + mp\left(u^2 - 1\right)z\right)
\end{align*}
\begin{align*}
    &= \frac{1}{u}\left(mz - mKu + m\left({\left(e^{(R-Y)\Delta t}-\frac{1}{u}\right)}/{\left(u-\frac{1}{u}\right)}\right)(u^2 - 1)z\right)\\
    &= \frac{1}{u}\left(- mKu + u \cdot e^{-Y\Delta t}z\right) = - e^{-R\Delta t} K + e^{-Y\Delta t}
\end{align*}
Then we have: $z \ge K\frac{1- e^{-R\Delta t}}{1- e^{-Y\Delta t}}
    \implies  u^{2j - i} \ge \frac{K}{S}\frac{1- e^{-R\Delta t}}{1- e^{-Y\Delta t}}$
\begin{align*}
    \implies & j \ge \frac{i}{2} + \frac{1}{2V\sqrt{\Delta t}}\left(\ln{\frac{K}{S}}+\ln{\frac{1-e^{-R\Delta t}}{1-e^{-Y\Delta t}}}\right) 
\end{align*}
So, $\left( G_{i, j} = G^{green}_{i, j} \right) \implies \left( j \ge \frac{i}{2} + \frac{1}{2V\sqrt{\Delta t}}\ln{\left(\frac{K}{S} \frac{1-e^{-R\Delta t}}{1-e^{-Y\Delta t}}\right)} \right)$
\end{proof}

Lemma \ref{lma:bottomOfGreenIsGreen} shows that within the upper-left triangular area of $G$ if a cell is green then the cell below it must also be green, and Lemma \ref{lma:diagLeftOfRedIsRed} shows that if a cell is red then the cell diagonally left below it must also be red.

\begin{lemma}
    \label{lma:bottomOfGreenIsGreen}
    $\left( G_{i+1,j} = G^{green}_{i+1,j} \right) \implies \left( G_{i,j} = G^{green}_{i,j} \right)$ for $i \in [0, T - 1]$ and $j \in [0, i]$.
\end{lemma}

\begin{proof}
$\left( G_{i+1, j} = G^{green}_{i+1, j} \right) \implies \left( G_{i+1, j+1} = G^{green}_{i+1, j+1} \right)$ by Lemma \ref{lma:rightOfGreenIsGreen}. Let $z = S u^{2j-i}$, then $G_{i+1,j} = \frac{z}{u} - K$ and $G_{i+1,j+1} = uz - K$.
\vspace{-.5cm}
\begin{align*}
    &G^{red}_{i, j} = s_0 G_{i+1,j} + s_1 G_{i+1,j+1} = m(1-p)\left(\frac{z}{u} -K\right) + mp(uz - K)\\
     &= - S u^{2j-i} \left(1- e^{-Y\Delta t}\right) + Su^{2j-i} -K e^{-R\Delta t } \\
    &\leq -S u^{ 1 +\frac{1}{V\sqrt{\Delta t}}\left(\ln{\frac{K}{S}}+\ln{\frac{1-e^{-R\Delta t}}{1-e^{-Y\Delta t}}}\right)} \left(1-e^{-Y\Delta t}\right) + Su^{2j-i} -K e^{-R\Delta t } \\
    &\quad\quad \mbox{(by Lemma \ref{lma:rangeOfGreen} at $G_{i+1, j}$)} \\
    &= -uK\left(1-e^{-R\Delta t}\right) + Su^{2j-i} -Ke^{-R\Delta t} \quad (\textrm{since~} u^{\frac{1}{V\sqrt{\Delta t}}}=e) \\
    &\le -K\left(1-e^{-R\Delta t}\right) + Su^{2j-i} -Ke^{-R\Delta t}
    =  Su^{2j-i} - K 
    = G^{green}_{i, j} 
\end{align*}
Hence, the statement holds.
\end{proof}

\hide{
\begin{corollary}
    \label{crl:noHole} 
    Any rectangular area in $G$ with a green cell at the top left corner does not have any red cells in it.
\end{corollary}
\begin{proof}
 Follows from Lemma \ref{lma:rightOfGreenIsGreen} and Lemma \ref{lma:bottomOfGreenIsGreen}.
\end{proof}
}

\begin{lemma}\label{timedecay}
    For any $i \in [0, T - 2]$ and $j < i$: $G_{i,j} \geq G_{i+2,j+1}$.
\end{lemma}

\begin{proof}
    We use mathematical induction. For $i = T - 2$, we have our base case: $G_{i,j} \geq \max(0,Su^{2j-i}- K) = G_{T,j+1}$.
\hide{    
    \begin{align*}
        G_{i,j} \geq \max(0,Su^{2j-i}- K) = G_{T,j+1}
    \end{align*}
}  

Suppose it holds true for all $G_{i+1,j}$ for some given $i \geq T - 3$ and $j \in [0, i]$. Then
$G_{i,j} = \max(s_0 G_{i+1,j} + s_1 G_{i+1,j+1},$ $Su^{2j-i} -K) \geq \max(s_0 G_{i+3,j+1} + s_1 G_{i+3,j+2},Su^{2j-i} -K) = G_{i+2,j+1}$. \end{proof}


\begin{lemma} \label{lma:diagLeftOfRedIsRed}
    $\left(G_{i+1,j+1} = G_{i+1,j+1}^{red}\right) 
    \Rightarrow \left(G_{i,j} = G_{i,j}^{red}\right)$ for $i \in [0, T - 1]$ and $j \in [0, i - 1]$.
\end{lemma}
\vspace{-0.2cm}

\begin{proof}
    $G_{i+1,j+1} = G_{i+1,j+1}^{red}$ and Lemma \ref{lma:bottomOfGreenIsGreen} gives: $G_{i+2,j+1} = G_{i+2,j+1}^{red} \geq G_{i+2,j+1}^{green} = Su^{2j-i} - K$. Now, by Lemma \ref{timedecay}, $G_{i,j} \geq G_{i+2,j+1} \geq Su^{2j-i} - K = G_{i,j}^{green}$. So, $G_{i,j} = G_{i,j}^{red}$.
\end{proof}

The following corollary says that at every time step all red cells appear to the left of all green cells, and the boundary between the green and the red regions either remains the same or moves by one cell towards the left every time step.

\begin{corollary}
    \label{crl:movingBoundary} 
    For every $i \in [0, T - 1]$, there exists an index $j_i \in [0, i]$ such that all cells $G_{i,j}$ with $0 \leq j \leq j_i$ are red and all (possibly zero) cells $G_{i,j}$ with $j_i < j \leq i$ are green. Also, for $i \in [0, T - 2]$, $j_{i+1} - 1 \leq j_i \leq j_{i+1}$.
\end{corollary}
\vspace{-0.2cm}

\begin{proof}
 Follows from Lemmas \ref{lma:rightOfGreenIsGreen}, \ref{lma:bottomOfGreenIsGreen}, and \ref{lma:diagLeftOfRedIsRed}.
\end{proof}

\subsection{Algorithm for American call option pricing under BOPM}
\label{ssec:BOPM-algo}
\vspace{-0.2cm}

\begin{figure}[t]
    \centering
    \begin{subfigure}[t]{0.45\linewidth}
        \centering
        \includegraphics[width=0.8\linewidth]{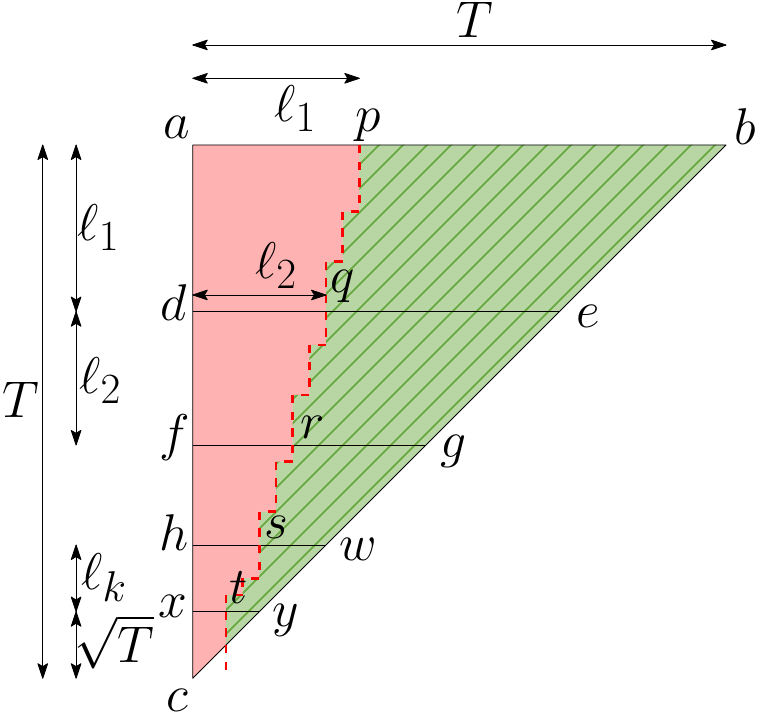}
        \caption{Partitioning the solution space into trapezoids}
    \label{img:trapezoid}
    \end{subfigure}
    \hfill
    \begin{subfigure}[t]{0.5\linewidth}
        \centering
        \includegraphics[width=0.9\linewidth]{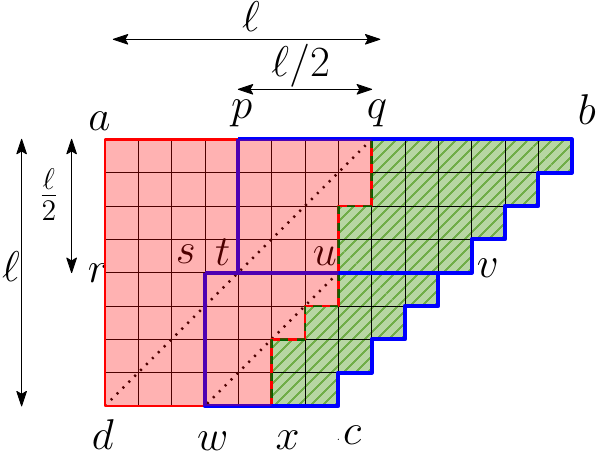}
        \caption{Decomposing a trapezoid into smaller trapezoids}
        \label{img:aoption}
    \end{subfigure}
    \vspace{-0.2cm}    
    \caption{American call option pricing under BOPM}
    \vspace{-0.5cm}
\end{figure}

The solution space is a right-angle isosceles triangle with base length $T$. We know the boundary between the red and green cells in the first row of the triangle (solution space); however, we do not know the locus of the boundary in the subsequent rows of the triangle. We compute the boundary in the following process. 

We partition the triangle (solution space) into trapezoids (see Figure \ref{img:trapezoid}). We compute the first trapezoid with its first row as the same first row of the triangle (the solution space) and solve this newly created trapezoid (we explain how we create a trapezoid and solve it later in this section). Then we compute the second trapezoid with its first row as the last row of the first trapezoid and solve the second trapezoid. This process continues until we are left with a right-angle isosceles triangle with base size at most $\sqrt{T}$. We solve this triangle iteratively by doing quadratic work in time $\Oh{T}$. We describe the process in detail below.

\para{Partitioning the Triangle into Trapezoids.}
 Let $abc$ be a right-angle isosceles triangle with base length $T$ (see Figure~\ref{img:trapezoid}). 
 Let $\ell_1$ be the number of red cells on the line segment $ab$. 
 They will appear consecutively from $a$ to some point $p$. Let $d$ be the point in the line segment $ac$ that is $\ell_1$ distance from $a$. Draw a horizontal line from $d$; let the line intersect $bc$ at point $e$. Therefore, we get a trapezoid $abed$ with height $\ell_1$ with $\ell_1$ red cells in its first row. 
 
We solve trapezoid $abed$, which means that we compute the values of all red cells in its last row and thus find the boundary point $q$ between red and green cells in $de$. Let $|dq| = \ell_2$. Let $f$ be the point on $dc$ that $|df| = \ell_2$. We draw a horizontal line $fg$, and get the second trapezoid $degf$ of height $\ell_2$. The last row of the trapezoid $abed$ becomes the first row of $degf$. 

We solve the trapezoid $degf$ and all subsequent trapezoids created following the approach described above until we are left with a right-angle isosceles triangle $xyc$ with base length at most $\sqrt{T}$. We compute all cell values of $xyc$ iteratively in $\Oh{T}$ time. 

\para{Solving a Trapezoid.} Solving a trapezoid means given all red cell values in its first row computing all red cell values in its last row. Let $abcd$ be a trapezoid of height $\ell$ with $\ell$ red cells from $a$ to $q$ in its first row (Figure~\ref{img:aoption}). Let $r$ be the point on $ad$ such that $|ar| = \floor{\ell/2}$. Draw a horizontal line from $r$ intersecting $bc$ at $v$. 
To compute the values of red cells on $dc$, we $(1)$ compute all red cells on $rv$, and $(2)$ using the cell values on $rv$, compute all red cells on $dc$. 

\vspace{0.05cm}
\subparagraph*{(1) \underline{\textit{Computing all red cells on $\pmb{rv}$}}.}
Let $rv$ and $qd$ intersect $t$. We compute the cell values on line segment $rt$ using the FFT-based stencil algorithm of \cite{ahmad2021fast} which are guaranteed to be red because of Corollary \ref{crl:movingBoundary}. Note that $|rt| = |pt| = \floor{\ell/2}$. The newly created trapezoid $pbvt$ has height $\floor{\ell/2}$, and there are $\floor{\ell/2}$ red cells in its first row. We solve trapezoid $pbvt$ recursively similar to trapezoid $abcd$, and thus compute all red cell values on $tv$. 

\vspace{0.05cm}
\subparagraph*{(2) \underline{\textit{Computing all red cells on $\pmb{dc}$ using the red cells on $\pmb{rv}$}}.}
Let point $u$ be on the boundary between red and green cells on $rv$. Draw a line from $u$ parallel to $qd$ intersecting $dc$ at $w$. Next, draw a line through $w$ perpendicular to $dc$ intersecting $rv$ at point $s$. 
We compute all red cells on $dc$ in exactly the same way we computed the red cells on $rv$ above. We first use the FFT-based stencil algorithm of \cite{ahmad2021fast} to find all cell values on $dw$, and then recursively solve trapezoid $svcw$ of height 
$\ceil{\ell/2}$ to find all red cell values on $wc$.

\vspace{0.05cm}
 \subparagraph*{\underline{\textit{Solving trapezoids of height $\pmb{\Oh{1}}$ (base case)}}.}
 We solve trapezoids of height $\Oh{1}$ in $\Oh{1}$ work and span using the na\"ive looping code 
 (e.g., see the pseudocode in Figure \ref{fig:notations-and-BOPM-loops}). 

 The following theorem gives the work and span of our algorithm.

\vspace{-0.2cm}
\begin{theorem}\label{thm:AOT2}
    Our algorithm solves the American call option pricing problem under BOPM in $\Oh{T\log^2 T}$ work and $\Oh{ T }$ span, where $T$ is the number of time steps.
\end{theorem}
\vspace{-0.2cm}
%

\vspace{-0.48cm}
\begin{proof}
Let $\zeta_1(\ell)$ and $\zeta_\infty(\ell)$ be the work and span, respectively, for solving a trapezoid of height $\ell$ recursively. We call the $\Oh{\ell \log \ell}$-work FFT-based periodic algorithm \cite{ahmad2021fast} twice and solve two smaller trapezoids of height $\ell/2$ each recursively. Hence, $\zeta_1(\ell) = 2\zeta_1 (\ceil{\ell/2}) + \Theta(\ell \log \ell) = \Oh{\ell \log^2 \ell}$. Observe that although the two smaller trapezoids must be solved one after the other (e.g., $svcw$ after $pbvt$ in Figure \ref{img:aoption}), each of them can be solved in parallel with the FFT-based algorithm (of span $\Oh{\log \ell \log \log \ell}$ \cite{ahmad2021fast}) called to find the red cells on its left (on line segments $rt$ and $dw$, respectively). Hence, $\zeta_{\infty}(\ell) = 2 \zeta_{\infty} (\ceil{\ell/2}) + \Oh{\log \ell \log  \log \ell} = \Oh{\ell}$.

Suppose that we solve $k$ trapezoids of heights $\ell_1, \ell_2, \ldots, \ell_k$, respectively, using the above process as shown in Figure \ref{img:trapezoid}. Let $\Psi_1$ and $\Psi_\infty$ be the total work and span, respectively, for solving those $k$ trapezoids followed by the time needed to solve the leftover triangle of height $\Oh{\sqrt{T}}$. Then $\Psi_1 = \left( \sum_{1 \leq i \leq k}{\Oh{\ell_i \log^2 \ell_i}} \right) + \Oh{T} = \Oh{ T\log^2 T }$. Since those trapezoids and the triangle are solved in sequence, we have $\Psi_{\infty} = \sum_{1 \leq i \leq k}{\Oh{ \ell_i }} + \Oh{\sqrt{T}} = \Oh{T}$.
\end{proof}



%% file: 4-TOPM.tex
\section{American Call Option under the Trinomial Option Pricing Model}
\label{sec:TOPM}


The trinomial options pricing model (TOPM) encodes the possible sequences of prices for a given asset within the structure of a trinomial tree (see Figure \ref{fig:trees}$(c)$). It expands on BOPM by allowing the value of an asset to remain unchanged after a given time step. TOPM was introduced by Boyle \cite{boyle1986option}, and while it is less popular than the BOPM, it provides for the possibility of more accurate predictions than BOPM at the cost of only a constant factor blowup in runtime when using na\"ive $\Oh{T^2}$ methods. Langat, Mwaniki, and Kiprop showed that TOPM converges to the same solution as Black-Scholes with half as many time steps  \cite{LangatMwanikiKiprop}. TOPM is also equivalent to the explicit finite difference method \cite{hull2003options}.

TOPM carries over many properties from BOPM, e.g., $X_{node} = S \times u^{N_u-N_d}$ since the number of ``remain the same'' moves does not factor into the price. Here $u = e^{V\sqrt{2 \Delta t}}$, and $d = 1/u$. The exercise value in TOPM is $\text{max}(X_{node} - K,0)$.

The transition probabilities can be expressed as \\$p_u = \left( {(e^{(R-Y)\Delta t/2}-\sqrt{d})}/{(\sqrt{u} - \sqrt{d})} \right) ^2$, \\$p_d = \left( {( \sqrt{u} -e^{(R-Y)\Delta t/2} )}/{(\sqrt{u} - \sqrt{d})} \right) ^2$, and $p_{\probm} = 1 - p_u - p_d$, which are alternate forms of those given in \cite{hull2003options}.

Let $m = e^{-R\Delta t}$, $s_0 = m p_u, s_1 = m p_{\probm}, s_2 = m p_d$. 
Let $G_{i, j}$ denote the grid value in the row $i \in [0, T]$ and column $j \in [0, 2i]$ of the $(T+1) \times (2T+1)$ grid $G$. Let $G^{green}_{i, j} = S \cdot u^{j-i}-K$, and let $G^{red}_{i, j} = \sum_{k \in \{0, 1, 2\}}{ s_k G_{i+1,j + k} }$ if $i \in [0, T)$, and $0$ otherwise.
Then similar to BOPM:
%
\begin{align*}
    G_{i, j} = 
        \begin{cases}
           G^{red}_{i, j},   & \text{if }G^{red}_{i, j} \ge G^{green}_{i, j}\\
           G^{green}_{i, j}, & \text{otherwise.}
       \end{cases}
\end{align*}

\hide{
\begin{table}[h]
    \centering
    \begin{tabular}{@{}c@{}}
         $G_{i, j} = 
        \begin{cases}
           G^{red}_{i, j},   &\hspace{-0.2cm}\text{if }G^{red}_{i, j} \ge G^{green}_{i, j}\\
           G^{green}_{i, j}, &\hspace{-0.2cm}\text{otherwise.}
       \end{cases}$ where\\ \\
       $
        \begin{aligned}
        & G^{red}_{i, j} = \sum_{k \in \{0, 1, 2\}}{ s_k G_{i+1,j + k} } \text{~if $i \in [0, T)$, and $0$ if $i = T$}\\
        & G^{green}_{i, j} = S \cdot u^{j-i}-K
        \end{aligned}$
    \end{tabular}
\end{table}
}


In the supplementary material,
we show that the TOPM grid shows properties similar to the BOPM grid, that is, in every row all red cells appear first in contiguous locations followed by all green cells, and with every time step the red-green boundary moves by at most one cell to the left. This allows us to use a similar algorithm to that given for BOPM (Section \ref{ssec:BOPM-algo}) with the identical work and span.


%% file: 5-BSM.tex
\vspace{-0.3cm}
\section{American Put Option Under the Black-Scholes-Merton Model}
\label{sec:APO-BSM}

\subsection{Black-Scholes-Merton Pricing Model (BSM)}
\label{ssec:BSM}

BSM is a mathematical method to calculate the theoretical value of an option contract. The option pricing problem is transformed into a partial differential equation (PDE) with variable coefficients. An explicit formula for the price can be obtained assuming a log-normal distribution of the asset price. Note that
the limit of the discrete-time BOPM approximates the continuous-time BSM under the same assumption.
While BOPM utilizes simple statistical methods, BSM requires a solution of a stochastic differential equation.

Denote stock price at time $t$ by $\mathcal{S}(t)$. BSM claims that 

\noindent
there is a deterministic relation between the option price and the stock price and time. This means that there is a deterministic function $v(t,x)$ for option price $x$ and time $t$ such that: $X(t) = v(t,\mathcal{S}(t))$, where $X(t)$ represents the value of the option at time $t$. Now BSM derives that $v(t,x)$ satisfies the following two-area classical form: 
%

\vspace{-0.3cm}
\begin{align}\label{eq:BSM}
    v(t,x) = \left \{ \begin{array}{lr@{}}
    \frac{1}{r}\left( \begin{array}{@{}c@{}}
    \frac{\partial v}{\partial t}(t,x) + rx \frac{\partial v}{\partial x}(t,x)\\
    + \frac{1}{2} \sigma^2 x^2 \frac{\partial^2 v}{\partial x^2}(t,x)
    \end{array} \right),
    & \hspace{-0.9cm}\mbox{if $x > L(T-t)$}\\\vspace{-0.4cm}
    & \\
    K - x, & \hspace{-0.9cm}\mbox{if $0 \leq x \leq L(T-t)$}
    \end{array} \right .
\end{align}

\vspace{-0.1cm}
\noindent
where $L(0) = K$, $\frac{\partial v}{\partial x}(t,L(T-t)^+) = \frac{\partial v}{\partial x}(t,L(T-t)^-) = -1$, and $v(t,L(T-t)^+) = v(t,L(T-t)^-)$ for $0 \leq t  < T$.

It is equivalent to satisfying the following:

\vspace{-0.3cm}
\begin{align}\label{eq:BSM2}
\hspace{-0.2cm}v(t,x) \geq \left \{ \begin{array}{lr@{}}
    \hspace{-0.2cm}\frac{1}{r}\left( \begin{array}{@{}c@{}}
    \frac{\partial v}{\partial t}(t,x) + rx \frac{\partial v}{\partial x}(t,x)\\
    + \frac{1}{2} \sigma^2 x^2 \frac{\partial^2 v}{\partial x^2}(t,x)
    \end{array} \right),
    & \hspace{-0.3cm}\mbox{if $0 \leq x \leq L(T-t)$}\\\vspace{-0.4cm}
    & \\
    K - x, & \hspace{-0.3cm}\mbox{if $x > L(T-t)$}
    \end{array} \right .  
\end{align}

\vspace{-0.1cm}
\noindent
where, $t \in [0,T]$, 
Recall that at the maturity time $T$, the option price (call option case) will be $
    X(T) = \max(\mathcal{S}(T) - K,0) = (K-\mathcal{S}(T))^{+}$.

It also means that $v(T,x) = (K-x)^{+}$ on the boundary. Therefore, the goal becomes solving Equation (\ref{eq:BSM}) or Inequality (\ref{eq:BSM2}). For more details on how the BSM model is formulated or why the complementary form is equivalent to the classical form, see \cite{chen2007mathematical,shreve2004stochastic,van1976optimal}. We show the two areas are contiguous and find properties of their boundary in Section \ref{ssec:BSM-boundary}, which can be exploited by our algorithm in Section \ref{ssec:BSM-algo}.

\subsection{Properties of the Two-Area Boundary}
\label{ssec:BSM-boundary}

Note that Equation (\ref{eq:BSM}) includes dimensional variables. First, we find nondimensionalized forms of Equations (\ref{eq:BSM}) and (\ref{eq:BSM2}). 

Let $ s = \ln \frac{x}{K} $, $\tau = \frac{1}{2}\sigma^2  (T-t)$, $\tilde{v}(\tau,s) = \frac{1}{K}v(t,x) $, $\tilde{L}(\tau) = L(T-t)$ and $\omega = \frac{2r}{\sigma^2}$. Then: $\frac{\partial \tilde{v}}{\partial s} = \frac{x}{K} \frac{\partial v}{\partial x}, \frac{\partial^2 \tilde{v}}{\partial s^2} = \frac{x^2}{K} \frac{\partial^2 v}{\partial x^2} + \frac{x}{K} \frac{\partial v}{\partial x}, 
    \frac{\partial \tilde{v}}{\partial \tau} = -\frac{2}{K\sigma^2}\frac{\partial v}{\partial t}$.
%

Applying Equation (\ref{eq:BSM}), and $\tilde{v}(\tau,s)$ satisfes the following: 

\vspace{-0.2cm}
\begin{align}\label{Non-di-BSM}
      \tilde{v}(\tau,s) = \left \{ \begin{array}{lr@{}}
         \frac{1}{\omega}\left( \begin{array}{@{}c@{}}
         (\omega-1)\frac{\partial \tilde{v}}{\partial s} (\tau,s)\\ + \frac{\partial ^2\tilde{v}}{\partial s^2}(\tau,s)
         - \frac{ \partial \tilde{v}}{ \partial \tau}(\tau,s)
         \end{array} \right),
          & \mbox{if $s > \tilde{L}(\tau)$}\\
           1-e^s, & \mbox{if $s \leq  \tilde{L}(\tau)$}
    \end{array} \right .  
\end{align}

It also satisfies the dimensionless complementary form after plugging into Inequality (\ref{eq:BSM2}):

\vspace{-0.2cm}
\begin{align}\label{Non-di-BSM2}
      \tilde{v}(\tau,s) \geq \left \{ \begin{array}{lr@{}}
         \frac{1}{\omega}\left( \begin{array}{@{}c@{}}
         (\omega-1)\frac{\partial \tilde{v}}{\partial s} (\tau,s)\\ + \frac{\partial ^2\tilde{v}}{\partial s^2}(\tau,s)
         - \frac{ \partial \tilde{v}}{ \partial \tau}(\tau,s)
         \end{array} \right),
          & \mbox{if $s \leq \tilde{L}(\tau)$}\\
           1-e^s, & \mbox{if $s >  \tilde{L}(\tau)$}
    \end{array} \right .  
\end{align}

\noindent
where, $\tilde{L}(0) = 1 $ and $\tilde{v}(s,0) = \max(1- e^s,0)$. Consider an approximation $ v_k^n$ of $\tilde{v}(n\Delta \tau,k\Delta s)$ and use the following finite-difference approximations (where, $\tilde{v}$ denotes $\tilde{v} (n \Delta \tau, k \Delta s)$):

\begin{align*}
  \hspace{-0.2cm}\frac{ \partial \tilde{v}}{ \partial \tau}  \approx \frac{v_k^{n+1} - v_k^{n}}{\Delta t}, \frac{\partial \tilde{v}}{\partial s}  \approx \frac{v_{k+1}^n - v_{k-1}^{n}}{2\Delta s},
    \frac{\partial ^2\tilde{v}}{\partial s^2} \approx \frac{v_{k+1}^n - 2v_{k}^n+v_{k-1}^{n}}{(\Delta s)^2}
\end{align*}

\hide{
\begin{align}\label{dis}
    \nonumber \frac{ \partial \tilde{v}}{ \partial \tau}(n \Delta \tau, k \Delta s)   &\approx \frac{v_k^{n+1} - v_k^{n}}{\Delta t}, \frac{\partial \tilde{v}}{\partial s}(n \Delta \tau, k \Delta s)  \approx   \nonumber \frac{v_{k+1}^n - v_{k-1}^{n}}{2\Delta s} \\ 
    \frac{\partial ^2\tilde{v}}{\partial s^2}(n \Delta \tau, k \Delta s)&\approx \frac{v_{k+1}^n - 2v_{k}^n+v_{k-1}^{n}}{(\Delta s)^2}
\end{align}
}

Now plug it into Equation (\ref{Non-di-BSM}) to obtain these two regions:

\noindent
\noindent
\begin{align}\label{bsmG}
    v_k^{n+1} = \left \{ \begin{array}{lr@{}}
    \left( \begin{array}{@{}c@{}}
    \left(1- \omega \Delta \tau- 2\frac{\Delta \tau}{(\Delta s)^2}\right)v_k^{n}+\\ 
    ~\sum\limits_{h \in \{-1, 1\}}{\left( \frac{\Delta \tau}{(\Delta s)^2} +  h \frac{(\omega-1)}{2} \frac{\Delta \tau}{\Delta s}\right)v_{k+h}^{n}} \end{array}\right),
    & \hspace{-0.2cm}\mbox{if } k > \widehat{L}\\ \vspace{-.4cm}
    & \\
    1 - e^{k \Delta s}, & \hspace{-0.2cm}\mbox{if } k  \leq \widehat{L}
    \end{array} \right .
\end{align}

\hide{
\begin{align}\label{bsmG}
      v_k^{n+1} = \left \{ \begin{array}{lr@{}}
         \left( \begin{array}{@{}c@{}}
         \left(1- \omega \Delta \tau- 2\frac{\Delta \tau}{(\Delta s)^2}\right)v_k^{n}\\
         +~\left( \frac{\Delta \tau}{(\Delta s)^2} +  \frac{(\omega-1)}{2} \frac{\Delta \tau}{\Delta s}\right)v_{k+1}^{n}\\
         +~\left(\frac{\Delta \tau}{(\Delta s)^2} -  \frac{(\omega-1)}{2} \frac{\Delta \tau}{\Delta s}\right)v_{k-1}^{n}
         \end{array} \right),
          & \mbox{if } k > \frac{\tilde{L}((n+1)\Delta \tau)}{\Delta s}\\
           1 - e^{k \Delta s}, & \mbox{if } k  \leq \frac{\tilde{L}((n+1)\Delta \tau)}{\Delta s}
    \end{array} \right .  
\end{align}
}

\noindent
where, $\widehat{L} = \frac{\tilde{L}((n+1)\Delta \tau)}{\Delta s}$ and $v_k^0 = \max(1-e^{k\Delta s},0)$ for all integer $k$. It also satisfies the following condition by discretizing Equation (\ref{Non-di-BSM2}):

\vspace{-.5cm}
\begin{align}\label{bsmG2}
      v_k^{n+1} \geq \left \{ \begin{array}{lr@{}}
      \left( \begin{array}{@{}c@{}}
          \left(1- \omega \Delta \tau- \frac{2\Delta \tau}{(\Delta s)^2}\right)v_k^{n}+\\
         \sum\limits_{h \in \{-1, 1\}}{\left( \frac{\Delta \tau}{(\Delta s)^2} +  \frac{h(\omega-1)}{2} \frac{\Delta \tau}{\Delta s}\right)v_{k+h}^{n}}
         \end{array} \right),
          &  \hspace{-0.3cm}\mbox{if } k \leq \widehat{L}\\
           1 - e^{k \Delta s}, & \hspace{-0.3cm}\mbox{if } k  > \widehat{L}
    \end{array} \right .  
\end{align}

\hide{
\begin{align}\label{bsmG2}
      v_k^{n+1} \geq \left \{ \begin{array}{lr@{}}
         \left( \begin{array}{@{}c@{}}
         \left(1- \omega \Delta \tau- 2\frac{\Delta \tau}{(\Delta s)^2}\right)v_k^{n}+\\
         ~\left( \frac{\Delta \tau}{(\Delta s)^2} + \frac{(\omega-1)}{2} \frac{\Delta \tau}{\Delta s}\right)v_{k+1}^{n}\\
         +~\left(\frac{\Delta \tau}{(\Delta s)^2} - \frac{(\omega-1)}{2} \frac{\Delta \tau}{\Delta s}\right)v_{k-1}^{n}
         \end{array} \right),
          & \hspace{-1cm}\mbox{if } k \leq \frac{\tilde{L}((n+1)\Delta \tau)}{\Delta s}\\
          & \\
           1 - e^{k \Delta s}, & \hspace{-1cm}\mbox{if } k  > \frac{\tilde{L}((n+1)\Delta \tau)}{\Delta s}
    \end{array} \right .  
\end{align}
}

Similar to the BOPM for American option, we define the green zone and the red zone:
\begin{definition}
We call that $v_k^n$ is in the green zone when $k \Delta s \leq \tilde{L}(n\Delta \tau )$, otherwise it is in the red zone.
\end{definition}
To apply Equations (\ref{bsmG})--(\ref{bsmG2}), we need to determine the form of $\tilde{L}((n+1)\Delta \tau)$ which has no analytical form, although it has asymptotic results \cite{alobaidi2001asymptotic} or approximation results
\cite{zhu2007calculating,zhu2006new}. Instead we can use the following theorem from \cite{chen2007mathematical}:
\begin{theorem}[\cite{chen2007mathematical}]\label{T58}
The early exercise boundary curve $\tilde{L}(\tau)$ is monotonically decreasing.
\end{theorem}

\begin{theorem}\label{thm:BSM-boundary}
    Let  
    $a =  \frac{\Delta \tau}{(\Delta s)^2} + (\omega-1) \frac{\Delta \tau}{\Delta s}$,    
    $b =  \frac{\Delta \tau}{(\Delta s)^2} - (\omega-1) \frac{\Delta \tau}{\Delta s}$, 
    $c = 1 - a - b -\omega \Delta \tau$, and $k_n$ be the largest integer such that $k_n \leq \frac{\tilde{L}(n\Delta \tau)}{\Delta s}$. Then $0 \leq k_n - k_{n+1} \leq 1$ when $a,b,c \geq 0$.
\end{theorem}
\vspace{-0.2cm}

\hide{\begin{proof}
    We give the full proof in the supplementary material.\end{proof}
}
\begin{proof}
    We first prove that $k_n - k_{n+1} \geq 0$. Suppose that this is not true. Then we will have:
        $\tilde{L}((n+1)\Delta \tau) \geq k_{n+1} \geq k_n + 1 > \tilde{L}(n \Delta \tau )$,
    which contradicts Theorem \ref{T58}.
    
    Now we prove that $k_n - k_{n+1}\leq 1 $. 
   Because $k_n$ is in the green zone, we will have the following:
    \begin{align*}
        & v_{k_n}^{n} = 1 - e^{k_n \Delta s} 
        \geq (1- \omega \Delta \tau- a-b)v_{k_n}^{n-1} + a v_{k_n+1}^{n-1} + b v_{k_n - 1}^{n-1} \\
        & \geq (1- \omega \Delta \tau- a-b)\left(1 - e^{k_n \Delta s}\right) + a \left(1 - e^{(k_n+1) \Delta s}\right) 
         \\&+ b \left(1 - e^{(k_n-1) \Delta s}\right)\\
        & \Rightarrow \omega \Delta \tau\left(1-e^{k_n \Delta s}\right) + e^{k_n \Delta s}\left(a\left(e^{\Delta s}-1\right) + b \left(e^{-\Delta s}-1\right) \right) \geq 0\\        
         & \Rightarrow 1 - e^{(k_n-1) \Delta s} 
         > (1- \omega \Delta \tau- a-b)\left(1 - e^{(k_n-1) \Delta s}\right)  \\
         &\quad + a \left(1 - e^{(k_n) \Delta s}\right) + b \left(1 - e^{(k_n-2) \Delta s}\right)
    \end{align*}
\hide{    
   It leads to: $\omega \Delta \tau\left(1-e^{k_n \Delta s}\right) + e^{k_n \Delta s}\left(a\left(e^{\Delta s}-1\right) + b \left(e^{-\Delta s}-1\right) \right) \geq 0$
    \begin{align}
         \nonumber & \Rightarrow 1 - e^{(k_n-1) \Delta s} 
         > (1- \omega \Delta \tau- a-b)\left(1 - e^{(k_n-1) \Delta s}\right)  \\
         &\quad + a \left(1 - e^{(k_n) \Delta s}\right) + b \left(1 - e^{(k_n-2) \Delta s}\right)\nonumber \label{eq16}
    \end{align}
}    
     Considering $v_{k_n-1}^{n+1}$, we first observe that:
    \begin{align*}
        v_{k_n-2}^{n} = 1-e^{(k_n-2)\Delta s}, v_{k_n-1}^{n} = 1-e^{(k_n-1)\Delta s}, v_{k_n}^{n} = 1-e^{(k_n)\Delta s}
    \end{align*}
    Now we show that $v_{k_{n}-1}^{n+1}$ is in the green zone. Suppose that it is in the red zone. Then: 
    \begin{align*}
        v_{k_{n}-1}^{n+1} &= (1- \omega \Delta \tau- a-b)\left(1 - e^{(k_n-1) \Delta s}\right) + a \left(1 - e^{(k_n) \Delta s}\right) \\
        &\quad + b \left(1 - e^{(k_n-2) \Delta s}\right) 
        < 1 - e^{(k_n-1) \Delta s} 
    \end{align*}
    which leads to a contradiction because $v_{k_{n}-1}^{n+1} $ should be $\ge 1 - e^{(k_n-1) \Delta s} $. By the definition of $k_{n+1}$, we must have $k_{n+1} \geq k_n - 1$, completing the proof of the theorem.
\end{proof}

\subsection{Algorithm for American put option pricing under BSM}
\label{ssec:BSM-algo}

\begin{figure}[t]
    \centering
    \begin{subfigure}[t]{0.45\linewidth}
        \centering
        \includegraphics[width=\linewidth]{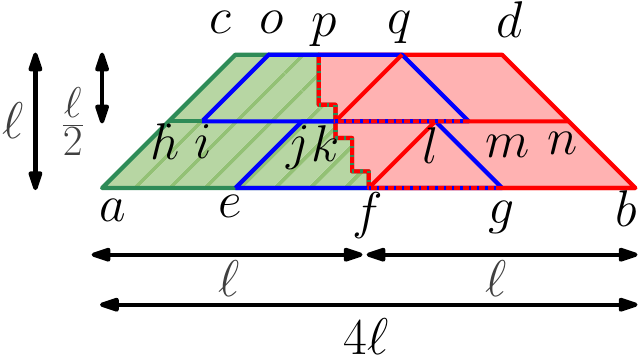}
        \caption{Decomposing trapezoid $abcd$ into smaller trapezoids}
        \label{img:trapezoidBSM}
    \end{subfigure}
    \hfill
    \begin{subfigure}[t]{0.45\linewidth}
        \centering
        \includegraphics[width=1\linewidth]{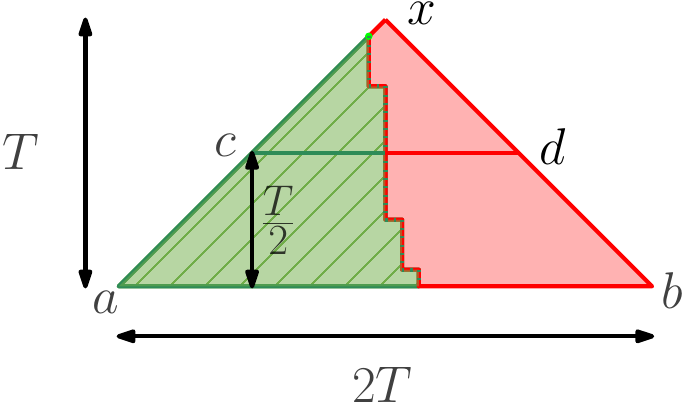}
        \caption{Partitioning the solution space into trapezoids}
        \label{img:BSMoption}
    \end{subfigure}
    \vspace{-.3cm}
    \caption{American put option pricing under BSM}
    \vspace{-.5cm}

\end{figure}

Our algorithm for the American put option under BSM is similar to our algorithm for the American call option under BOPM as described in Section \ref{ssec:ACO-BOPM}.

Observe that we will have to perform a nonlinear stencil computation based on the update equation (\ref{bsmG}). For $T$ time steps we use a $T \times 2T$ space-time grid with the time dimension being $T$ and spatial dimension $2T$. According to Equation (\ref{bsmG}), we compute a cell $v_k^{n+1}$ of that grid, where $n + 1$ represents the time coordinate and $k$ the spatial coordinate, from cells $v_{k - 1}^{n}$, $v_{k}^{n}$, and $v_{k + 1}^{n}$ using a 3-point stencil provided $k > \frac{\tilde{L}((n+1)\Delta \tau)}{\Delta s}$. Otherwise, we set it to $1 - e^{k \Delta s}$. In the first case, cell $v_k^{n + 1}$ will be in the red zone, and in the second case it will be in the green zone. As explained in Section \ref{ssec:BSM-boundary}, the entire boundary between these two zones in not known ahead of time, but it moves by at most one cell toward the green region with every time step. The goal of the algorithm is to compute the value of the central cell of the spatial dimension at time step $T$ (e.g., apex $x$ of the isosceles triangle $abx$ in Figure \ref{img:BSMoption}).

We solve the problem by decomposing the isosceles triangle $abx$ into a sequence of the isosceles trapezoids of geometrically decreasing heights (see Figure \ref{img:BSMoption}) and solving (i.e., find the cell values of top base given those of the bottom base of the trapezoid) them one by one from bottom to top until we reach a leftover triangle of small constant size which we solve na\"ively to find the value of $x$. We solve an isosceles trapezoid recursively by decomposing it into two smaller trapezoids of smaller height and solving them recursively and also using the FFT-based algorithm from \cite{ahmad2021fast} solve two subtrapezoids that are entirely composed of red cells (see Figure \ref{img:trapezoidBSM}). Details of this algorithm are given below.


    We first show how to solve an isosceles trapezoid $abdc$ (as shown in Figure \ref{img:trapezoidBSM}) of height $l$, bottom/longer base ($ab$) length $4\ell$, and $\angle cab = \angle dba = 45^{\circ}$. Thus, the top/shorter base $dc$ is of length $2\ell$. Solving trapezoid $abdc$ means computing the values of the cell at the top base $cd$ given the values of the cells at the bottom base $ab$.

   If $\ell \leq 10$, we na\"ively solve $abdc$ and identify the location of the red-green boundary point $p$ in $cd$ in $\Oh{1}$ time. If $\ell > 10$, we find the row $hn$ at height $\frac{\ell}{2}$ and calculate all the cells in it. To do so, we recursively solve the trapezoid $eglj$ which is found as follows:

\vspace{0.1cm}
\textbf{1.} Let $f$ be the point on $ab$ that lies in the green region, but $f + 1$ is in the red region. Identify the points $e$ and $j$ to the left and right of $f$, respectively, such that $|ef| = |fj| = \ell$.

\textbf{2.} Construct an isosceles trapezoid with base $eg$, height $\frac{\ell}{2}$ and top $jl$ such that $\angle jeg = \angle lge = 45^{\circ}$. Thus, $|jl| = \ell$.
\vspace{0.1cm}

   After solving the trapezoid $eglj$, we have the cell values in $jl$ and the location of the red-green boundary point $k$ in $jl$. We can easily calculate the values of cells in $hj$ since those values are independent of time and depend only on spatial coordinates. Finally, we use the $FFT$-based algorithm of \cite{ahmad2021fast} to solve the trapezoid $fbnl$, where the point $l$ is found by forcing $fbnl$ to be an isosceles trapezoid.
   
   Therefore, we can get the values of the cells in $ln$ and thus the values of all the cells in $hn$. 
   Then we can calculate the values of the cells in $dc$ given the values of the cells in $hn$ exactly the same way as we computed the cell values in $hn$ from those in $ab$.

   Now, let us go back to Figure \ref{img:BSMoption} to see how to compute the value of the apex $x$. After solving $abdc$ as above to calculate the cells in $cd$, if $|cd| \leq 10$, we na\"ively calculate the value of $x$, which takes $\Oh{1}$ time. But if $|cd| > 10$, we recursively apply our trapezoid algorithm to solve a smaller trapezoid with the bottom base $cd$.
   
   \hide{
   We recursively apply our trapezoid algorithm again to trapezoid $imqo$ (construction is similar to $eglj$) ,which calculates the cells in $oq$ and the location of the point $p$. We can easily calculate the cells on $co$. Finally, we use the FFT-based algorithm again to calculate the trapezoid $kndq$ (construction is similar to $fbnl$), which calculates the cells on $qd$. Now we finish calculating the value on $cd$. 
}

\begin{theorem}
     \label{thm:BSMParallel-inAppendix}
    Our algorithm solves the American put option pricing problem under BSM in $\Oh{T\log^2 T}$ work and $\Oh{ T }$ span, where $T$ is the number of time steps.
\end{theorem}

\vspace{-.5cm}
\begin{proof}
    
   The proof is very similar to that of Theorem \ref{thm:AOT2}.
   Let $\zeta_1(\ell)$ and $\zeta_\infty(\ell)$ be the work and span, respectively, for solving a trapezoid of height $\ell$ (see Figure \ref{img:trapezoidBSM}). We recursively solve two trapezoids of height $\ell / 2$ each in sequence but use a parallel FFT-based algorithm \cite{ahmad2021fast} on each half, both size $\Th{\ell}$. Since the FFT-based algorithm performs $\Oh{\ell \log{\ell}}$ work in $\Oh{\log{\ell} \log\log{\ell}}$ span, we can write: $\zeta_1(\ell) = 2 \zeta_1\left(\frac{\ell}{2}\right) + \Oh{\ell \log \ell}$ if $\ell > 10$ and $\Oh{1}$ otherwise. Similarly, $\zeta_\infty(\ell) = 2 \zeta_\infty\left(\frac{\ell}{2}\right) + \Oh{\log{\ell} \log\log{\ell}}$ if $\ell > 10$ and $\Oh{1}$ otherwise. Solving, $\zeta_1(\ell) = \Oh{ \ell \log^2{\ell}}$ and $\zeta_\infty(\ell) = \Oh{ \ell }$.

   Now, let $\Psi_1(T)$ and $\Psi_\infty(T)$ be the work and the span, respectively, of solving an isosceles triangle of base size $T$ (see Figure \ref{img:BSMoption}). Then $\Psi_1(T) = \Psi_1\left(\frac{T}{2}\right) + \Oh{T \log^2 T}$ if $T > 10$ and $\Oh{1}$ otherwise. Also, $\Psi_\infty(T) = \Psi_{\infty}\left(\frac{T}{2}\right) + \Oh{T}$ if $T > 10$ and $\Oh{1}$ otherwise. Solving, $\Psi_1(T) = \Oh{T \log^2{T}}$ and $\Psi_\infty = \Oh{T}$.\end{proof}

%% file: 6-experiment.tex
\begin{figure*}[t!]
 \scalebox{1}[0.7]{
    \centering
    \begin{subfigure}[t]{0.3\linewidth}
        \centering
        \includegraphics[width=\linewidth]{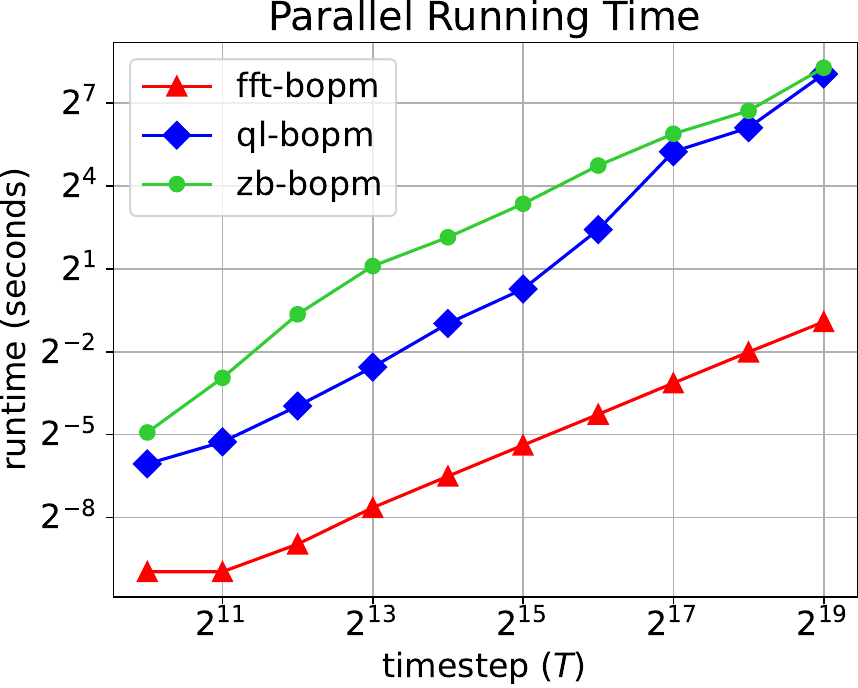}
        \caption{BOPM}
        \label{fig:bopm-runtime}
    \end{subfigure}
    \hfill
    \begin{subfigure}[t]{0.3\linewidth}
        \centering
        \includegraphics[width=\linewidth]{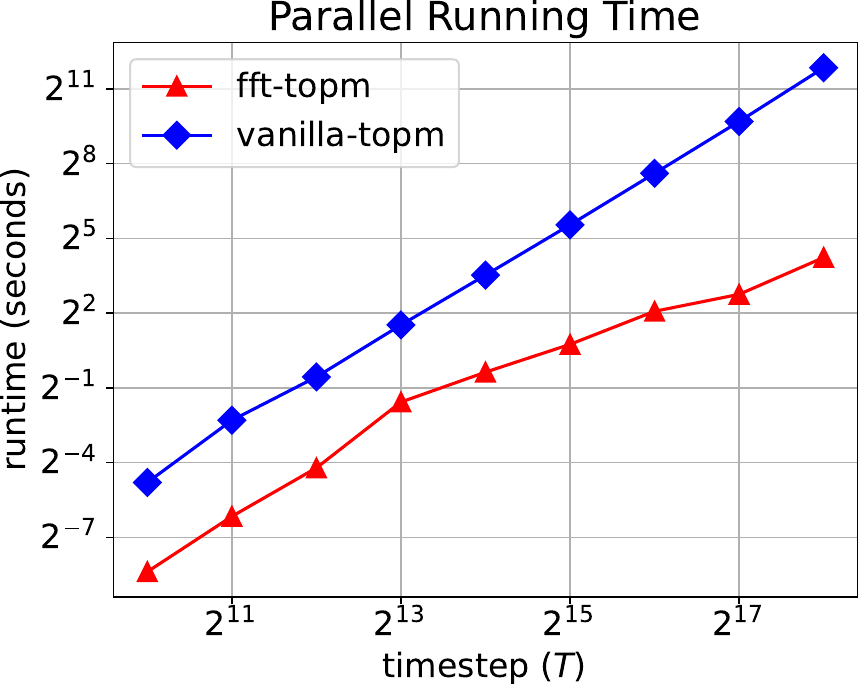}
        \caption{TOPM}
        \label{fig:topm-runtime}
    \end{subfigure}
    \hfill
    \begin{subfigure}[t]{0.3\linewidth}
        \centering
        \includegraphics[width=\linewidth]{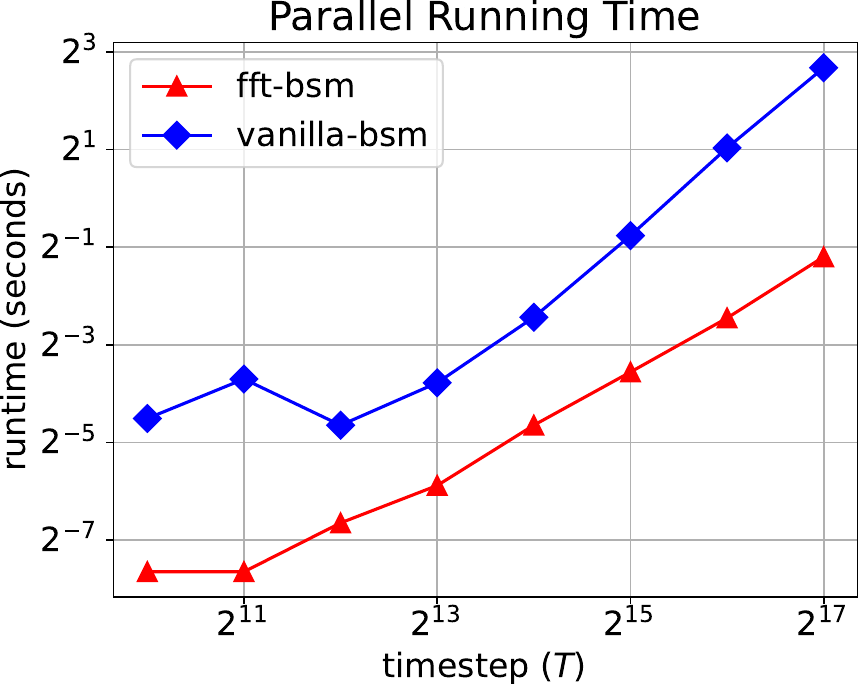}
        \caption{BSM}
        \label{fig:bsm-runtime}
    \end{subfigure}
 }    
 \vspace{-0.4cm}
    \caption{Running time comparisons of our algorithms with state-of-the-art parallel algorithms.}
    \label{fig:runtime-plot}
    \vspace{-0.1cm}
\end{figure*}

\begin{figure*}[ht!]
    \centering
\scalebox{1}[0.7]{
     \begin{subfigure}[b]{0.3\linewidth}
        \centering
        \includegraphics[width=\linewidth]{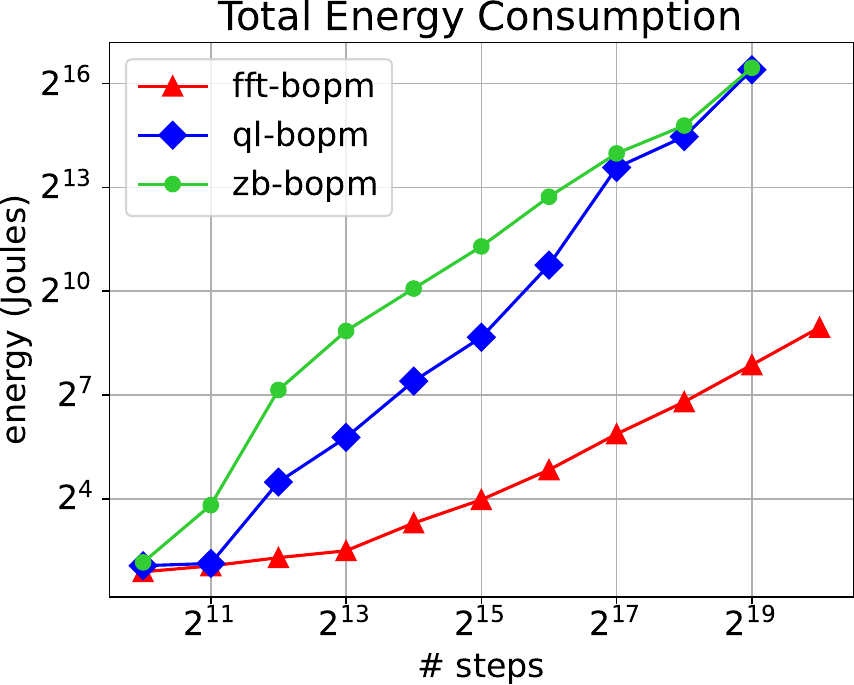}
        \caption{BOPM}
        \label{fig:bopm-energy-total}
     \end{subfigure}
     \hfill
     \begin{subfigure}[b]{0.3\linewidth}
        \centering
        \includegraphics[width=\linewidth]{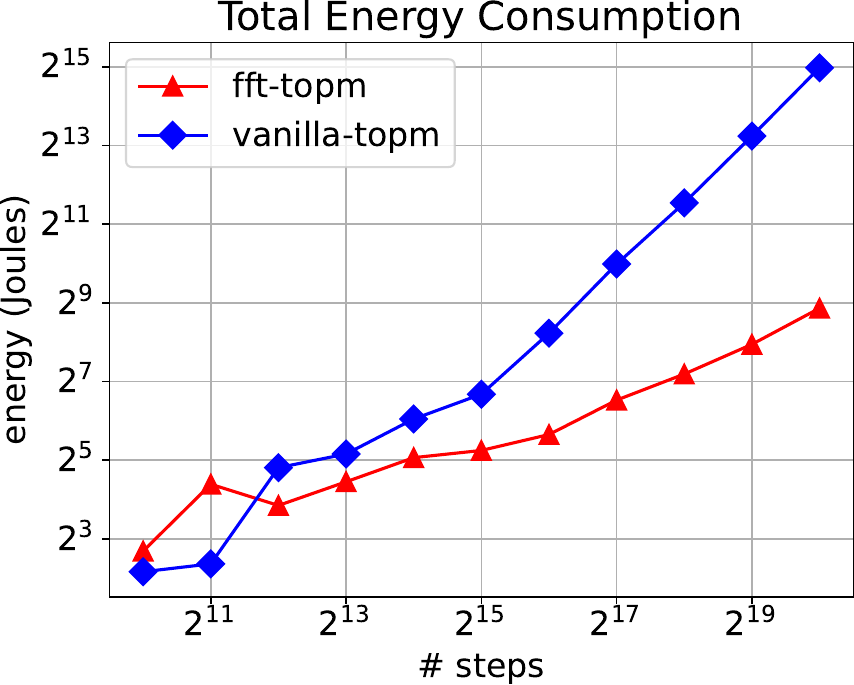}
        \caption{TOPM}
        \label{fig:topm-energy-total}
     \end{subfigure}
    \hfill
     \begin{subfigure}[b]{0.3\linewidth}
        \centering
        \includegraphics[width=\linewidth]{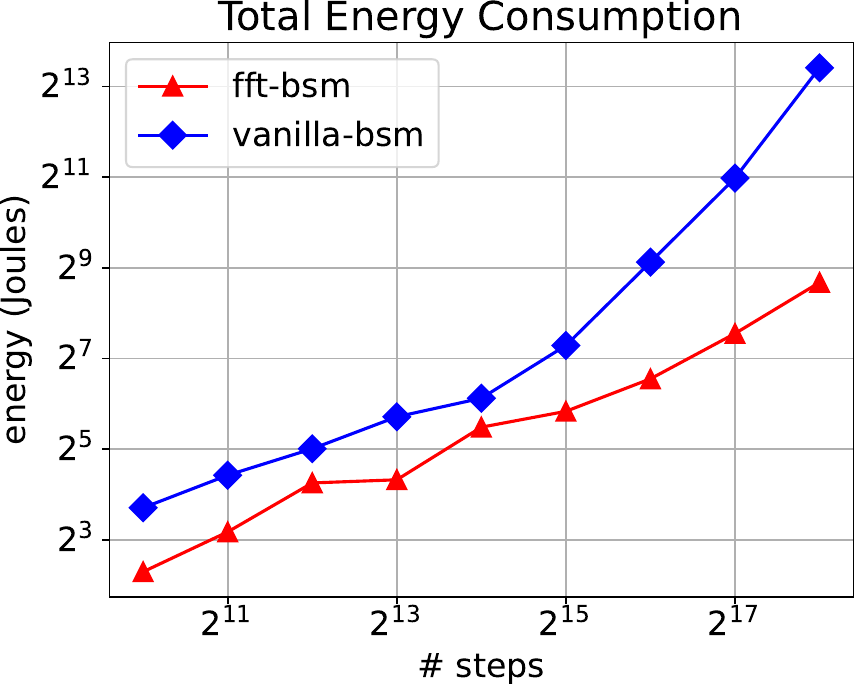}
        \caption{BSM}
        \label{fig:bsm-energy-total}
     \end{subfigure}
     }
\vspace{-0.4cm}     
     \caption{Comparison of energy consumption.}
     \label{fig:energy-plot}
    \vspace{-0.1cm}     
\end{figure*}

\begin{figure*}[ht!]
    \centering
\scalebox{1}[0.7]{    
     \begin{subfigure}[b]{0.3\linewidth}
        \centering
        \includegraphics[width=\linewidth]{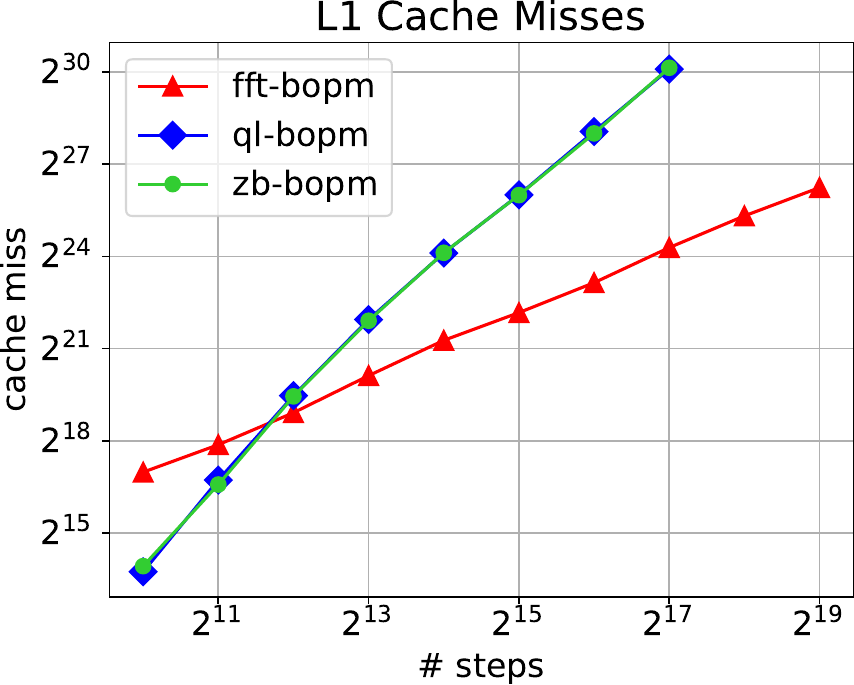}
        \caption{BOPM}
        \label{fig:bopm-L1-cache-miss}
     \end{subfigure}
     \hfill
      \begin{subfigure}[b]{0.3\linewidth}
        \centering
        \includegraphics[width=\linewidth]{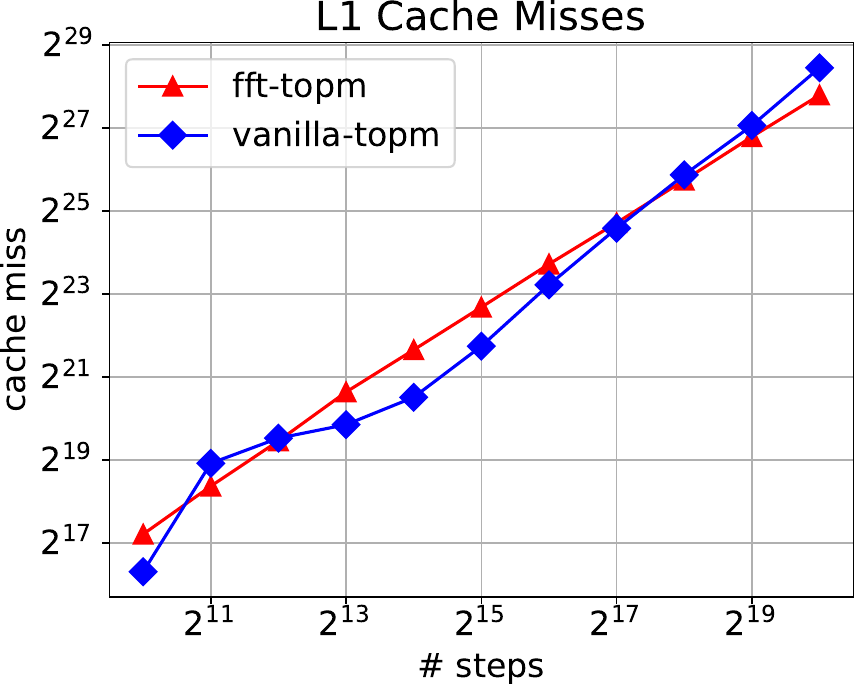}
        \caption{TOPM}
        \label{fig:topm-L1-cache-miss}
     \end{subfigure}
     \hfill
      \begin{subfigure}[b]{0.3\linewidth}
        \centering
        \includegraphics[width=\linewidth]{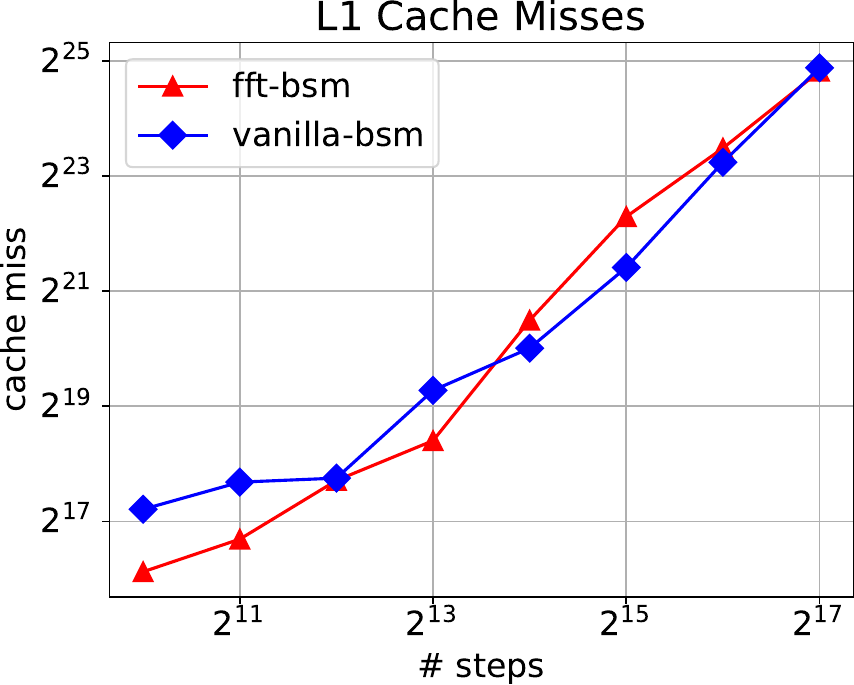}
        \caption{BSM}
        \label{fig:bsm-L1-cache-miss}
     \end{subfigure}
     }

    \centering
\scalebox{1}[0.7]{    
     \begin{subfigure}[b]{0.3\linewidth}
        \centering
        \includegraphics[width=\linewidth]{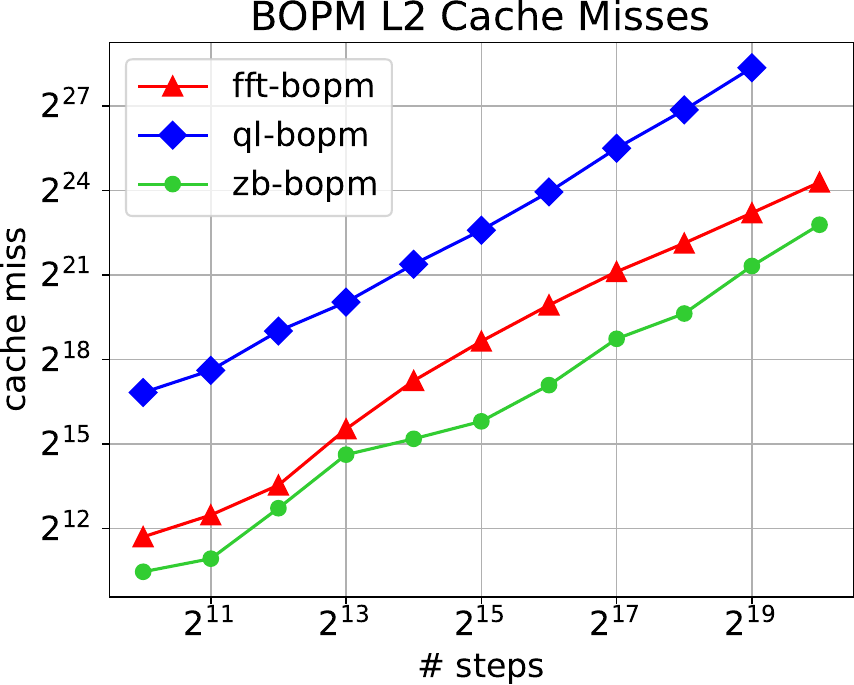}
        \caption{BOPM}
        \label{fig:bopm-L2-cache-miss}
     \end{subfigure}
     \hfill
      \begin{subfigure}[b]{0.3\linewidth}
        \centering
        \includegraphics[width=\linewidth]{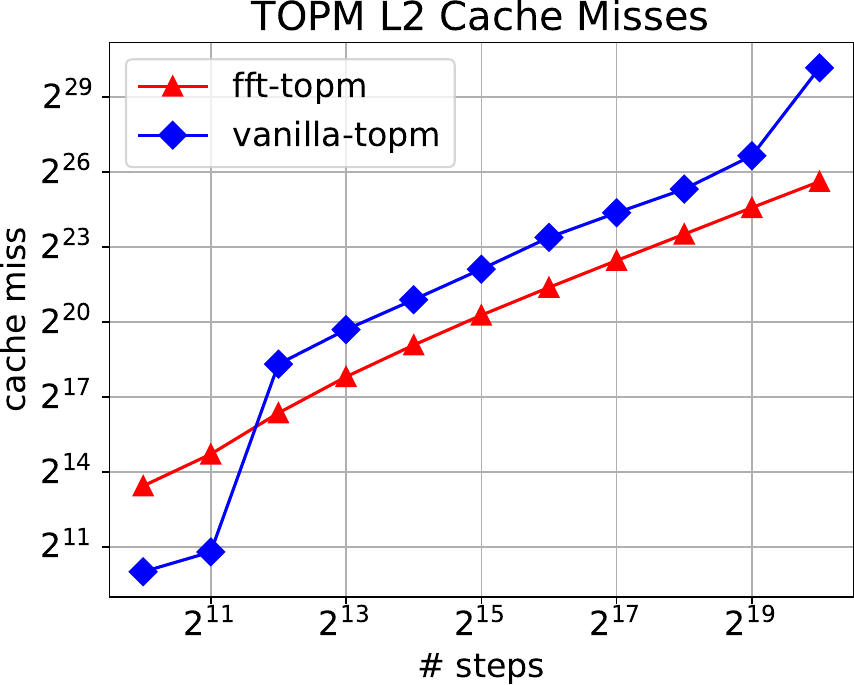}
        \caption{TOPM}
        \label{fig:topm-L2-cache-miss}
     \end{subfigure}
     \hfill
      \begin{subfigure}[b]{0.3\linewidth}
        \centering
        \includegraphics[width=\linewidth]{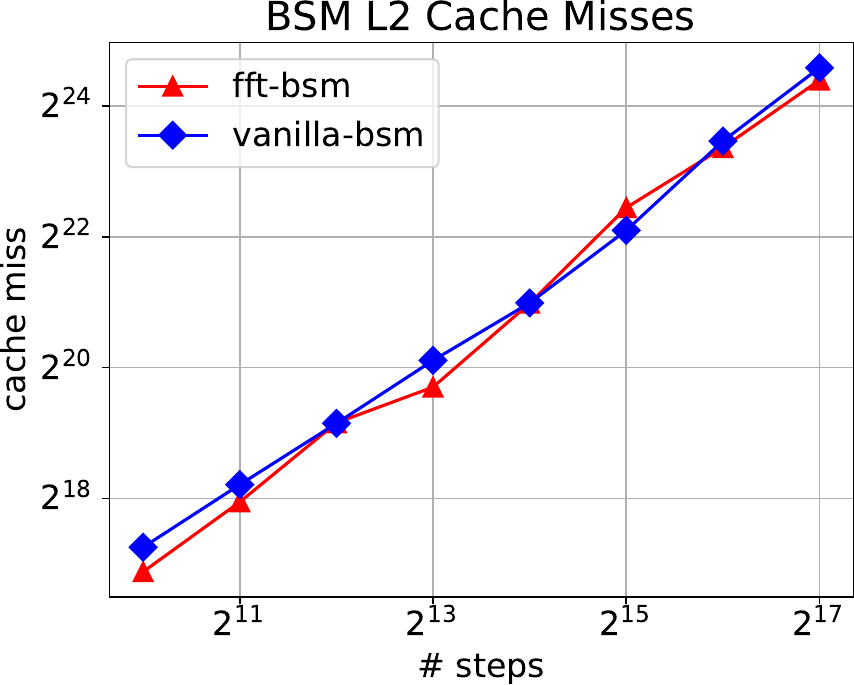}
        \caption{BSM}
        \label{fig:bsm-L2-cache-miss}
     \end{subfigure}
     }
\vspace{-0.4cm}     
     \caption{Comparison of L1 and L2 Cache misses.}
     \label{fig:cache-miss-plot}
     \vspace{-0.3cm}
\end{figure*}

\vspace{-0.4cm}
\section{Experimental Results}
\vspace{-0.1cm}
In this section, we present an experimental evaluation of our algorithms and compare them with the fastest existing solutions. Our experimental setup is shown in Table \ref{tab:experimental-setup}. The legends used in our plots are listed in Table \ref{tab:legends}, which are described in more detail in the next few paragraphs.

\begin{table}[t]
\caption{\small Experimental setup on a Stampede2 \cite{Stampede2} SKX node.}
\centering
\vspace{-0.3cm}
\scalebox{0.74}{
\begin{tabular}{l|l}
\hline
Processor & Intel Xeon Platinum 8160 (Skylake / SKX)\\
Cores & 24 cores per socket, 2 sockets (total: 48 cores)\\
Cache sizes & L1 32 KB / core, L2 1 MB / core, L3 33 MB / socket\\ 
Memory & 144GB /tmp partition on a 200GB SSD\\ \hline
Compiler & Intel C++ Compiler (ICC) v18.0.2\\
Compiler flags & \texttt{-O3 -xhost -ansi-alias -ipo -AVX512}\\
Parallelization & OpenMP 5.0\\ 
Thread affinity & \texttt{GOMP\_CPU\_AFFINITY} \\ \hline
\end{tabular}}
\label{tab:experimental-setup}
\vspace{-0.1cm}
\end{table}

\begin{table}[t]
  \caption{Legends used in plots and tables.}
  \vspace{-0.3cm}  
\scalebox{0.74}{
  \begin{tabular}{>{\centering\arraybackslash}m{1.1in}|>{\centering\arraybackslash}m{2.8in}} 
  \hline
  {\bf{Legend}} & {\bf{Meaning}} \\ \hline
  \texttt{fft-bopm}, \texttt{fft-topm}, \texttt{fft-bsm} & our FFT-based implementations for BOPM, TOPM, and BSM, respectively\\ \hline
  \texttt{ql-bopm}, \texttt{zb-bopm} & BOPM implementations from \parops{} based on \quantlib{} and Zubaer et al.'s work, respectively\\ \hline
  \texttt{vanilla-topm}, \texttt{vanilla-bsm} & our parallel looping implementations for TOPM and BSM, respectively \\ \hline
  \end{tabular}
  }
  \label{tab:legends}
\vspace{-0.3cm}
\end{table}

\para{Benchmarks.} For benchmarks, we use American call option pricing under BOPM and TOPM, and American put option under the BSM. For the BOPM call option benchmarks, our baselines are the option call probability calculations from \quantlib{} \cite{quantlib} and Zubair et al.'s parallel cache optimized model \cite{Zubair}. We use the optimized implementations of these two baselines available in \parops{} \cite{brunelle2022parallelizing}.
%
%
These implementations are the fastest existing implementations of BOPM call option pricing. For the TOPM call option and BSM put option, our FFT-based implementations are 
compared with our parallel looping-based vanilla implementations,
as we could not find any publicly available faster implementations.

We use the {\sl perf} (version: 3.10.0-1160.53.1.el7.x86\_64.debug) tool \cite{perftool} to analyze the system-wide energy consumption, and the {PAPI} (version: 5.6.0) library \cite{papilib} for cache miss counts for our implementations and benchmarks. 

\para{\parops.} In 2022, Brunelle et al. \cite{brunelle2022parallelizing} released an open-source framework that can leverage parallel, cache-optimized algorithms to compute a variety of binomial option types and enables a simple interface for developers. We used the latest version from \href{https://github.com/nyrret/par-bin-options}{Github}. Experimental evaluations by Brunelle et al. \cite{brunelle2022parallelizing} has shown that \parops{} achieves more than $139\times$ speedup over the \quantlib{} library when evaluating a European call option using 200,000 steps. Therefore, we chose the \parops{} tool to benchmark our implementations of our FFT-based algorithms. 
%
%
To have a valid comparison of running times, we use \parops{} for both \quantlib{} and Zubaer et al.'s \cite{Zubair} option probability calculation equations and report the running times comparing with our FFT-based implementations. We use the stencil-based cache-optimized version of Zubaer et al.'s algorithm from \parops{}.

\para{Parameter Values.} As Table \ref{tbl:work_span} shows, the only option pricing parameter that influences the performance bounds of the algorithms in our experiments is the number of time steps $T$. Therefore, we keep all other option pricing parameters fixed in all of our experiments. We use the following parameter values: $E = 252$,  $K = 130$, $S = 127.62$, $R = 0.00163$, $V = 0.2$, $Y = 0.0163$ (notations are in Table \ref{fig:notation}).

\hide{
\para{Legends.} $(1)$ \texttt{fft-bopm}, \texttt{fft-topm}, and \texttt{fft-bsm} are our FFT-based implementations for BOPM, TOPM, and BSM, respectively; $(2)$ \texttt{ql-bopm} and \texttt{zb-bopm} are BOPM implementations from \parops{} based on \quantlib{} and Zubaer et al.'s work \cite{Zubair}; and $(3)$ \texttt{vanilla-topm} and \texttt{vanilla-bsm} are our parallel looping implementations of option pricing under TOPM and BSM, respectively.
}

\hide{
\begin{table}[ht!]
  \vspace{-0.1cm}  
  \caption{Legends used in plots and tables.}
  \vspace{-0.3cm}  
\scalebox{0.85}{
  \begin{tabular}{>{\centering\arraybackslash}m{1in}|>{\centering\arraybackslash}m{2.6in}} 
  \hline
  {\bf{Legend}} & {\bf{Meaning}} \\ \hline
  \texttt{fft-bopm}, \texttt{fft-topm}, \texttt{fft-bsm} & our FFT-based implementations for BOPM, TOPM, and BSM, respectively\\ \hline
  \texttt{ql-bopm}, \texttt{zb-bopm} & BOPM implementations from \parops{} based on \quantlib{} and Zubaer et al.'s work, respectively\\ \hline
  \texttt{vanilla-topm}, \texttt{vanilla-bsm} & our parallel looping implementations for TOPM and BSM, respectively \\ \hline
  \end{tabular}
  }
  \label{tbl:legends}
  \vspace{-0.4cm}
\end{table}
}
\vspace{-0.55cm}
\subsection{Parallel Running Times}
\vspace{-0.05cm}

Figure \ref{fig:runtime-plot}$(a)$ shows the parallel running time comparison of American call option pricing calculations under BOPM. Our FFT-based algorithm uses a recursive divide-and-conquer approach. We have found empirically that a base case size of 8 steps yields the best running times. Experimental results show that our FFT-based algorithm can outperform \parops{} for any number of step sizes for both serial and parallel implementations. We achieve more than $16\times$ speedup for $T \approx 1000$ and more than $500\times$ speedup for $T \approx 0.5~\textrm{million}$ w.r.t. \parops{} implementations.

Our TOPM algorithm runs more than $2.5\times$ faster for $T \approx 2000$ and more than $390\times$ faster for $T \approx 2.1~ \textrm{million}$ w.r.t. the parallel vanilla code. Figure \ref{fig:runtime-plot}$(b)$ shows the comparisons.

Figure \ref{fig:runtime-plot}$(c)$ shows the parallel running time comparisons of the American put option pricing computations under BSM. Our FFT-based implementation is compared with the looping-based vanilla implementation. Our algorithm achieves more than $8\times$ speedup for $T \approx 1000$ and more than $14\times$ speedup for $T \approx 0.13~\textrm{million}$ w.r.t. the vanilla implementation. 

\vspace{-0.55cm}
\subsection{Energy Consumption}
\label{sec:energy}
\vspace{-0.05cm}

Figure \ref{fig:energy-plot} shows the comparison of system-wide energy consumption while running our FFT-based implementation and respective benchmarks. We collected the energy consumption estimate of the entire package ({\sl pkg}) and the main memory ({\sl RAM}) through the RAPL (Running Average Power Limit) interface of model-specific registers (MSR) using the {\sl perf} \cite{perftool} profiling tool. Our FFT-based BOPM, TOPM, and BSM implementations consume 99\%, 99\%, and 96\% less energy, respectively, compared to their benchmarks for large $T$ values used in our experiments. For $T \approx 4000$, the energy savings are 80\%, 50\%, and 40\%, respectively. Energy plots for {\sl pkg} and {\sl RAM} separately are included as supplementary material.

\vspace{-0.3cm}
\subsection{Cache Misses}
\vspace{-0.05cm}

Figure \ref{fig:cache-miss-plot} shows L1 cache-miss (= L2 cache access) and L2 cache-miss results, respectively, for all implementations. Our FFT-based implementation incurs far fewer L1 cache misses than both \parops{} implementations for BOPM. However, while we incur far fewer L2 cache-misses than \texttt{ql-bopm}, the other one (\texttt{zb-bopm}) incurs fewer L2 misses than ours. In case of TOPM, while our FFT-based implementation incurs far fewer L2 misses than our parallel looping implementation (\texttt{vanilla-topm}), the trend is not clear for L1 misses. For BSM, neither L1 nor L2 misses seem to have a clear winner.

\vspace{-0.3cm}
\subsection{Scalability}
\vspace{-0.05cm}

As Table \ref{tbl:work_span} shows, our algorithms reduce the span $\m{T}_\infty$ of solving the option pricing problems that we consider by a factor of $\Om{\log{T}}$, but they reduce the work $\m{T}_1$ by a substantially larger factor of $\Th{T/\log^2{T}}$. While the reduction in work leads to a significant reduction in energy consumption (see Section \ref{sec:energy}), it also leads to a low parallelism of $\m{T}_1 / \m{T}_\infty = \Th{\log^2{T}}$. But it is still easy to see from the complexities given in Table \ref{tbl:work_span} that the parallel running time $\m{T}_p$ of our algorithms will be asymptotically lower than that of the corresponding fastest existing parallel algorithms for every value of $p$ (stated in Proposition \ref{prop:Tp}). 

\begin{table}[h]
\vspace{-0.3cm}
    \caption{Parallel run times (in ms) for $T = 2^{15}$ as $p$ varies.}  
  \centering
  \vspace{-0.3cm}
  \scalebox{0.8}{
  \begin{tabular}{l|rrrrrrr}
  \hline
                 & $p=1$ & $p=2$ & $p=4$ & $p=8$ & $p=16$ & $p=32$ & $p=48$ \\ \hline
    \texttt{fft-bopm} & $32$ & $28$ & $24$ & $26$ & $29$ & $33$ & $38$ \\ 
    \texttt{ql-bopm} & $26552$ & $12498$ & $6785$ & $3530$ & $1950$ & $1324$ & $1191$ \\ \hline
  \end{tabular}
  }

\label{fig:scaling-BOPM}
\vspace{-0.3cm}

\end{table}

Table \ref{fig:scaling-BOPM} above shows how the parallel running times of our FFT-based BOPM implementation (\texttt{fft-bopm}) and that of the \quantlib{}-based implementation from \parops{} (\texttt{ql-bopm}) vary with $p$ as $T$ is kept fixed at $2^{15}$. We see that although the parallel running time of \texttt{fft-bopm} stops decreasing somewhere between $p = 4$ and $p = 8$, it remains significantly faster than \texttt{ql-bopm} even when $p = 48$. Our algorithm scales better as $T$ increases, e.g., for $T = 2^{19}$, it scales to a $p \in [8, 16)$, and compared to a subsecond running time of \texttt{fft-bopm} for $p = 1$, \texttt{ql-bopm} takes $\approx 2$ hours to run. 

\hide{
\begin{figure}[ht!]
 \scalebox{1}[0.65]{
    \centering
    \includegraphics[width=0.8\linewidth]{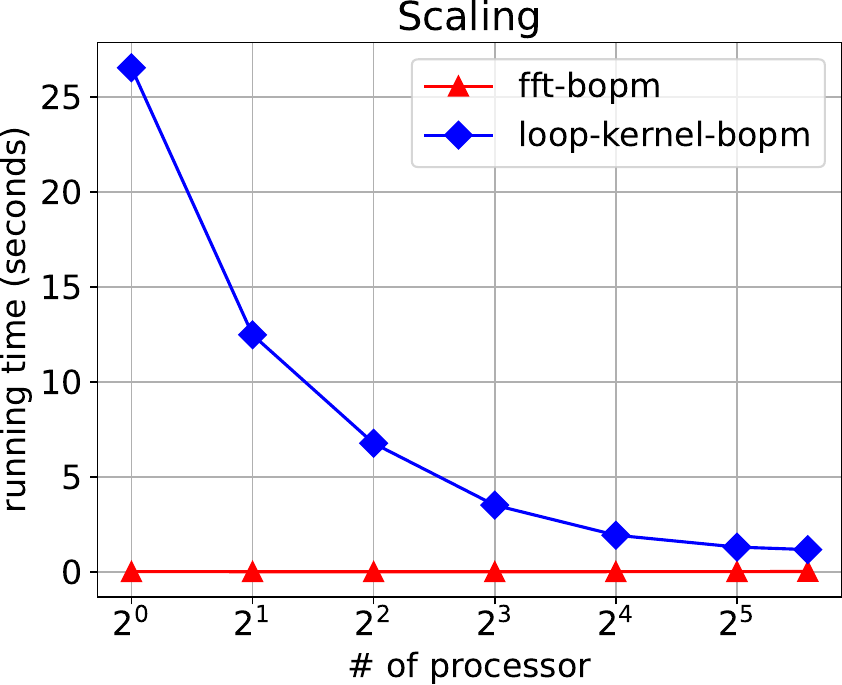}
   }
 \vspace{-0.2cm}
    \caption{Parallel running times of BOPM implementations for $T = 2^{15}$ as $p$ varies.}
    \label{fig:scaling-BOPM}
\end{figure}
}

\hide{
\begin{table}[h]
  \centering
  \begin{tabular}{c|rrrrr}
                 & $p=1$ & $p=2$ & $p=4$ & $p=8$ & $p=16$ \\ \hline
    $T = 2^{15}$ & $32$ & $28$ & $24$ & $26$ & $29$ \\ 
    $T = 2^{19}$ & $708$ & $611$ & $570$ & $511$ & $525$ \\ \hline
  \end{tabular}
  \caption{Parallel running time (in ms) of our BOPM algorithm as $p$ varies.}
\end{table}
}

\hide{
\begin{table}[h]
  \centering
  \begin{tabular}{c|rrrrrrr}
                 & $p=1$ & $p=2$ & $p=4$ & $p=8$ & $p=16$ & $p=32$ & $p=48$ \\ \hline
    $T = 2^{15}$ & $32$ & $28$ & $24$ & $26$ & $29$ & $33$ & $38$ \\ 
    $T = 2^{19}$ & $708$ & $611$ & $570$ & $511$ & $525$ & $741$ & $872$ \\ \hline
  \end{tabular}
\end{table}
}

\vspace{-.2cm}
    \section{Conclusion and Future Work}
\vspace{-.05cm}

We have designed fast American option pricing algorithms under the binomial, trinomial, and the Black-Scholes-Merton models. We solve a type of nonlinear stencil problem that is of independent interest with potential applications beyond quantitative finance.

Future work may explore other models for American option pricing, such as the time dependent volatility model, stochastic volatility model, and the jump diffusion model. European and Asian options, Lookback options, Knock-out Barrier options, and Bermudan options are also of interest. 


\hide{
It is also good to explore whether we can speed up the stencil problem formulated by the Finite Element methods for Black Scholes Model.
}

\para{Supplementary Material.} Includes our FFT-based implementations and details on our TOPM results (Section \ref{sec:TOPM}).

\hide{
\input{5-European-BOPM}

\section{Experimental Results}
\subsection{European Call Option}
Figure \ref{fig:runtime_european} shows the parallel running time comparison of European call option pricing calculations. Our FFT-based implementation outperforms the \parops{} binomial tree traversal framework by a significant margin. Our parallel FFT-based implementation achieves more than 600x speedup for 262144 steps compared to the best running times of parallel vanilla and stencil-based implementations of \parops{}. Moreover, our $\Oh{\sqrt{T}}$ time Gaussian approximation-based implementation completes the calculation very close to zero seconds.
}

\eject

%% file: 5-European-BOPM.tex
 \section{European Call Option under the Binomial Option Pricing Model}
\label{sec:EOP-BOPM}
All results in this section are implied by very recent work by Ahmad et al. \cite{ahmad2021fast,ahmad2022brief} on fast linear stencil computations. We are adding this section for completeness since these results have not been reported in the literature before.

Section \ref{ssec:BOPM} explains that for European options under the Binomial Option Pricing Model, the value of the node at spatial coordinate $j$ and time step $t - 1$ is computed as follows: $X_{t - 1, j} = s_0 \cdot X_{t, j} + s_1 \cdot X_{t, j+1}$. This is a 2-point linear stencil computation on a 1D spatial grid. The goal is to compute $X_{0, 0}$. 

The approach used by Ahmad et al. \cite{ahmad2021fast} for fast stencil computations on periodic grids can be used for solving the stencil problem above in $\Oh{T \log{T}}$ time. It involves applying forward FFT to both the input grid $X[T, 0..T]$ and the stencil grid (created from $s_0$ and $s_1$), performing repeated squaring of the stencil grid in the frequency domain for $\log_{2}{T}$ times, taking point-wise products of that repeatedly squared stencil grid with the input grid in the frequency domain, and finally applying inverse FFT on that product grid to obtain $X[0, 0..T]$ and thus $X[0, 0]$.

The advantage of the FFT-based algorithm described above is that, unlike existing FFT-based algorithms for European options, one does not need to know a closed-form expression for the characteristic function of the log price and it is also not restricted to problems with infinite domains. 

\hide{
Recall that we have $T$ steps, expiration time in $E$ days, strike price $K$,  risk-free rate $R$, volatility $V$, and
dividend yield $Y$. We also have a spot price $S$. The spot price is the current market price which can be considered as the reference price while the parties agree to a certain strike price.

Denote the following:
\begin{align*}
    \Delta t &= \frac{E}{365 \cdot T}\\
    m &= e^{- R \cdot \Delta t }\\
    p &= \frac{ e^{( R - Y ) \cdot \Delta t } - d }{u - d}\\
    \mu &= - p T\\
    \sigma &= p(1 - p) T
\end{align*}
}

A more recent work by Ahmad et al. \cite{ahmad2022brief} implies that the following sum gives a good approximation of $X[0, 0]$ in $\Oh{\sqrt{T}}$ time.
$$\frac{m^T}{\sqrt{2 \pi \sigma}}\sum_{x =  \max{\{\mu - 6\sqrt{\sigma}, 0\}}}^{x =  \min{\{\mu + 6\sqrt{\sigma}, T\}}   - 1}{ \left( \max{\left\{ S \cdot u^{2x - T} - K, 0 \right\}} \cdot e^{- \frac{(x + \mu)^2}{2 \sigma}}\right)  }$$

A better bound can be approximated in $O(T)$ time by extending the summation bound.
$$\frac{m^T}{\sqrt{2 \pi \sigma}}\sum_{x =  0}^{x =  T - 1}{ \left( \max{\left\{ S \cdot u^{2x - T} - K, 0 \right\}} \cdot e^{- \frac{(x + \mu)^2}{2 \sigma}}\right)  }$$

We can get rid of the $\max$ operator by changing the summation bound as follows.
$$\frac{m^T}{\sqrt{2 \pi \sigma}}\sum_{x =  \left\lceil { \frac{T}{2} + \frac{1}{2 \ln u} \ln{\left(\frac{K}{S}\right)} } \right\rceil  }^{x = T - 1}{ \left( \left( S \cdot u^{2x-T} - K \right) \cdot e^{-\frac{(x + \mu)^2}{2 \sigma}} \right)  }$$

We can approximate the sum above using the following integral:
$$\frac{m^T}{\sqrt{2 \pi \sigma}}\int_{x =  \left\lceil { \frac{T}{2} + \frac{1}{2 \ln u} \ln{\left(\frac{K}{S}\right)} } \right\rceil  }^{x = T}{ \left( \left( S \cdot u^{2x-T} - K \right) \cdot e^{-\frac{(x + \mu)^2}{2 \sigma}} \right)  }$$

The integral solves the following:
\begin{align*}
& \frac{m^T}{2}\biggl[ S \cdot u^{2( \ln u \sigma - \mu ) - T ) } \cdot \textrm{erf}{\left( \frac{x + \mu}{\sqrt{2 \sigma}} - \ln{u}\sqrt{2\sigma}\right)} - K \cdot \textrm{erf}{\left( \frac{x + \mu}{\sqrt{2 \sigma}} \right)} \biggr]_{x = \left\lceil {\frac{T}{2} + \frac{1}{2\ln{u}} \ln{\left(\frac{K}{S}\right)} } \right\rceil}^{x = T}
\end{align*}

\hide{
Let's simplify (\ref{epr:sum-in-logn}):

\begin{align*}
&\frac{m^T}{\sqrt{2 \pi \sigma}}\sum_{x =  0 }^{x = T - 1}{ \left( \max{\left\{ S \cdot u^{2x - T} - K, 0 \right\}} \cdot e^{- \frac{(x + \mu)^2}{2 \sigma}}\right)  }\\
=&\frac{m^T}{\sqrt{2 \pi \sigma}}\sum_{x =  \left\lceil { \frac{T}{2} + \frac{1}{2 \ln u} \ln{\left(\frac{K}{S}\right)} } \right\rceil  }^{x = T - 1}{ \left( \left( S \cdot u^{2x-T} - K \right) \cdot e^{-\frac{(x + \mu)^2}{2 \sigma}} \right)  }\\
\approx & \frac{m^T}{2}\biggl[ S \cdot u^{2( \ln u \sigma - \mu ) - T ) } \cdot \textrm{erf}{\left( \frac{x + \mu}{\sqrt{2 \sigma}} - \ln{u}\sqrt{2\sigma}\right)} \\
& \quad - K \cdot \textrm{erf}{\left( \frac{x + \mu}{\sqrt{2 \sigma}} \right)} \biggr]_{x = \left\lceil {\frac{T}{2} + \frac{1}{2\ln{u}} \ln{\left(\frac{K}{S}\right)} } \right\rceil}^{x = T} \quad \mbox{(use integral)}
\end{align*}
}

Since the error function \textbf{erf} can be approximated well in $\Oh{1}$ time, we can evaluate the expression above in $\Oh{1}$ time as well.

Hence, one can approximate $X[0,0]$ in $\Oh{1}$ time.

\begin{figure}
    \centering
    \includegraphics[width=.5\textwidth]{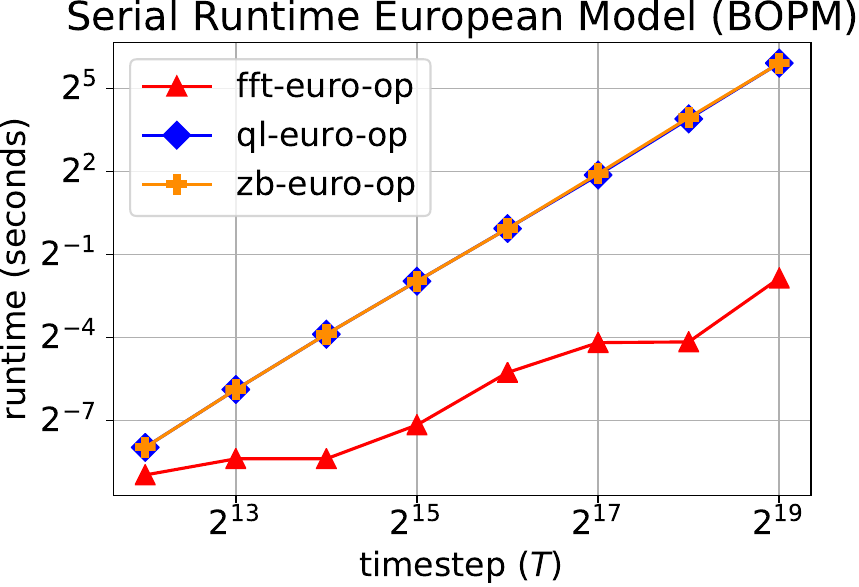}
    \caption{Caption}
    \label{fig:runtime_european}
\end{figure}

%% file: appendix_arxiv.tex
\pagestyle{plain}


\section{American Call Option under Trinomial Option Pricing Model} \label{sec:TOPM2}

\subsection{Trinomial Option Pricing Model (TOPM)}

\begin{figure}[t]
    \centering
    \includegraphics[width=0.3\textwidth]{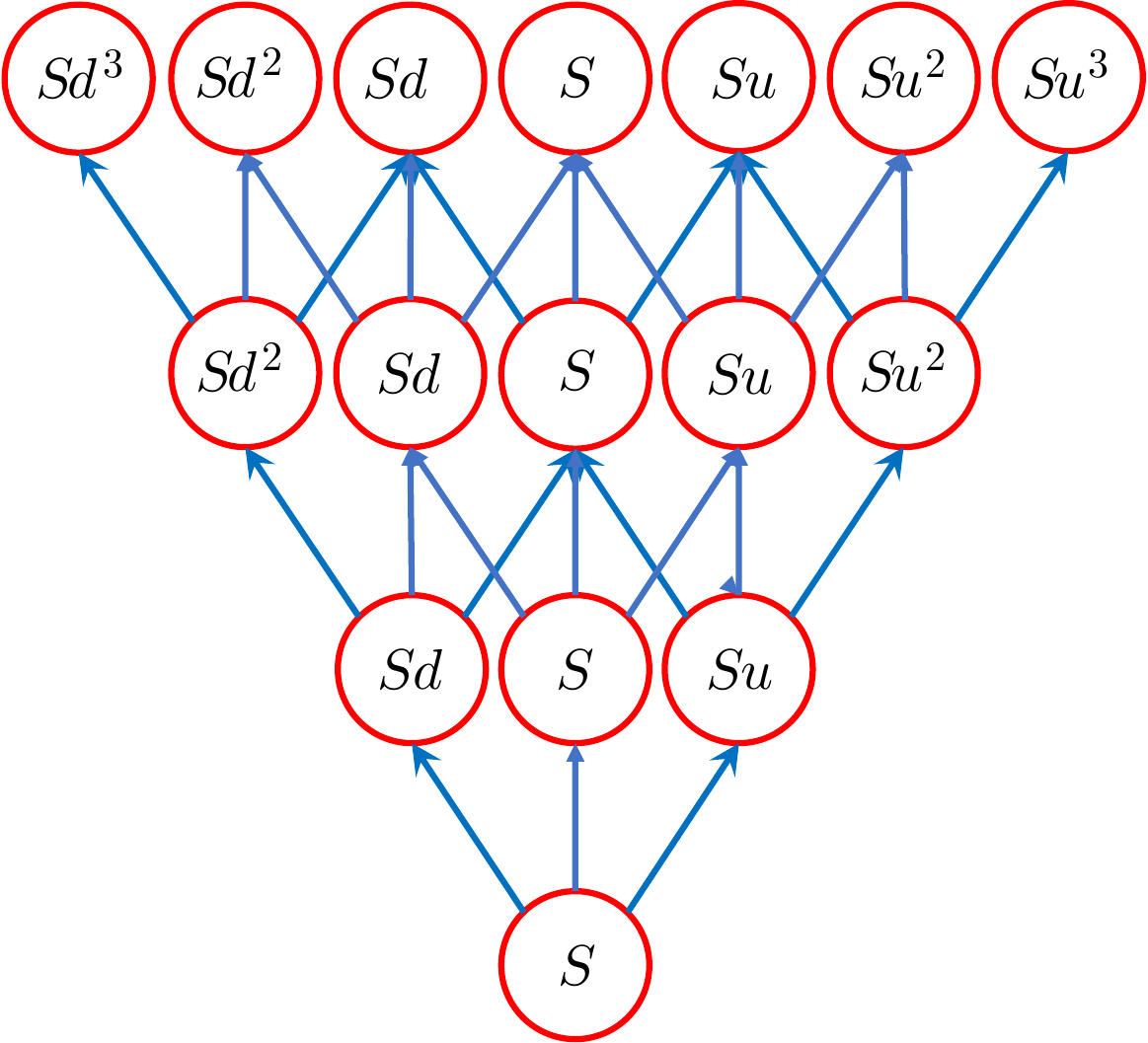}
    \caption{A trinomial tree}
    \label{fig:tritreeapp}
\end{figure}

TOPM shares many properties with BOPM, which is natural as a trinomial tree (see Figure \ref{fig:tritreeapp} has similar structure to that of a binomial tree. The node-value remains $X_{node} = S \times u^{N_u-N_d}$ since the number of "remain the same" moves does not factor into the price. However, the value of $u$ and $d$ is subtly different, as $u = e^{V\sqrt{2 \Delta t}}$, and $d = 1/u$. We represent the exercise value very similarly, $\text{max}(X_{node}-K,0)$.

For the \textbf{trinomial value}, we have now three dependencies, $X_{t,j},X_{t,j+1},$ and $X_{t,j+2}$. Since we have three possible changes in value, we have three probabilities corresponding to moving up $p_u$, moving down $p_d$, and remaining the same $p_{\probm}$. Typically, the literature refers to the probability of remaining the same as $p_m$, but as since use $m$ to refer to the quantity $e^{-Rt}$, we use $p_{\probm}$ to avoid confusion.  This lets us write the trinomial value as $e^{-R\Delta t} (p_d X_{t,j} + p_{\probm} X_{t,j+1} + p_u X_{t,j+2})$, where $p_u = \left( \frac{e^{(R-Y)\Delta t/2}-\sqrt{d}}{\sqrt{u} - \sqrt{d}} \right) ^2$,\\ $p_d = \left( \frac{ \sqrt{u} -e^{(R-Y)\Delta t/2}}{\sqrt{u} - \sqrt{d}} \right) ^2$, and $p_{\probm} = 1 - p_u - p_d$. The probabilities are alternate forms of those given in \cite{hull2003options}.

Let $m = e^{-R\Delta t}$, $s_0 = m p_u, s_1 = m p_{\probm}, s_2 = m p_d$.
We will reuse the notion of red and green cells to distinguish those with larger binomial and excercise values respectively. 

\vspace{-0.2in}
\begin{equation*}
\begin{aligned}
&G^{red}_{i, j} = 
    \begin{cases}
       0, &\text{if }i=T\\
       s_0 G_{i+1,j}+ s_1 G_{i+1,j+1} + s_2 G_{i+1,j+2}, &0 \le i < T
   \end{cases}\\
&G^{green}_{i, j} = S \cdot u^{j-i}-K
\end{aligned}
\end{equation*}

\subsection{Properties of the Red-Green Divider}\label{ssec:TOPM-boundary}
Similarly to the binomial tree, we assume that the trinomial tree will be embedded in a $(2T+2) \times (T+1)$ grid $G$ with the root at the bottom-left corner $G[0,0]$ and the leaves in the top row $G[T,0...T]$. The children of node $G[i,j]$ are at $G[i+1, j], G[i+1, j+1], G[i+1, j+2]$. $G[i+1,j]$ has a price change factor of $d$, $G[i+1,j+1]$ has the same price and $G[i+1, j+2]$ has a price change factor of $u$. The tree only takes up the upper left triangle, such that the diagonal has a slope of 1/2. 

Lemma \ref{lma:rightOfGreenIsGreenTri} shows that within the upper-left triangular area of $G$ if a cell is green then the cell to its right must also be green.

\begin{lemma}
    \label{lma:rightOfGreenIsGreenTri}
    For $i \in [0, T]$ and $j \in [0, 2i - 1]$, $\left( G_{i,j} = G^{green}_{i,j} \right) \implies \left( G_{i,j+1} = G^{green}_{i,j+1} \right)$.
\end{lemma}
\begin{proof}

We will reuse the notions $\delta _{i,j} = G_{i,j}^{red} - G_{i,j}, \Delta _{i,j} = \delta _{i,j+1} - \delta_{i,j}, \Tilde{\delta}_{i,j} = G_{i,j} - G_{i,j}^{green}$ and $\Tilde{\Delta}_{i,j} = \Tilde{\delta}_{i,j+1} - \Tilde{\delta}_{i,j}$.
Therefore, it suffices again to show that $\Delta _{i,j} \leq 0$ for all $0 \leq i \leq T$ and $j < 2i$. This can again be show through induction.

Consider the base case as $i = T$. $\delta _{T,j} = -\left( Su^{j-T} - K \right)$, which in turn means that $\Delta _{T,j} = -S u^{j-T} (u-1) \leq 0$ since $u geq 1$. So, the claim holds for $i=T$.

Assuming the claim holds for some $i +1 \leq T$, we will show that it holds for $i$ ($\Delta_{i+1,j} \leq 0$ for $0 \leq j < i+1$). From this assumption, this implies the existence of a $j_{i+1}$ such that $G_{i,j} = G_{i,j}^{green}$ for $j > j_{i+1}$ and $G_{i,j} 
= G_{i,j}^{red}$ for $j \leq j_{i+1}$. This $j_{i+1}$ is defined as the largest $j$ such that $\delta_{i+1,j} > 0$, the last "red" cell in row $i+1$. Therefore:
\begin{itemize}
\item[$-$] for $j > j_{i+1}$, $\tilde{\Delta}_{i+1,j} = \tilde{\delta}_{i+1,j+1} - \tilde{\delta}_{i+1,j} 
     = 0 - 0 = 0$;
\item[$-$] for $j  = j_{i+1}$, $\tilde{\Delta}_{i+1,j} = \tilde{\delta}_{i+1,j+1} - \tilde{\delta}_{i+1,j} 
     = -\tilde{\delta}_{i+1,j} \leq 0$; and
\item[$-$] for $j < j_{i+1}$, $\tilde{\Delta}_{i+1,j} = \Delta_{i+1,j} \leq 0$.  
\end{itemize}

Now we examine:
\begin{align*}
    \Delta_{i,j} =& \delta_{i,j+1} - \delta_{i,j}\\
    =& s_0 G_{i+1,j+1} + s_1 G_{i+1,j+2} + s_2 G_{i+1,j+3} - G_{i,j+1}^{green} - s_0 G_{i+1,j} \\
    &- s_1 G_{i+1,j+1}
    - s_2 G_{i+1,j+2} + G_{i,j}^{green} \\
    =& s_0 (G_{i+1,j+1} - G_{i+1,j+1}^{green} - G_{i+1,j} + G_{i+1,j}^{green}) + s_1 (G_{i+1,j+2} \\
    &- G_{i+1,j+2}^{green} - G_{i+1,j+1} + G_{i+1,j+1}^{green}) + s_2 (G_{i+1,j+3} - G_{i+1,j+3}^{green} \\
    &- G_{i+1,j+2} + G_{i+1,j+2}^{green}) + s_0 (G_{i+1,j+1}^{green} - G_{i+1,j}^{green}) + s_1 (G_{i+1,j+2}^{green} \\
    &-G_{i+1,j+1}^{green}) + s_2 (G_{i+1,j+3}^{green} - G_{i+1,j+2}^{green}) + G_{i,j}^{green} - G_{i,j+1}^{green} \\
    =& s_0 \tilde{\delta}_{i+1,j} + s_1 \tilde{\delta}_{i+1,j+1} + s_2 \tilde{\delta}_{i+1,j+2} + s_0 (G_{i+1,j+1}^{green} - G_{i+1,j}^{green}) \\
    &+ s_1 (G_{i+1,j+2}^{green}-G_{i+1,j+1}^{green}) + s_2 (G_{i+1,j+3}^{green} - G_{i+1,j+2}^{green}) + G_{i,j}^{green} \\
    &- G_{i,j+1}^{green}
\end{align*}

From the inductive assumption, we can assume that $s_0 \tilde{\delta}_{i+1,j} + s_1 \tilde{\delta}_{i+1,j+1} + s_2 \tilde{\delta}_{i+1,j+2} \leq 0$. Therefore, to show that $\Delta_{i,j} \leq 0$, it suffices to show that the remaining terms are less than or equal to 0. A useful note is that since diagonals indicate no change with respect to the exercise value, $G_{i+1,j+2}^{green} - G_{i+1,j+1}^{green} = G_{i,j+1}^{green} - G_{i,j}^{green}$, which lets us simplify the rest of this expression. 

\vspace{-0.15in}
\begin{align*}
    0 \geq & s_0 (G_{i+1,j+1}^{green} - G_{i+1,j}^{green}) + (s_1-1) (G_{i+1,j+2}^{green} - G_{i+1,j+1}^{green})
     \\
    &\quad + s_2 (G_{i+1,j+3}^{green} - G_{i+1,j+2}^{green}) \\
    \geq & s_0 (S u^{j+1-i-1} - K - S u^{j-i-1} + K) + (s_1-1) (S u^{j+2-i-1} - K \\
    &\quad- S u^{j+1-i-1} + K) \\
    &\quad + s_2 (S u^{j+3-i-1} - K - S u^{j+2-i-1} + K) \\
    0 \geq & Su^{j-i-1} (s_0 (u-1) + (s_1-1)(u^2-u) + s_2 (u^3 - u^2) )    
\end{align*}
From here, we can start dividing out common terms so long as they are greater than 0.
\begin{align*}
    0 \geq & s_0 (u-1) + (s_1-1)(u^2-u) + s_2 (u^3 - u^2) \\
    0 \geq &m p_d (u-1) + (m p_{\probm} -1) (u^2-u) + m p_u (u^3 - u^2) \\
    0 \geq &mp_d (u-1) + (mp_{\probm} - 1) u(u-1) + mp_u u^2 (u-1) \\
    0 \geq &mp_d + (m-m p_d - m p_u -1) u + mp_u u^2 \\
    0 \geq & mp_d(1-u) + mp_u(u^2-u) + (m-1)u \\
    0 \geq & m(u-1) (up_u - p_d) + (m-1)u
\end{align*}

We can transform the expression $up_u - p_d$. We will let $e^{(R-Y)\Delta t} = n$ to simplify.

\vspace{-0.15in}
\begin{align*}
    u p_u - p_d =& u\left(\frac{\sqrt{n}-\sqrt{d}}{\sqrt{u}-\sqrt{d}} \right)^2 - \left(\frac{\sqrt{u}-\sqrt{n}}{\sqrt{u}-\sqrt{d}} \right)^2 \\
    =& \frac{1}{(\sqrt{u}-\sqrt{d})^2} \left((\sqrt{nu} - \sqrt{du})^2 - (\sqrt{u}-\sqrt{n})^2  \right) \\
    =& \frac{1}{(\sqrt{u}-\sqrt{d})^2} \left( \sqrt{nu} - 1 - \sqrt{u} + \sqrt{n} \right)
    \\& \quad \cdot
    \left(  \sqrt{nu} - 1 + \sqrt{u} - \sqrt{n} \right) \\
    =& \frac{1}{(\sqrt{u}-\sqrt{d})^2} \left( (\sqrt{n} - 1)(\sqrt{u} + 1)\right)\left(  (\sqrt{n} +1)(\sqrt{u}-1) \right) \\
    =& \frac{(n-1)(u - 1)}{(\sqrt{u}-\sqrt{d})^2}
\end{align*}

We can then substitute this back into the previous inequality: $0 \geq (n-1)m(u-1)^2(\sqrt{u}-\sqrt{d})^{-2} + um - u$.
Since $(\sqrt{u}-\sqrt{d})^{-2} = (u+d-2)^{-1} \geq (u-1)^{-1}$, we can make this substitution and if the resulting expression is less than or equal to 0, clearly the current expression as is will be less than or equal to 0 as well.
\begin{align*}
    0 \geq& m(u-1)(n-1) + um - u \\
    0 \geq& m(un -n -u + 1) +um - u \\
    0 \geq &mun - mn - mu + m + um - u \\
    u \geq &mun - mn + m \\ 
    u -1 \geq & (u-1) e^{-Y \Delta t} + e^{-R \Delta t} -1
\end{align*}

Since $e^{-R \Delta t} \leq 1$, we get:
$u-1 \geq (u-1) e^{-Y \Delta t} \geq  (u-1) e^{-Y \Delta t} + e^{-R \Delta t} -1$
Which clearly holds since $e^{-Y \Delta t} \leq 1$. 

By induction, $\Delta _{i,j} \leq 0$ for all $0 \leq i \leq T$, $0 \leq j \leq 2i-1$, we prove Lemma \ref{lma:rightOfGreenIsGreenTri}.
\end{proof}

Lemma \ref{lma:bottomOfGreenIsGreenTri} shows that within the upper-left triangular area of $G$ if a cell is green then the cell below it must also be green.

\begin{lemma}
    \label{lma:bottomOfGreenIsGreenTri}
    For $i \in [0, T - 1]$ and $j \in [0, 2i]$, $\left( G_{i+1,j} = G^{green}_{i+1,j} \right)$ $ \implies \left( G_{i,j} = G^{green}_{i,j} \right)$.
\end{lemma}
\begin{proof}
 
Assume that this is false. This would imply that for some i,j $\left( G_{i,j} = G^{red}_{i,j} \right)$ when $\left( G_{i+1,j} = G^{green}_{i+1,j} \right)$. So this would imply:

$Su^{j-i} - K \leq s_0 G_{i+1,j} + s_1 G_{i+1,j+1} + s_2 G_{i+1,j+2}$

From Lemma \ref{lma:rightOfGreenIsGreenTri}, we know that all the $G_{i+1,k} = G_{i+1,k}^{green}$ for all $j <= k <= 2i+1$. Therefore, we can expand this to:
\begin{align*}
    Su^{j-i} - K &\leq s_0 (S u^{j-i-1} -K) + s_1 (S u^{j-i}-K) \\
    &\quad+ s_2 (S u^{j-i+1}-K)\\
    Su^{j-i} - K &\leq Su^{j-i-1}(s_0 + s_1u + s_2u^2)
    \\ &\quad
    - K(s_0+s_1+s_2)\\
    K(s_0+s_1+s_2 -1) &\leq  Su^{j-i-1}(s_0 + s_1u - u + s_2u^2)\\
    Km(p_d + p_{\probm} + p_u - 1/m) &\leq Su^{j-i-1}(s_0 + (s_1 -1)u + s_2u^2)
\end{align*}
Note that since $p_d+p_{\probm}+p_u = 1$, we know that the expression on the left side is strictly greater than 0, allowing us to make the following substitution:
\begin{align*}
    0 & < Su^{j-i-1}(s_0 + (s_1 -1)u + s_2u^2) \\
    0  &  < s_0 + (s_1 - 1) u + s_2 u^2
\end{align*}
Since $u>1$, so multiplying both sides by $u-1$ gives us this expression from Lemma \ref{lma:rightOfGreenIsGreenTri}.
$$0 < s_0(u-1) + (s_1-1)(u^2 - u) + s_2 (u^3-u^2)$$
In Lemma \ref{lma:rightOfGreenIsGreenTri}, we showed this expression to be less than or equal to 0, which is a contradiction. Therefore, Lemma \ref{lma:bottomOfGreenIsGreenTri} holds.
\end{proof}

\begin{lemma} \label{lma:VNonDecAtRowTri}
     $G_{i, j-1} \le G_{i, j}$ for any $0 \le i \le T-1$ and $1\le j \le 2i$
\end{lemma}
\begin{proof}
    This can be done exactly as is done in Lemma 2.5, by induction. We have the same base case, when $i=T$ we have $G_{T,j-1} \leq G_{T,j}$ because $S u^{j-1-T} - K \leq S u^{j-T} - K $. We then can consider the general case of $i=\ell$, where all $T \leq i > \ell$ are true.
\begin{align}
    & \begin{aligned}
    s_0 G_{\ell,j-1}+ s_1 G_{\ell,j}  + s_2  G_{\ell,j+1}\leq s_0 G_{\ell,j}+ s_1 G_{\ell,j+1} + G_{\ell,j+2}\\ \mbox{(by induc. Hyp.)}
  \end{aligned}\label{eq2tri}\\
    & S u^{j-1-\ell}-K \leq S u^{j-\ell}-K  \quad \mbox{(}\because u \ge 1 \mbox{)} \label{eq3tri}
\end{align}
It implies that
\begin{align*}
    G_{\ell-1,j-1} &= \max( s_0 G_{\ell,j-1}+ s_1 G_{\ell,j} + s_2  G_{\ell,j+1},S u^{2j-1-\ell}-K) \\
    & \leq \max( s_0 G_{\ell,j}+ s_1 G_{\ell,j+1}  + s_2  G_{\ell,j+2} ,S u^{2j-\ell}-K)\\ &\hspace{4.2cm}  \mbox{(by (\ref{eq2tri}), and (\ref{eq3tri}))}\\
    & = G_{\ell-1,j}
\end{align*}\end{proof}

\begin{lemma} \label{lma:NonDecAtDiagTri}
     $G_{i, j} \le G_{i-1, j-1}$ for any $1 \le i \le T$ and $1\le j \le 2i$
\end{lemma}

\begin{proof}
    This can also be proved by induction using the similar methods as Lemma 2.5. We have the same base case, when $i=T$ we have $G_{T-1,j-1} \leq G_{T,j}$. If both are green, then since $S u^{j-1-T +1} - K = S u^{j-T} - K $ it holds. If $G_{T-1, j-1}$ is green and $G_{T, j}$ is red, then we know that $G_{i-1,j-1} \geq 0$ since it is green and that $G_{T,j} = 0$ because it is red. If both are red, then the same applies. We then can consider the general case of $i=\ell$, where if all $\ell < i \leq T$ hold.
\begin{align}
    &\begin{aligned}
    s_0 G_{\ell+1,j}+ s_1 G_{\ell+1,j+1}  + s_2  G_{\ell+1,j+2} \leq s_0 G_{\ell,j-1}+ s_1 G_{\ell,j} + G_{\ell,j+1}
    \\\mbox{(by induc. Hyp.)} 
  \end{aligned}\label{eq4tri}\\
    & S u^{j-\ell}-K \leq S u^{j-1-\ell+1}-K \quad \mbox{(}\because  \mbox{They are equal)} \label{eq5tri}
\end{align}
It implies that
\begin{align*}
    G_{\ell,j} &= \max( s_0 G_{\ell+1,j}+ s_1 G_{\ell+1,j+1} + s_2  G_{\ell+1,j+2},S u^{2j-\ell}-K) \\
    & \leq \max( s_0 G_{\ell,j-1}+ s_1 G_{\ell,j}  + s_2  G_{\ell,j+1} ,S u^{2j-1-\ell+1}-K) 
    \\
    & \hspace{5.55cm} \mbox{(by (\ref{eq4tri}),(\ref{eq5tri}))}\\
    & = G_{\ell-1,j-1}
\end{align*}\end{proof}

\begin{lemma} \label{lma:diagLeftOfRedIsRedTri}
    For $i \in [0, T - 1]$ and $j \in [0, 2i - 1]$,\\ $\left(G_{i+1,j+1} = \\G_{i+1,j+1}^{red}\right)$ $\Rightarrow \left(G_{i,j} = G_{i,j}^{red}\right)$.
\end{lemma}

\begin{proof}
    This follows directly from Lemma \ref{lma:NonDecAtDiagTri}. Recall $G_{i+1,j+1}^{green} = G_{i,j}^{green} = Su^{j-i} - K$. If $G_{i+1,j+1} = G_{i+1,j+1}^{red} \geq G_{i+1,j+1}^{green} = G_{i,j}^{green}$ and  $G_{i+1,j+1} \leq G_{i,j}$, then $G_{i,j} = G_{i,j}^{red}$.
\end{proof}

\begin{corollary}
    \label{crl:movingBoundaryTri} 
    For every $i \in [0, T - 1]$, there exists an index $j_i \in [0, 2i]$ such that all cells $G_{i,j}$ with $0 \leq j \leq j_i$ are red and all (possibly zero) cells $G_{i,j}$ with $j_i < j \leq i$ are green. Also, for $i \in [0, T - 2]$, $j_{i+1} - 1 \leq j_i \leq j_{i+1}$.
\end{corollary}
\begin{proof}
 Follows from Lemma \ref{lma:rightOfGreenIsGreenTri}, Lemma \ref{lma:bottomOfGreenIsGreenTri}, and Lemma \ref{lma:diagLeftOfRedIsRedTri}.
\end{proof}

\begin{figure}[t!]
  \centering
  \includegraphics[width=\linewidth]{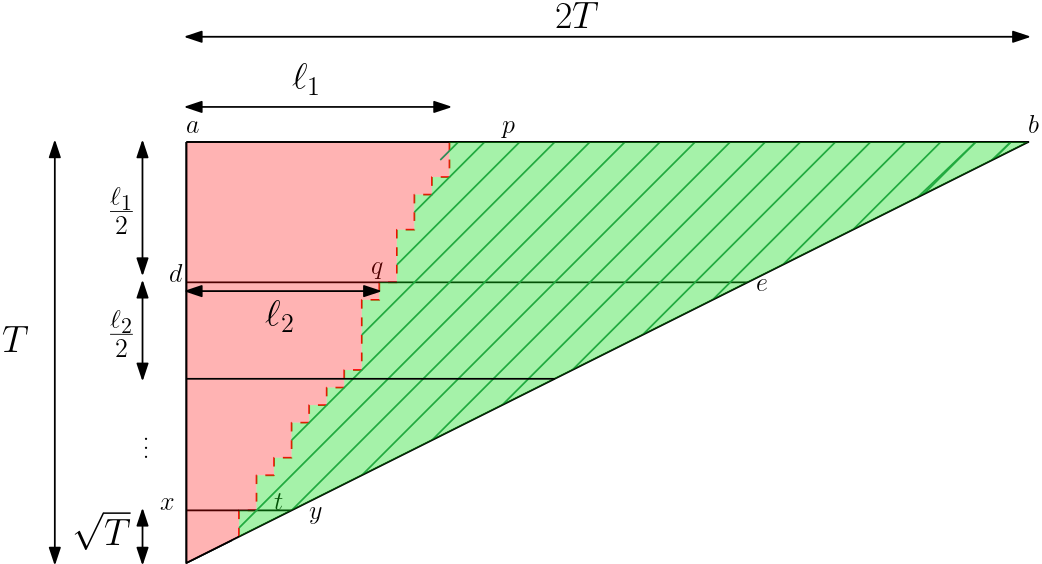}
  \caption{Partitioning the solution space into trapezoids.}
  \label{img:aoptiontrioverview}
\end{figure}

\begin{figure*}[t!]
    \centering
 \scalebox{1}[0.74]{
    \begin{subfigure}[t]{0.32\linewidth}
        \centering
        \includegraphics[width=\linewidth]{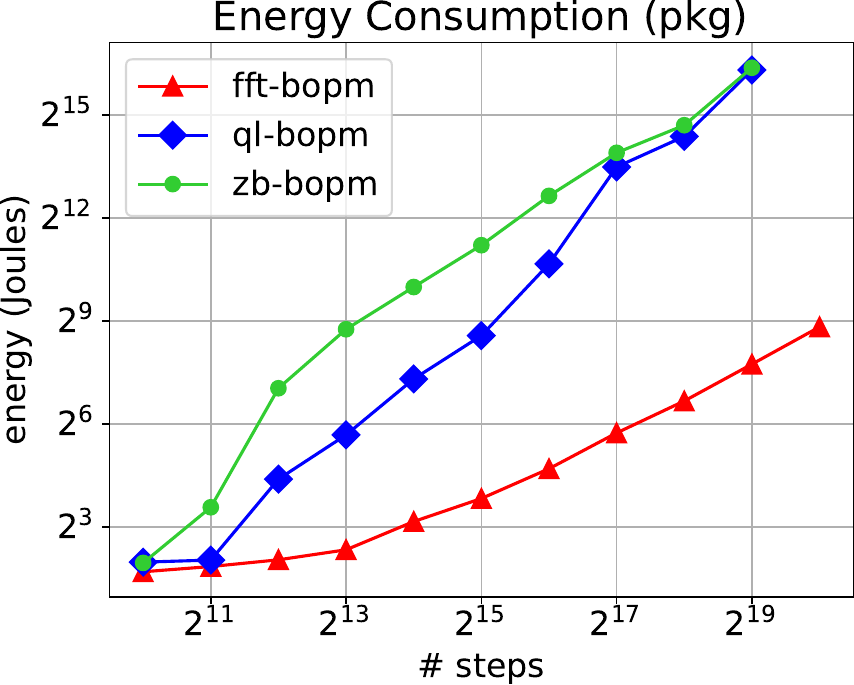}
        \label{fig:bopm-energy-pkg}
     \end{subfigure}
     \hfill
    \begin{subfigure}[t]{0.32\linewidth}
        \centering
        \includegraphics[width=\linewidth]{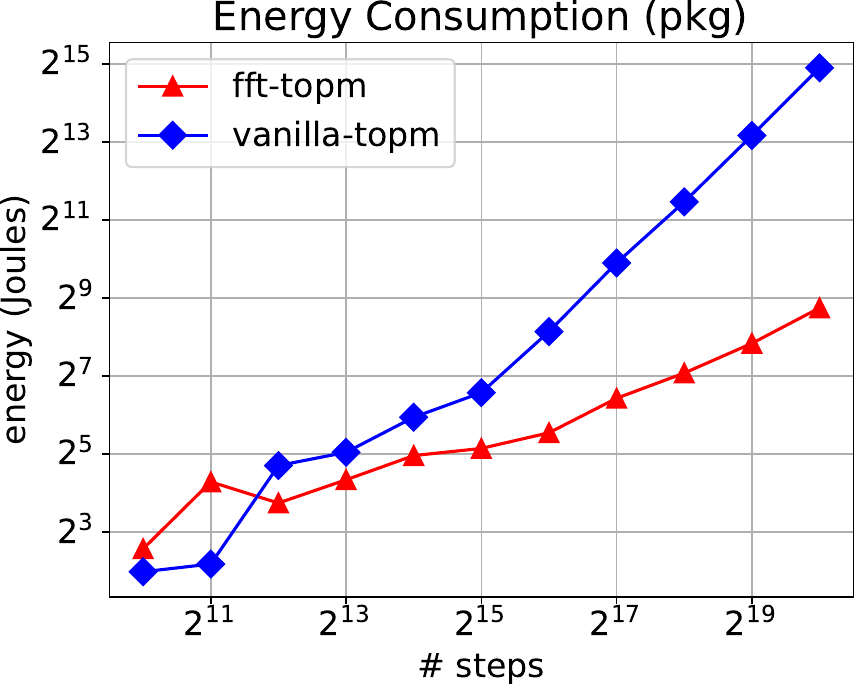}
        \label{fig:topm-energy-pkg}
     \end{subfigure}
     \hfill
    \begin{subfigure}[t]{0.32\linewidth}
        \centering
        \includegraphics[width=\linewidth]{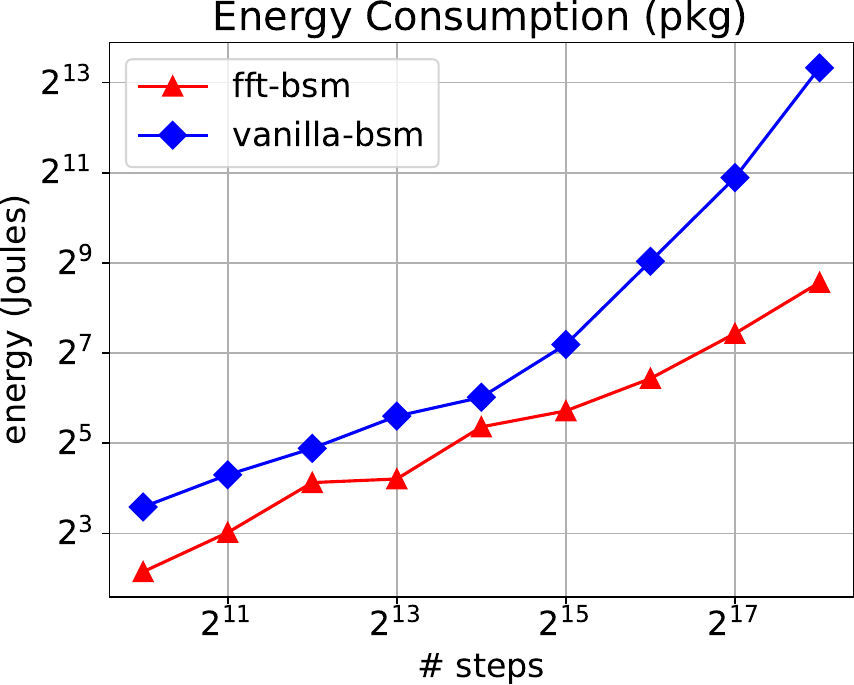}
        \label{fig:bsm-energy-pkg}
     \end{subfigure}
}     

\vspace{-0.2cm}
 \scalebox{1}[0.74]{
     \begin{subfigure}[t]{0.32\linewidth}
        \centering
        \includegraphics[width=\linewidth]{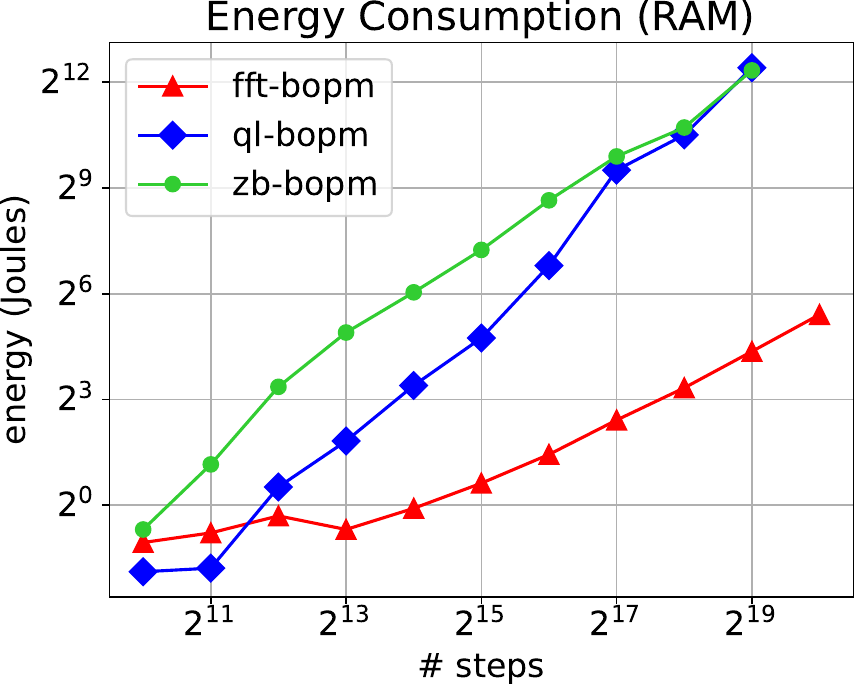}
        \caption{BOPM}
        \label{fig:bopm-energy-ram}
     \end{subfigure}
     \hfill
     \begin{subfigure}[t]{0.32\linewidth}
        \centering
        \includegraphics[width=\linewidth]{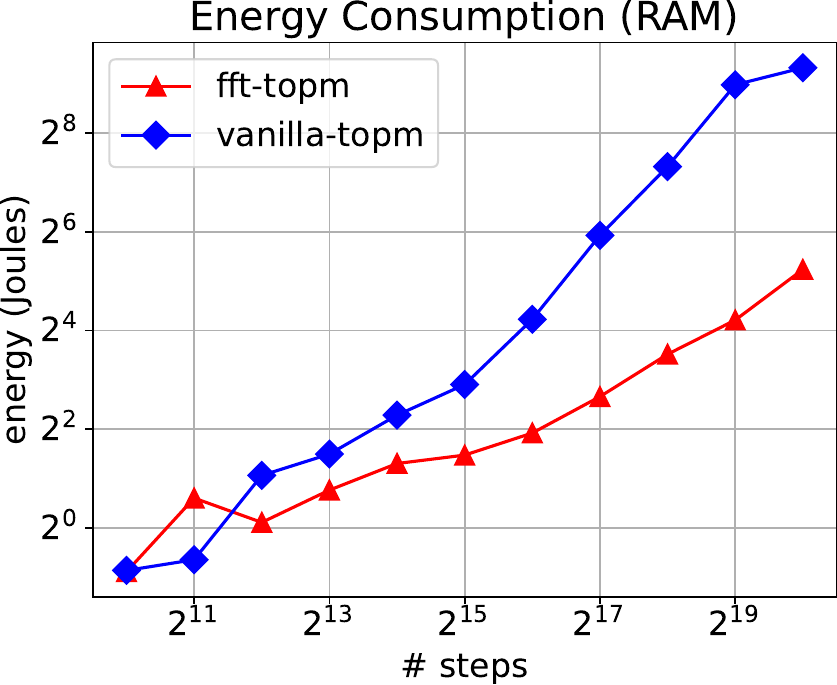}
        \caption{TOPM}
        \label{fig:topm-energy-ram}
     \end{subfigure}
     \hfill
     \begin{subfigure}[t]{0.32\linewidth}
        \centering
        \includegraphics[width=\linewidth]{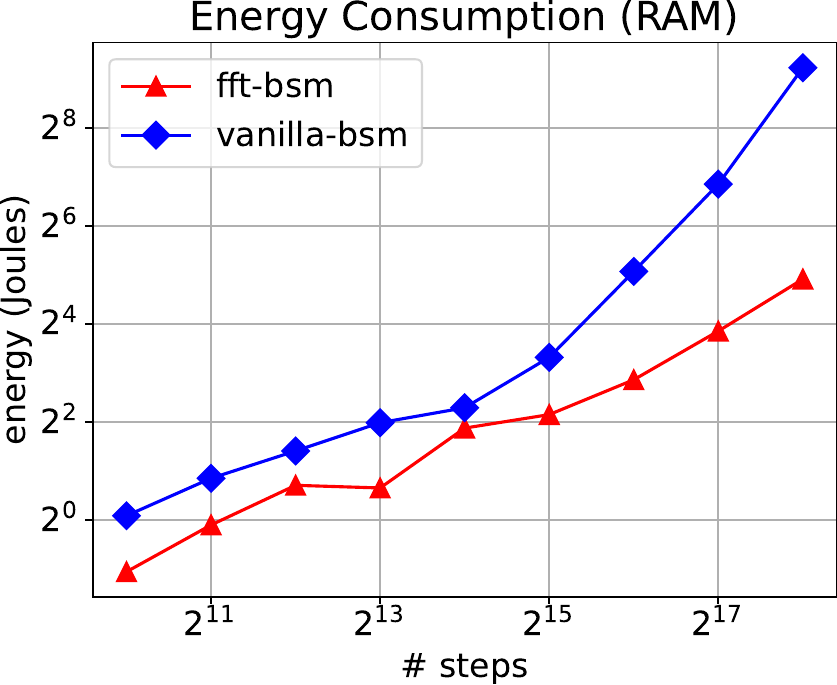}
        \caption{BSM}
        \label{fig:bsm-energy-ram}
     \end{subfigure}    
}     
\vspace{-0.4cm}
     \caption{Additional comparisons of energy consumption by domain, Package (pkg) and RAM.}
     \label{fig:energy-plot-more}
\end{figure*}

\subsection{Algorithm for American call option pricing under TOPM}
\label{ssec:TOPM-algo}

The solution space, is a right-angle triangle with a vertical base of length $T$ and a horizontal base of length $2T$. We prove the following theorem, which is an adaptation of the Binomial Model Algorithm.

\begin{theorem}
     \label{thm:TrinomParallel}
    Our algorithm solves the American call option pricing problem under TOPM in $\Oh{T\log^2 T}$ work and $\Oh{ T }$ span, where $T$ is the number of time steps.
\end{theorem}
\begin{proof}
    We use the exact same techniques as for BOPM. We again are using the properties of the boundary between the red and green cells to reduce the problem to a stencil-computation. The primary differences are:
    \begin{enumerate}
        \item In the usage of the FFT-based stencil algorithm of \cite{ahmad2021fast}, our stencil uses a larger (but still constant sized) neighborhood of grid values in the computation. Specifically, we have a dependence on the values $G_{i+1,j} G_{i+1,j+1},$ and $G_{i+1,j+2}$ to compute $G_{i,j}$, whereas in BOPM we depend only on $G_{i+1,j}$ and $ G_{i+1,j+1}$.
        \item The width of the grid we are operating on is twice that of BOPM. This only contributes a constant factor to runtime in terms of the worst-case for the width of the trapezoids. This means that some claims regarding the fact that the solution space is a right isosceles triangle need to be slightly adjusted.
    \end{enumerate}

    \subparagraph*{\textbf{Partitioning the triangle into trapezoids.}}
    Let $abc$ be the right-angle triangle representing the solution space as in Figure \ref{img:aoptiontrioverview}. We can partition $abc$ into a sequence of trapezoids just as with BOPM. We know the boundary at the row $i=T$, where $Su^{j-T} - K > 0$, so we let $\ell_1$ be the length from the $a$ to this boundary. Let $d$ be the point a distance $\ell_1/2$ downward from $a$. We draw a horizontal through $d$ and get the point $e$ from the intersection of the horizontal with the sloped side of $abc$. This defines a trapezoid $abde$. We then follow along the horizontal through $d$ to red-green divider to point q, to get a distance $\ell_2$, and repeat. We determine this boundary by solving the trapezoid corresponding to $\ell_1$. We continue repeating this process of partitioning until we are left with a triangle $xyc$ of height $\sqrt{T}$. 

    Once we have solved all of the trapezoids above $xyc$, we know that we can apply a naive $O(T^2)$ algorithm, which for $xyc$ (of height $\sqrt{T}$) will take a total of $\Oh{T}$ to solve, which will give us the final value for $i=0$. Solving the trapezoid of height $\ell_i$ will again take $\Oh{\ell_i \log ^2 \ell_i}$ time, which gives a total runtime $\psi$ of $\psi = \left( \sum _{1 \leq i \leq k} \Oh{\ell_i \log ^2 \ell_i} \right) + \Oh{T} = \Oh{T \log ^2 T}$.

    \subparagraph*{\textbf{Solving a trapezoid}}
        Solving a trapezoid is equivalent to solving the values of all red cells in the last row. Refer to figure \ref{img:aoptiontri}. We are given the solutions to the last row in the trapezoid above, which makes this possible. Let us say this trapezoid has height $\ell$, meaning that the boundary at the top of the trapezoid is $\ell$ away from the top left point $a$, landing on a point $q$. We again find the midpoint of $aq$, called $p$ and use this measure of $\ell/2$ to find a point below $a$, $r$. We again draw a horizontal line from $r$, intersecting $bc$ at $v$. We compute all the red cells on $dc$ (solving the trapezoid), in the same two steps:
        \begin{enumerate}
            \item Compute the cells on segment $rv$.
            \item Using the cells of $rv$, compute the cells of $dc$.
        \end{enumerate}

\begin{figure}[ht]
  \centering
  \includegraphics[width=\linewidth]{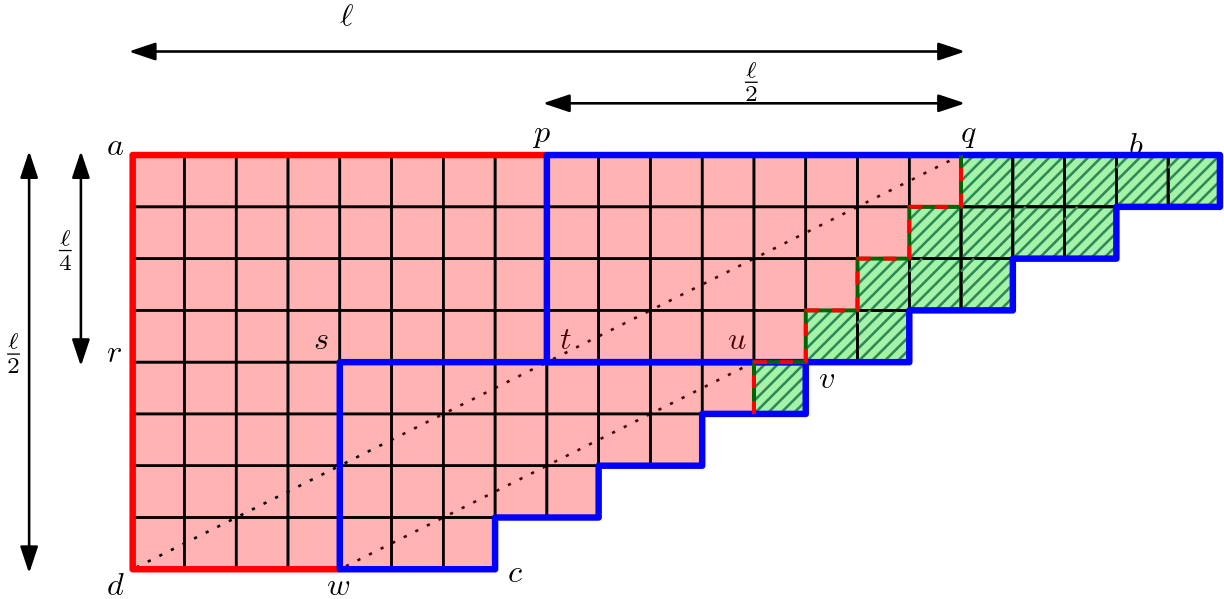}
  \caption{Decomposing a trapezoid into smaller trapezoids for trinomial.}
  \label{img:aoptiontri}
\end{figure}

    \subparagraph*{\textbf{Computing the values of all red cells on line segment $rv$}}

    Let $t$ be the intersection of $qd$ and $rv$. First, compute the cells of $rt$ using the FFT-based stencil algorithm of \cite{ahmad2021fast}, as all these cells are red and their values depend only on the red section of the top row, $aq$. We can then compute the second part, $tv$, by recursing on at the smaller trapezoid $pbvt$. Note that the height of $pbvt$ is half that of $abcd$.

    \subparagraph*{\textbf{Computing the values of all red cells on line segment $dc$}}
    Once the values of $rv$ have been calculated, we can use them to calculate the final row $dc$. From the calculation of $rv$, we know where the red-green boundary lies for that row. Let $u$ be that boundary point. We can take a line parallel to $qd$ through $u$ to determine a point $w$ along $dc$. all of the points $dw$ depend only on the values in $ru$, which have already been solved. We can again use the FFT-based stencil algorithm from \cite{ahmad2021fast}. The rest of $dc$, $wc$, can be recursed on. Take a vertical through $w$, and note its intersection with $rv$ as point $s$. $svcw$ forms a trapezoid that fits our specifications, which allows us to recurse on it to determine the rest of the red values along $dc$. Note: we do not need to calculate the green cells since they have a closed formula which can be evaluated at need.
    This gives us the recurrence  $\zeta(\ell) = 2 \zeta(\lceil \ell/2 \rceil) + \Oh{\ell \log \ell}$, which solves to $\zeta(\ell) = \ell \log ^2 \ell$, assuming the base case of $\ell = \Oh{1}$ has time complexity $\Oh{1}$, which we show next.

    \subparagraph*{\textbf{Base case.}}
    The base case entails a trapezoid of $\Oh{1}$ height. Such a trapezoid can have arbitrary width, but will have a constant number of red cells. Since we only wish to compute the red cells since the green can be computed at will in $\Oh{1}$ time, we will use a $\Oh{1}$ number of max operators to determine all the values of the red cells. Hence, the base case only takes $\Oh{1}$ time.

    What remains is to show the span. The span of computing one trapezoid of height $\ell$ depends on the time to compute the subtrapezoids and the time to compute the two stencils $rt$ and $dw$. The FFT-based stencil algorithm \cite{ahmad2021fast} runs in $\Oh{\log \ell \log \log \ell}$ span. Let $\zeta_\infty(\ell)$ be the span for a single trapezoid of height $\ell$.
    $$\zeta_\infty(\ell) = 2 \zeta_\infty(\lceil \ell/2 \rceil) + \Oh{\log \ell \log \log \ell}$$
    This solves to $\Oh{\ell}$. Let $\Psi_\infty$ be the total span of the algorithm. $\Psi = \Oh{\ell_1} + \Oh{\ell_2} + . ... \Oh{\ell _k} + \Oh{\sqrt{T}} = \Oh{T}$.
\end{proof}

\section{Energy Comsumption}

We give additional plots of energy consumption in Fig. \ref{fig:energy-plot-more}. In addition to outperforming existing approaches in terms of total energy consumption, the reduced energy consumption is pronounced when restricted to specific domains. Here we give the values for the package domain (\textit{pkg}) and the RAM domain as described in Section 5.2 of the main body. For both measures across all algorithms our algorithm outperforms vanilla and standard implementations after instances of size $2^13$. Our algorithm for the Black-Scholes-Merton Model was even more competitive, outperforming on instances of all sizes.



%% file: main.bbl

\begin{thebibliography}{127}


\ifx \showCODEN    \undefined \def \showCODEN     #1{\unskip}     \fi
\ifx \showDOI      \undefined \def \showDOI       #1{#1}\fi
\ifx \showISBNx    \undefined \def \showISBNx     #1{\unskip}     \fi
\ifx \showISBNxiii \undefined \def \showISBNxiii  #1{\unskip}     \fi
\ifx \showISSN     \undefined \def \showISSN      #1{\unskip}     \fi
\ifx \showLCCN     \undefined \def \showLCCN      #1{\unskip}     \fi
\ifx \shownote     \undefined \def \shownote      #1{#1}          \fi
\ifx \showarticletitle \undefined \def \showarticletitle #1{#1}   \fi
\ifx \showURL      \undefined \def \showURL       {\relax}        \fi
\providecommand\bibfield[2]{#2}
\providecommand\bibinfo[2]{#2}
\providecommand\natexlab[1]{#1}
\providecommand\showeprint[2][]{arXiv:#2}

\bibitem[Ahmad et~al\mbox{.}(2021)]%
        {ahmad2021fast}
\bibfield{author}{\bibinfo{person}{Zafar Ahmad}, \bibinfo{person}{Rezaul Chowdhury}, \bibinfo{person}{Rathish Das}, \bibinfo{person}{Pramod Ganapathi}, \bibinfo{person}{Aaron Gregory}, {and} \bibinfo{person}{Yimin Zhu}.} \bibinfo{year}{2021}\natexlab{}.
\newblock \showarticletitle{Fast stencil computations using fast {Fourier} transforms}. In \bibinfo{booktitle}{\emph{Proceedings of the 33rd ACM Symposium on Parallelism in Algorithms and Architectures}}. \bibinfo{pages}{8--21}.
\newblock


\bibitem[Alobaidi and Mallier(2001)]%
        {alobaidi2001asymptotic}
\bibfield{author}{\bibinfo{person}{Ghada Alobaidi} {and} \bibinfo{person}{Roland Mallier}.} \bibinfo{year}{2001}\natexlab{}.
\newblock \showarticletitle{Asymptotic analysis of American call options}.
\newblock \bibinfo{journal}{\emph{International Journal of Mathematics and Mathematical Sciences}} \bibinfo{volume}{27}, \bibinfo{number}{3} (\bibinfo{year}{2001}), \bibinfo{pages}{177--188}.
\newblock


\bibitem[Ames(2014)]%
        {ames2014numerical}
\bibfield{author}{\bibinfo{person}{William~F. Ames}.} \bibinfo{year}{2014}\natexlab{}.
\newblock \bibinfo{booktitle}{\emph{Numerical methods for partial differential equations}}.
\newblock \bibinfo{publisher}{Academic press}.
\newblock


\bibitem[Andonov and Rajopadhye(1997)]%
        {Andonov1997}
\bibfield{author}{\bibinfo{person}{Rumen Andonov} {and} \bibinfo{person}{Sanjay Rajopadhye}.} \bibinfo{year}{1997}\natexlab{}.
\newblock \showarticletitle{Optimal orthogonal tiling of 2-D iterations}.
\newblock \bibinfo{journal}{\emph{Journal of Parallel and Distributed computing}} \bibinfo{volume}{45}, \bibinfo{number}{2} (\bibinfo{year}{1997}), \bibinfo{pages}{159--165}.
\newblock


\bibitem[Atangana and Nieto(2015)]%
        {atangana2015numerical}
\bibfield{author}{\bibinfo{person}{Abdon Atangana} {and} \bibinfo{person}{Juan~J. Nieto}.} \bibinfo{year}{2015}\natexlab{}.
\newblock \showarticletitle{Numerical solution for the model of RLC circuit via the fractional derivative without singular kernel}.
\newblock \bibinfo{journal}{\emph{Advances in Mechanical Engineering}} \bibinfo{volume}{7}, \bibinfo{number}{10} (\bibinfo{year}{2015}), \bibinfo{pages}{1687814015613758}.
\newblock


\bibitem[Aubin et~al\mbox{.}(2004)]%
        {aubin2004modeling}
\bibfield{author}{\bibinfo{person}{Joelle Aubin}, \bibinfo{person}{David~F. Fletcher}, {and} \bibinfo{person}{Catherine Xuereb}.} \bibinfo{year}{2004}\natexlab{}.
\newblock \showarticletitle{Modeling turbulent flow in stirred tanks with CFD: the influence of the modeling approach, turbulence model and numerical scheme}.
\newblock \bibinfo{journal}{\emph{Experimental thermal and fluid science}} \bibinfo{volume}{28}, \bibinfo{number}{5} (\bibinfo{year}{2004}), \bibinfo{pages}{431--445}.
\newblock


\bibitem[Avissar and Pielke(1989)]%
        {avissar1989parameterization}
\bibfield{author}{\bibinfo{person}{Roni Avissar} {and} \bibinfo{person}{Roger~A Pielke}.} \bibinfo{year}{1989}\natexlab{}.
\newblock \showarticletitle{A parameterization of heterogeneous land surfaces for atmospheric numerical models and its impact on regional meteorology}.
\newblock \bibinfo{journal}{\emph{Monthly Weather Review}} \bibinfo{volume}{117}, \bibinfo{number}{10} (\bibinfo{year}{1989}), \bibinfo{pages}{2113--2136}.
\newblock


\bibitem[Bachelier(1900)]%
        {bachelier1900theorie}
\bibfield{author}{\bibinfo{person}{Louis Bachelier}.} \bibinfo{year}{1900}\natexlab{}.
\newblock \showarticletitle{Th{\'e}orie de la sp{\'e}culation}. In \bibinfo{booktitle}{\emph{Annales scientifiques de l'{\'E}cole normale sup{\'e}rieure}}, Vol.~\bibinfo{volume}{17}. \bibinfo{pages}{21--86}.
\newblock


\bibitem[Bakshi and Madan(2000)]%
        {bakshi2000spanning}
\bibfield{author}{\bibinfo{person}{Gurdip Bakshi} {and} \bibinfo{person}{Dilip Madan}.} \bibinfo{year}{2000}\natexlab{}.
\newblock \showarticletitle{Spanning and derivative-security valuation}.
\newblock \bibinfo{journal}{\emph{Journal of financial economics}} \bibinfo{volume}{55}, \bibinfo{number}{2} (\bibinfo{year}{2000}), \bibinfo{pages}{205--238}.
\newblock


\bibitem[Ball and Roma(1994)]%
        {ball1994stochastic}
\bibfield{author}{\bibinfo{person}{Clifford~A. Ball} {and} \bibinfo{person}{Antonio Roma}.} \bibinfo{year}{1994}\natexlab{}.
\newblock \showarticletitle{Stochastic volatility option pricing}.
\newblock \bibinfo{journal}{\emph{Journal of Financial and Quantitative Analysis}} \bibinfo{volume}{29}, \bibinfo{number}{4} (\bibinfo{year}{1994}), \bibinfo{pages}{589--607}.
\newblock


\bibitem[Bandishti et~al\mbox{.}(2012)]%
        {Bandishti2012}
\bibfield{author}{\bibinfo{person}{Vinayaka Bandishti}, \bibinfo{person}{Irshad Pananilath}, {and} \bibinfo{person}{Uday Bondhugula}.} \bibinfo{year}{2012}\natexlab{}.
\newblock \showarticletitle{Tiling stencil computations to maximize parallelism}. In \bibinfo{booktitle}{\emph{Proceedings of the International Conference on High Performance Computing, Networking, Storage and Analysis}}. \bibinfo{pages}{1--11}.
\newblock


\bibitem[Barth and Deconinck(2013)]%
        {barth2013high}
\bibfield{author}{\bibinfo{person}{Timothy~J. Barth} {and} \bibinfo{person}{Herman Deconinck}.} \bibinfo{year}{2013}\natexlab{}.
\newblock \bibinfo{booktitle}{\emph{High-order Methods for Computational Physics}}. Vol.~\bibinfo{volume}{9}.
\newblock \bibinfo{publisher}{Springer Science \& Business Media}.
\newblock


\bibitem[Bensoussan(1984)]%
        {bensoussan1984theory}
\bibfield{author}{\bibinfo{person}{Alain Bensoussan}.} \bibinfo{year}{1984}\natexlab{}.
\newblock \showarticletitle{On the theory of option pricing}.
\newblock \bibinfo{journal}{\emph{Acta Applicandae Mathematica}} \bibinfo{volume}{2}, \bibinfo{number}{2} (\bibinfo{year}{1984}), \bibinfo{pages}{139--158}.
\newblock


\bibitem[Bergomi(2015)]%
        {bergomi2015stochastic}
\bibfield{author}{\bibinfo{person}{Lorenzo Bergomi}.} \bibinfo{year}{2015}\natexlab{}.
\newblock \bibinfo{booktitle}{\emph{Stochastic volatility modeling}}.
\newblock \bibinfo{publisher}{CRC press}.
\newblock


\bibitem[Black and Scholes(1973)]%
        {black1973pricing}
\bibfield{author}{\bibinfo{person}{Fischer Black} {and} \bibinfo{person}{Myron Scholes}.} \bibinfo{year}{1973}\natexlab{}.
\newblock \showarticletitle{The pricing of options and corporate liabilities}.
\newblock \bibinfo{journal}{\emph{Journal of political economy}} \bibinfo{volume}{81}, \bibinfo{number}{3} (\bibinfo{year}{1973}), \bibinfo{pages}{637--654}.
\newblock


\bibitem[Bondhugula et~al\mbox{.}(2016)]%
        {bondhugula2016plutoplus}
\bibfield{author}{\bibinfo{person}{Uday Bondhugula}, \bibinfo{person}{Aravind Acharya}, {and} \bibinfo{person}{Albert Cohen}.} \bibinfo{year}{2016}\natexlab{}.
\newblock \showarticletitle{The Pluto+ algorithm: A practical approach for parallelization and locality optimization of affine loop nests}.
\newblock \bibinfo{journal}{\emph{ACM Transactions on Programming Languages and Systems}} \bibinfo{volume}{38}, \bibinfo{number}{3} (\bibinfo{year}{2016}), \bibinfo{pages}{1--32}.
\newblock


\bibitem[Bondhugula et~al\mbox{.}(2017)]%
        {bondhugula2017}
\bibfield{author}{\bibinfo{person}{Uday Bondhugula}, \bibinfo{person}{Vinayaka Bandishti}, {and} \bibinfo{person}{Irshad Pananilath}.} \bibinfo{year}{2017}\natexlab{}.
\newblock \showarticletitle{Diamond tiling: Tiling techniques to maximize parallelism for stencil computations}.
\newblock \bibinfo{journal}{\emph{IEEE Transactions on Parallel and Distributed Systems}} \bibinfo{volume}{28}, \bibinfo{number}{5} (\bibinfo{year}{2017}), \bibinfo{pages}{1285--1298}.
\newblock


\bibitem[Boyle(1986)]%
        {boyle1986option}
\bibfield{author}{\bibinfo{person}{Phelim Boyle}.} \bibinfo{year}{1986}\natexlab{}.
\newblock \showarticletitle{Option Valuation Using a Three-Jump Process}.
\newblock \bibinfo{journal}{\emph{International Options Journal}}  \bibinfo{volume}{3} (\bibinfo{year}{1986}), \bibinfo{pages}{7--12}.
\newblock


\bibitem[Boyle(1977)]%
        {boyle1977options}
\bibfield{author}{\bibinfo{person}{Phelim~P. Boyle}.} \bibinfo{year}{1977}\natexlab{}.
\newblock \showarticletitle{Options: A monte carlo approach}.
\newblock \bibinfo{journal}{\emph{Journal of financial economics}} \bibinfo{volume}{4}, \bibinfo{number}{3} (\bibinfo{year}{1977}), \bibinfo{pages}{323--338}.
\newblock


\bibitem[Brunelle(2022)]%
        {brunelle2022parallelizing}
\bibfield{author}{\bibinfo{person}{Terryn Brunelle}.} \bibinfo{year}{2022}\natexlab{}.
\newblock \emph{\bibinfo{title}{Parallelizing Tree Traversals for Binomial Option Pricing}}.
\newblock \bibinfo{thesistype}{Ph.\,D. Dissertation}. \bibinfo{school}{Massachusetts Institute of Technology}.
\newblock


\bibitem[Burden et~al\mbox{.}(2015)]%
        {burden2015numerical}
\bibfield{author}{\bibinfo{person}{Richard~L. Burden}, \bibinfo{person}{Douglas~J. Faires}, {and} \bibinfo{person}{Annette~M. Burden}.} \bibinfo{year}{2015}\natexlab{}.
\newblock \bibinfo{booktitle}{\emph{Numerical analysis}}.
\newblock \bibinfo{publisher}{Cengage learning}.
\newblock


\bibitem[Chang et~al\mbox{.}(2007)]%
        {chang2007richardson}
\bibfield{author}{\bibinfo{person}{Chuang-Chang Chang}, \bibinfo{person}{San-Lin Chung}, {and} \bibinfo{person}{Richard~C. Stapleton}.} \bibinfo{year}{2007}\natexlab{}.
\newblock \showarticletitle{Richardson extrapolation techniques for the pricing of American-style options}.
\newblock \bibinfo{journal}{\emph{Journal of Futures Markets: Futures, Options, and Other Derivative Products}} \bibinfo{volume}{27}, \bibinfo{number}{8} (\bibinfo{year}{2007}), \bibinfo{pages}{791--817}.
\newblock


\bibitem[Chen and Chadam(2007)]%
        {chen2007mathematical}
\bibfield{author}{\bibinfo{person}{Xinfu Chen} {and} \bibinfo{person}{John Chadam}.} \bibinfo{year}{2007}\natexlab{}.
\newblock \showarticletitle{A mathematical analysis of the optimal exercise boundary for American put options}.
\newblock \bibinfo{journal}{\emph{SIAM Journal on Mathematical Analysis}} \bibinfo{volume}{38}, \bibinfo{number}{5} (\bibinfo{year}{2007}), \bibinfo{pages}{1613--1641}.
\newblock


\bibitem[Chen and Forsyth(2008)]%
        {chen2008numerical}
\bibfield{author}{\bibinfo{person}{Zhuliang Chen} {and} \bibinfo{person}{Peter~A. Forsyth}.} \bibinfo{year}{2008}\natexlab{}.
\newblock \showarticletitle{A numerical scheme for the impulse control formulation for pricing variable annuities with a guaranteed minimum withdrawal benefit (GMWB)}.
\newblock \bibinfo{journal}{\emph{Numer. Math.}} \bibinfo{volume}{109}, \bibinfo{number}{4} (\bibinfo{year}{2008}), \bibinfo{pages}{535--569}.
\newblock


\bibitem[Cont and Voltchkova(2005)]%
        {cont2005finite}
\bibfield{author}{\bibinfo{person}{Rama Cont} {and} \bibinfo{person}{Ekaterina Voltchkova}.} \bibinfo{year}{2005}\natexlab{}.
\newblock \showarticletitle{A finite difference scheme for option pricing in jump diffusion and exponential L{\'e}vy models}.
\newblock \bibinfo{journal}{\emph{SIAM J. Numer. Anal.}} \bibinfo{volume}{43}, \bibinfo{number}{4} (\bibinfo{year}{2005}), \bibinfo{pages}{1596--1626}.
\newblock


\bibitem[contributors(2022)]%
        {quantlib}
\bibfield{author}{\bibinfo{person}{The~QuantLib contributors}.} \bibinfo{year}{2022}\natexlab{}.
\newblock \bibinfo{booktitle}{\emph{{QuantLib: a free/open-source library for quantitative finance}}}.
\newblock
\urldef\tempurl%
\url{https://doi.org/10.5281/zenodo.6461652}
\showDOI{\tempurl}
\newblock
\shownote{{If you use this software, please cite it using these metadata.}}.


\bibitem[Cormen et~al\mbox{.}(2009)]%
        {CormenLeRiSt2009}
\bibfield{author}{\bibinfo{person}{Thomas~H Cormen}, \bibinfo{person}{Charles~E Leiserson}, \bibinfo{person}{Ronald~L Rivest}, {and} \bibinfo{person}{Clifford Stein}.} \bibinfo{year}{2009}\natexlab{}.
\newblock \bibinfo{booktitle}{\emph{Introduction to Algorithms}}.
\newblock \bibinfo{publisher}{MIT Press}.
\newblock


\bibitem[Cox et~al\mbox{.}(1979)]%
        {cox1979option}
\bibfield{author}{\bibinfo{person}{John~C. Cox}, \bibinfo{person}{Stephen~A. Ross}, {and} \bibinfo{person}{Mark Rubinstein}.} \bibinfo{year}{1979}\natexlab{}.
\newblock \showarticletitle{Option pricing: A simplified approach}.
\newblock \bibinfo{journal}{\emph{Journal of financial Economics}} \bibinfo{volume}{7}, \bibinfo{number}{3} (\bibinfo{year}{1979}), \bibinfo{pages}{229--263}.
\newblock


\bibitem[Crank and Nicolson(1947)]%
        {crank1947practical}
\bibfield{author}{\bibinfo{person}{John Crank} {and} \bibinfo{person}{Phyllis Nicolson}.} \bibinfo{year}{1947}\natexlab{}.
\newblock \showarticletitle{A practical method for numerical evaluation of solutions of partial differential equations of the heat-conduction type}. In \bibinfo{booktitle}{\emph{Mathematical proceedings of the Cambridge philosophical society}}, Vol.~\bibinfo{volume}{43}. Cambridge University Press, \bibinfo{pages}{50--67}.
\newblock


\bibitem[Cundall and Hart(1992)]%
        {cundall1992numerical}
\bibfield{author}{\bibinfo{person}{Peter~A. Cundall} {and} \bibinfo{person}{Roger~D. Hart}.} \bibinfo{year}{1992}\natexlab{}.
\newblock \showarticletitle{Numerical modelling of discontinua}.
\newblock \bibinfo{journal}{\emph{Engineering computations}} (\bibinfo{year}{1992}).
\newblock


\bibitem[Dempster and Hong(2002)]%
        {dempster2002spread}
\bibfield{author}{\bibinfo{person}{Michael A.~H. Dempster} {and} \bibinfo{person}{George S.~S. Hong}.} \bibinfo{year}{2002}\natexlab{}.
\newblock \showarticletitle{Spread option valuation and the fast Fourier transform}. In \bibinfo{booktitle}{\emph{Mathematical Finance—Bachelier Congress 2000}}. Springer, \bibinfo{pages}{203--220}.
\newblock


\bibitem[Derman and Kani(1994)]%
        {derman1994riding}
\bibfield{author}{\bibinfo{person}{Emanuel Derman} {and} \bibinfo{person}{Iraj Kani}.} \bibinfo{year}{1994}\natexlab{}.
\newblock \showarticletitle{Riding on a smile}.
\newblock \bibinfo{journal}{\emph{Risk}} \bibinfo{volume}{7}, \bibinfo{number}{2} (\bibinfo{year}{1994}), \bibinfo{pages}{32--39}.
\newblock


\bibitem[Derman and Kani(1998)]%
        {derman1998stochastic}
\bibfield{author}{\bibinfo{person}{Emanuel Derman} {and} \bibinfo{person}{Iraj Kani}.} \bibinfo{year}{1998}\natexlab{}.
\newblock \showarticletitle{Stochastic implied trees: Arbitrage pricing with stochastic term and strike structure of volatility}.
\newblock \bibinfo{journal}{\emph{International journal of theoretical and applied finance}} \bibinfo{volume}{1}, \bibinfo{number}{01} (\bibinfo{year}{1998}), \bibinfo{pages}{61--110}.
\newblock


\bibitem[Duan and Simonato(1998)]%
        {duan1998empirical}
\bibfield{author}{\bibinfo{person}{Jin-Chuan Duan} {and} \bibinfo{person}{Jean-Guy Simonato}.} \bibinfo{year}{1998}\natexlab{}.
\newblock \showarticletitle{Empirical martingale simulation for asset prices}.
\newblock \bibinfo{journal}{\emph{Management Science}} \bibinfo{volume}{44}, \bibinfo{number}{9} (\bibinfo{year}{1998}), \bibinfo{pages}{1218--1233}.
\newblock


\bibitem[Duffie et~al\mbox{.}(2000)]%
        {duffie2000transform}
\bibfield{author}{\bibinfo{person}{Darrell Duffie}, \bibinfo{person}{Jun Pan}, {and} \bibinfo{person}{Kenneth Singleton}.} \bibinfo{year}{2000}\natexlab{}.
\newblock \showarticletitle{Transform analysis and asset pricing for affine jump-diffusions}.
\newblock \bibinfo{journal}{\emph{Econometrica}} \bibinfo{volume}{68}, \bibinfo{number}{6} (\bibinfo{year}{2000}), \bibinfo{pages}{1343--1376}.
\newblock


\bibitem[Dupire(1994)]%
        {dupire1994pricing}
\bibfield{author}{\bibinfo{person}{Bruno Dupire}.} \bibinfo{year}{1994}\natexlab{}.
\newblock \showarticletitle{Pricing with a smile}.
\newblock \bibinfo{journal}{\emph{Risk}} \bibinfo{volume}{7}, \bibinfo{number}{1} (\bibinfo{year}{1994}), \bibinfo{pages}{18--20}.
\newblock


\bibitem[Frigo and Strumpen(2005)]%
        {frigo2005cache}
\bibfield{author}{\bibinfo{person}{Matteo Frigo} {and} \bibinfo{person}{Volker Strumpen}.} \bibinfo{year}{2005}\natexlab{}.
\newblock \showarticletitle{Cache oblivious stencil computations}. In \bibinfo{booktitle}{\emph{Proceedings of the 19th International Conference on Supercomputing}}. \bibinfo{pages}{361--366}.
\newblock


\bibitem[Frigo and Strumpen(2009)]%
        {FrigoSt2009}
\bibfield{author}{\bibinfo{person}{Matteo Frigo} {and} \bibinfo{person}{Volker Strumpen}.} \bibinfo{year}{2009}\natexlab{}.
\newblock \showarticletitle{The cache complexity of multithreaded cache oblivious algorithms}.
\newblock \bibinfo{journal}{\emph{Theory of Computing Systems}} \bibinfo{volume}{45}, \bibinfo{number}{2} (\bibinfo{year}{2009}), \bibinfo{pages}{203--233}.
\newblock


\bibitem[Galai and Masulis(1976)]%
        {galai1976option}
\bibfield{author}{\bibinfo{person}{Dan Galai} {and} \bibinfo{person}{Ronald~W. Masulis}.} \bibinfo{year}{1976}\natexlab{}.
\newblock \showarticletitle{The option pricing model and the risk factor of stock}.
\newblock \bibinfo{journal}{\emph{Journal of Financial economics}} \bibinfo{volume}{3}, \bibinfo{number}{1-2} (\bibinfo{year}{1976}), \bibinfo{pages}{53--81}.
\newblock


\bibitem[Gammie et~al\mbox{.}(2003)]%
        {gammie2003harm}
\bibfield{author}{\bibinfo{person}{Charles~F. Gammie}, \bibinfo{person}{Jonathan~C. McKinney}, {and} \bibinfo{person}{G{\'a}bor T{\'o}th}.} \bibinfo{year}{2003}\natexlab{}.
\newblock \showarticletitle{HARM: a numerical scheme for general relativistic magnetohydrodynamics}.
\newblock \bibinfo{journal}{\emph{The Astrophysical Journal}} \bibinfo{volume}{589}, \bibinfo{number}{1} (\bibinfo{year}{2003}), \bibinfo{pages}{444}.
\newblock


\bibitem[Geman et~al\mbox{.}(2001)]%
        {geman2001time}
\bibfield{author}{\bibinfo{person}{H{\'e}lyette Geman}, \bibinfo{person}{Dilip~B. Madan}, {and} \bibinfo{person}{Marc Yor}.} \bibinfo{year}{2001}\natexlab{}.
\newblock \showarticletitle{Time changes for L{\'e}vy processes}.
\newblock \bibinfo{journal}{\emph{Mathematical Finance}} \bibinfo{volume}{11}, \bibinfo{number}{1} (\bibinfo{year}{2001}), \bibinfo{pages}{79--96}.
\newblock


\bibitem[Han and Wang(2015)]%
        {han2015second}
\bibfield{author}{\bibinfo{person}{Daozhi Han} {and} \bibinfo{person}{Xiaoming Wang}.} \bibinfo{year}{2015}\natexlab{}.
\newblock \showarticletitle{A second order in time, uniquely solvable, unconditionally stable numerical scheme for Cahn--Hilliard--Navier--Stokes equation}.
\newblock \bibinfo{journal}{\emph{J. Comput. Phys.}}  \bibinfo{volume}{290} (\bibinfo{year}{2015}), \bibinfo{pages}{139--156}.
\newblock


\bibitem[Heston(1993)]%
        {heston1993closed}
\bibfield{author}{\bibinfo{person}{Steven~L. Heston}.} \bibinfo{year}{1993}\natexlab{}.
\newblock \showarticletitle{A closed-form solution for options with stochastic volatility with applications to bond and currency options}.
\newblock \bibinfo{journal}{\emph{The review of financial studies}} \bibinfo{volume}{6}, \bibinfo{number}{2} (\bibinfo{year}{1993}), \bibinfo{pages}{327--343}.
\newblock


\bibitem[Hildebrand(1987)]%
        {hildebrand1987introduction}
\bibfield{author}{\bibinfo{person}{Francis~B. Hildebrand}.} \bibinfo{year}{1987}\natexlab{}.
\newblock \bibinfo{booktitle}{\emph{Introduction to numerical analysis}}.
\newblock \bibinfo{publisher}{Courier Corporation}.
\newblock


\bibitem[Hirouchi et~al\mbox{.}(2009)]%
        {hirouchi2009development}
\bibfield{author}{\bibinfo{person}{Tomoyuki Hirouchi}, \bibinfo{person}{Tomohiro Takaki}, {and} \bibinfo{person}{Yoshihiro Tomita}.} \bibinfo{year}{2009}\natexlab{}.
\newblock \showarticletitle{Development of numerical scheme for phase field crystal deformation simulation}.
\newblock \bibinfo{journal}{\emph{Computational materials science}} \bibinfo{volume}{44}, \bibinfo{number}{4} (\bibinfo{year}{2009}), \bibinfo{pages}{1192--1197}.
\newblock


\bibitem[H{\"o}gstedt et~al\mbox{.}(1999)]%
        {Hogstedt1999}
\bibfield{author}{\bibinfo{person}{Karin H{\"o}gstedt}, \bibinfo{person}{Larry Carter}, {and} \bibinfo{person}{Jeanne Ferrante}.} \bibinfo{year}{1999}\natexlab{}.
\newblock \showarticletitle{Selecting tile shape for minimal execution time}. In \bibinfo{booktitle}{\emph{ACM Symposium on Parallel algorithms and architectures}}. \bibinfo{pages}{201--211}.
\newblock


\bibitem[Holmes(2007)]%
        {holmes2007introduction}
\bibfield{author}{\bibinfo{person}{Mark~H. Holmes}.} \bibinfo{year}{2007}\natexlab{}.
\newblock \bibinfo{booktitle}{\emph{Introduction to numerical methods in differential equations}}.
\newblock \bibinfo{publisher}{Springer}.
\newblock


\bibitem[Hull(2003)]%
        {hull2003options}
\bibfield{author}{\bibinfo{person}{John~C. Hull}.} \bibinfo{year}{2003}\natexlab{}.
\newblock \bibinfo{booktitle}{\emph{Options futures and other derivatives} (\bibinfo{edition}{5th} ed.)}.
\newblock \bibinfo{publisher}{Pearson Education India}.
\newblock


\bibitem[Ibanez and Zapatero(2004)]%
        {ibanez2004monte}
\bibfield{author}{\bibinfo{person}{Alfredo Ibanez} {and} \bibinfo{person}{Fernando Zapatero}.} \bibinfo{year}{2004}\natexlab{}.
\newblock \showarticletitle{Monte Carlo valuation of American options through computation of the optimal exercise frontier}.
\newblock \bibinfo{journal}{\emph{Journal of Financial and Quantitative Analysis}} \bibinfo{volume}{39}, \bibinfo{number}{2} (\bibinfo{year}{2004}), \bibinfo{pages}{253--275}.
\newblock


\bibitem[Jackson et~al\mbox{.}(2008)]%
        {jackson2008fourier}
\bibfield{author}{\bibinfo{person}{Kenneth~R. Jackson}, \bibinfo{person}{Sebastian Jaimungal}, {and} \bibinfo{person}{Vladimir Surkov}.} \bibinfo{year}{2008}\natexlab{}.
\newblock \showarticletitle{Fourier space time-stepping for option pricing with L{\'e}vy models}.
\newblock \bibinfo{journal}{\emph{Journal of Computational Finance}} \bibinfo{volume}{12}, \bibinfo{number}{2} (\bibinfo{year}{2008}), \bibinfo{pages}{1--29}.
\newblock


\bibitem[James(1980)]%
        {james1980monte}
\bibfield{author}{\bibinfo{person}{Frederick James}.} \bibinfo{year}{1980}\natexlab{}.
\newblock \showarticletitle{Monte Carlo theory and practice}.
\newblock \bibinfo{journal}{\emph{Reports on progress in Physics}} \bibinfo{volume}{43}, \bibinfo{number}{9} (\bibinfo{year}{1980}), \bibinfo{pages}{1145}.
\newblock


\bibitem[Johnson(2010)]%
        {johnson2010numerical}
\bibfield{author}{\bibinfo{person}{Martin~T. Johnson}.} \bibinfo{year}{2010}\natexlab{}.
\newblock \showarticletitle{A numerical scheme to calculate temperature and salinity dependent air-water transfer velocities for any gas}.
\newblock \bibinfo{journal}{\emph{Ocean Science}} \bibinfo{volume}{6}, \bibinfo{number}{4} (\bibinfo{year}{2010}), \bibinfo{pages}{913--932}.
\newblock


\bibitem[Kalnay et~al\mbox{.}(1990)]%
        {kalnay1990global}
\bibfield{author}{\bibinfo{person}{Eugenia Kalnay}, \bibinfo{person}{Masao Kanamitsu}, {and} \bibinfo{person}{Wayman~E. Baker}.} \bibinfo{year}{1990}\natexlab{}.
\newblock \showarticletitle{Global numerical weather prediction at the National Meteorological Center}.
\newblock \bibinfo{journal}{\emph{Bulletin of the American Meteorological Society}} \bibinfo{volume}{71}, \bibinfo{number}{10} (\bibinfo{year}{1990}), \bibinfo{pages}{1410--1428}.
\newblock


\bibitem[Karatzas and Shreve(1998)]%
        {karatzas1998methods}
\bibfield{author}{\bibinfo{person}{Ioannis Karatzas} {and} \bibinfo{person}{Steven~E. Shreve}.} \bibinfo{year}{1998}\natexlab{}.
\newblock \bibinfo{booktitle}{\emph{Methods of mathematical finance}}. Vol.~\bibinfo{volume}{39}.
\newblock \bibinfo{publisher}{Springer}.
\newblock


\bibitem[Komissarov(2002)]%
        {komissarov2002time}
\bibfield{author}{\bibinfo{person}{Serguei~S. Komissarov}.} \bibinfo{year}{2002}\natexlab{}.
\newblock \showarticletitle{Time-dependent, force-free, degenerate electrodynamics}.
\newblock \bibinfo{journal}{\emph{Monthly Notices of the Royal Astronomical Society}} \bibinfo{volume}{336}, \bibinfo{number}{3} (\bibinfo{year}{2002}), \bibinfo{pages}{759--766}.
\newblock


\bibitem[Kou(2002)]%
        {kou2002jump}
\bibfield{author}{\bibinfo{person}{Steven~G. Kou}.} \bibinfo{year}{2002}\natexlab{}.
\newblock \showarticletitle{A jump-diffusion model for option pricing}.
\newblock \bibinfo{journal}{\emph{Management science}} \bibinfo{volume}{48}, \bibinfo{number}{8} (\bibinfo{year}{2002}), \bibinfo{pages}{1086--1101}.
\newblock


\bibitem[Krishnamoorthy et~al\mbox{.}(2007)]%
        {Sriram2007}
\bibfield{author}{\bibinfo{person}{Sriram Krishnamoorthy}, \bibinfo{person}{Muthu Baskaran}, \bibinfo{person}{Uday Bondhugula}, \bibinfo{person}{Jagannathan Ramanujam}, \bibinfo{person}{Atanas Rountev}, {and} \bibinfo{person}{Ponnuswamy Sadayappan}.} \bibinfo{year}{2007}\natexlab{}.
\newblock \showarticletitle{Effective automatic parallelization of stencil computations}.
\newblock \bibinfo{journal}{\emph{ACM sigplan notices}} \bibinfo{volume}{42}, \bibinfo{number}{6} (\bibinfo{year}{2007}), \bibinfo{pages}{235--244}.
\newblock


\bibitem[Kumar et~al\mbox{.}(2012)]%
        {kumar2012analytical}
\bibfield{author}{\bibinfo{person}{Sunil Kumar}, \bibinfo{person}{Ahmet Yildirim}, \bibinfo{person}{Yasir Khan}, \bibinfo{person}{Hossein Jafari}, \bibinfo{person}{Khosro Sayevand}, {and} \bibinfo{person}{Leilei Wei}.} \bibinfo{year}{2012}\natexlab{}.
\newblock \showarticletitle{Analytical solution of fractional Black-Scholes European option pricing equation by using Laplace transform}.
\newblock \bibinfo{journal}{\emph{Journal of fractional calculus and Applications}} \bibinfo{volume}{2}, \bibinfo{number}{8} (\bibinfo{year}{2012}), \bibinfo{pages}{1--9}.
\newblock


\bibitem[Langat et~al\mbox{.}(2019)]%
        {LangatMwanikiKiprop}
\bibfield{author}{\bibinfo{person}{Kenneth~Kiprotich Langat}, \bibinfo{person}{Joseph~Ivivi Mwaniki}, {and} \bibinfo{person}{George~Korir Kiprop}.} \bibinfo{year}{2019}\natexlab{}.
\newblock \showarticletitle{Pricing options using trinomial lattice method}.
\newblock \bibinfo{journal}{\emph{Journal of Finance and Economics}} \bibinfo{volume}{7}, \bibinfo{number}{3} (\bibinfo{year}{2019}), \bibinfo{pages}{81--87}.
\newblock


\bibitem[LeVeque(2007)]%
        {leveque2007finite}
\bibfield{author}{\bibinfo{person}{Randall~J. LeVeque}.} \bibinfo{year}{2007}\natexlab{}.
\newblock \bibinfo{booktitle}{\emph{Finite difference methods for ordinary and partial differential equations: steady-state and time-dependent problems}}.
\newblock \bibinfo{publisher}{SIAM}.
\newblock


\bibitem[Long et~al\mbox{.}(2008)]%
        {long2008numerical}
\bibfield{author}{\bibinfo{person}{Wen Long}, \bibinfo{person}{James~T. Kirby}, {and} \bibinfo{person}{Zhiyu Shao}.} \bibinfo{year}{2008}\natexlab{}.
\newblock \showarticletitle{A numerical scheme for morphological bed level calculations}.
\newblock \bibinfo{journal}{\emph{Coastal Engineering}} \bibinfo{volume}{55}, \bibinfo{number}{2} (\bibinfo{year}{2008}), \bibinfo{pages}{167--180}.
\newblock


\bibitem[Longstaff and Schwartz(2001)]%
        {longstaff2001valuing}
\bibfield{author}{\bibinfo{person}{Francis~A. Longstaff} {and} \bibinfo{person}{Eduardo~S. Schwartz}.} \bibinfo{year}{2001}\natexlab{}.
\newblock \showarticletitle{Valuing American options by simulation: a simple least-squares approach}.
\newblock \bibinfo{journal}{\emph{The review of financial studies}} \bibinfo{volume}{14}, \bibinfo{number}{1} (\bibinfo{year}{2001}), \bibinfo{pages}{113--147}.
\newblock


\bibitem[Lord and Rougemont(2004)]%
        {lord2004numerical}
\bibfield{author}{\bibinfo{person}{Gabriel~J. Lord} {and} \bibinfo{person}{Jacques Rougemont}.} \bibinfo{year}{2004}\natexlab{}.
\newblock \showarticletitle{A numerical scheme for stochastic PDEs with Gevrey regularity}.
\newblock \bibinfo{journal}{\emph{IMA journal of numerical analysis}} \bibinfo{volume}{24}, \bibinfo{number}{4} (\bibinfo{year}{2004}), \bibinfo{pages}{587--604}.
\newblock


\bibitem[Lord et~al\mbox{.}(2008)]%
        {lord2008fast}
\bibfield{author}{\bibinfo{person}{Roger Lord}, \bibinfo{person}{Fang Fang}, \bibinfo{person}{Frank Bervoets}, {and} \bibinfo{person}{Cornelis~W. Oosterlee}.} \bibinfo{year}{2008}\natexlab{}.
\newblock \showarticletitle{A fast and accurate FFT-based method for pricing early-exercise options under L{\'e}vy processes}.
\newblock \bibinfo{journal}{\emph{SIAM Journal on Scientific Computing}} \bibinfo{volume}{30}, \bibinfo{number}{4} (\bibinfo{year}{2008}), \bibinfo{pages}{1678--1705}.
\newblock


\bibitem[MacKean(1965)]%
        {mackean1965free}
\bibfield{author}{\bibinfo{person}{Henry~P. MacKean}.} \bibinfo{year}{1965}\natexlab{}.
\newblock \showarticletitle{A free boundary problem for the heat equation arising from a problem in mathematical economics}.
\newblock \bibinfo{journal}{\emph{Industrial Management Review}}  \bibinfo{volume}{6} (\bibinfo{year}{1965}), \bibinfo{pages}{32--39}.
\newblock


\bibitem[Madan et~al\mbox{.}(1998)]%
        {madan1998variance}
\bibfield{author}{\bibinfo{person}{Dilip~B. Madan}, \bibinfo{person}{Peter~P. Carr}, {and} \bibinfo{person}{Eric~C. Chang}.} \bibinfo{year}{1998}\natexlab{}.
\newblock \showarticletitle{The variance gamma process and option pricing}.
\newblock \bibinfo{journal}{\emph{Review of Finance}} \bibinfo{volume}{2}, \bibinfo{number}{1} (\bibinfo{year}{1998}), \bibinfo{pages}{79--105}.
\newblock


\bibitem[Malas et~al\mbox{.}(2015)]%
        {Malas2015}
\bibfield{author}{\bibinfo{person}{Tareq Malas}, \bibinfo{person}{Georg Hager}, \bibinfo{person}{Hatem Ltaief}, \bibinfo{person}{Holger Stengel}, \bibinfo{person}{Gerhard Wellein}, {and} \bibinfo{person}{David Keyes}.} \bibinfo{year}{2015}\natexlab{}.
\newblock \showarticletitle{Multicore-optimized wavefront diamond blocking for optimizing stencil updates}.
\newblock \bibinfo{journal}{\emph{SIAM Journal on Scientific Computing}} \bibinfo{volume}{37}, \bibinfo{number}{4} (\bibinfo{year}{2015}), \bibinfo{pages}{C439--C464}.
\newblock


\bibitem[Malas et~al\mbox{.}(2014)]%
        {Malas2014TowardsEE}
\bibfield{author}{\bibinfo{person}{Tareq~M. Malas}, \bibinfo{person}{Georg Hager}, \bibinfo{person}{Hatem Ltaief}, {and} \bibinfo{person}{David~E. Keyes}.} \bibinfo{year}{2014}\natexlab{}.
\newblock \showarticletitle{Towards energy efficiency and maximum computational intensity for stencil algorithms using wavefront diamond temporal blocking}.
\newblock \bibinfo{journal}{\emph{ArXiv}}  \bibinfo{volume}{abs/1410.5561} (\bibinfo{year}{2014}).
\newblock


\bibitem[Mangeney et~al\mbox{.}(2002)]%
        {mangeney2002numerical}
\bibfield{author}{\bibinfo{person}{Andr{\'e} Mangeney}, \bibinfo{person}{Francesco Califano}, \bibinfo{person}{Carlo Cavazzoni}, {and} \bibinfo{person}{Pavel~M. Travnicek}.} \bibinfo{year}{2002}\natexlab{}.
\newblock \showarticletitle{A numerical scheme for the integration of the Vlasov--Maxwell system of equations}.
\newblock \bibinfo{journal}{\emph{J. Comput. Phys.}} \bibinfo{volume}{179}, \bibinfo{number}{2} (\bibinfo{year}{2002}), \bibinfo{pages}{495--538}.
\newblock


\bibitem[Matsuda(2004)]%
        {matsuda2004introduction}
\bibfield{author}{\bibinfo{person}{Kazuhisa Matsuda}.} \bibinfo{year}{2004}\natexlab{}.
\newblock \showarticletitle{Introduction to Merton jump diffusion model}.
\newblock \bibinfo{journal}{\emph{Department of Economics, The Graduate Center, The City University of New York, New York}} (\bibinfo{year}{2004}).
\newblock


\bibitem[Merton(1973)]%
        {merton1973theory}
\bibfield{author}{\bibinfo{person}{Robert~C. Merton}.} \bibinfo{year}{1973}\natexlab{}.
\newblock \showarticletitle{Theory of rational option pricing}.
\newblock \bibinfo{journal}{\emph{The Bell Journal of economics and management science}} (\bibinfo{year}{1973}), \bibinfo{pages}{141--183}.
\newblock


\bibitem[Merton(1976)]%
        {merton1976option}
\bibfield{author}{\bibinfo{person}{Robert~C. Merton}.} \bibinfo{year}{1976}\natexlab{}.
\newblock \showarticletitle{Option pricing when underlying stock returns are discontinuous}.
\newblock \bibinfo{journal}{\emph{Journal of financial economics}} \bibinfo{volume}{3}, \bibinfo{number}{1-2} (\bibinfo{year}{1976}), \bibinfo{pages}{125--144}.
\newblock


\bibitem[Metropolis and Ulam(1949)]%
        {metropolis1949monte}
\bibfield{author}{\bibinfo{person}{Nicholas Metropolis} {and} \bibinfo{person}{Stanislaw Ulam}.} \bibinfo{year}{1949}\natexlab{}.
\newblock \showarticletitle{The monte carlo method}.
\newblock \bibinfo{journal}{\emph{Journal of the American statistical association}} \bibinfo{volume}{44}, \bibinfo{number}{247} (\bibinfo{year}{1949}), \bibinfo{pages}{335--341}.
\newblock


\bibitem[Murray and Simmonds(1991)]%
        {murray1991numerical}
\bibfield{author}{\bibinfo{person}{Ross~J. Murray} {and} \bibinfo{person}{Ian Simmonds}.} \bibinfo{year}{1991}\natexlab{}.
\newblock \showarticletitle{A numerical scheme for tracking cyclone centres from digital data}.
\newblock \bibinfo{journal}{\emph{Australian meteorological magazine}} \bibinfo{volume}{39}, \bibinfo{number}{3} (\bibinfo{year}{1991}), \bibinfo{pages}{155--166}.
\newblock


\bibitem[Najm et~al\mbox{.}(1998)]%
        {najm1998semi}
\bibfield{author}{\bibinfo{person}{Habib~N. Najm}, \bibinfo{person}{Peter~S. Wyckoff}, {and} \bibinfo{person}{Omar~M. Knio}.} \bibinfo{year}{1998}\natexlab{}.
\newblock \showarticletitle{A semi-implicit numerical scheme for reacting flow: I. stiff chemistry}.
\newblock \bibinfo{journal}{\emph{J. Comput. Phys.}} \bibinfo{volume}{143}, \bibinfo{number}{2} (\bibinfo{year}{1998}), \bibinfo{pages}{381--402}.
\newblock


\bibitem[Oliveira(2014)]%
        {oliveira2014convolution}
\bibfield{author}{\bibinfo{person}{Pedro~S. Oliveira}.} \bibinfo{year}{2014}\natexlab{}.
\newblock \emph{\bibinfo{title}{The convolution method for pricing American options under L{\'e}vy processes}}.
\newblock \bibinfo{thesistype}{Ph.\,D. Dissertation}. \bibinfo{school}{University of Lisbon}.
\newblock


\bibitem[Palamadai~Natarajan et~al\mbox{.}(2017)]%
        {Ekanathan2017}
\bibfield{author}{\bibinfo{person}{Ekanathan Palamadai~Natarajan}, \bibinfo{person}{Maryam Mehri~Dehnavi}, {and} \bibinfo{person}{Charles Leiserson}.} \bibinfo{year}{2017}\natexlab{}.
\newblock \showarticletitle{Autotuning divide-and-conquer stencil computations}.
\newblock \bibinfo{journal}{\emph{Concurrency and Computation: Practice and Experience}} \bibinfo{volume}{29}, \bibinfo{number}{17} (\bibinfo{year}{2017}), \bibinfo{pages}{e4127}.
\newblock


\bibitem[Pang(1999)]%
        {pang1999introduction}
\bibfield{author}{\bibinfo{person}{Tao Pang}.} \bibinfo{year}{1999}\natexlab{}.
\newblock \bibinfo{booktitle}{\emph{An Introduction to Computational Physics}}.
\newblock \bibinfo{publisher}{American Association of Physics Teachers}.
\newblock


\bibitem[Paoli and Schatzman(2002)]%
        {paoli2002numerical}
\bibfield{author}{\bibinfo{person}{Laetitia Paoli} {and} \bibinfo{person}{Michelle Schatzman}.} \bibinfo{year}{2002}\natexlab{}.
\newblock \showarticletitle{A numerical scheme for impact problems I: The one-dimensional case}.
\newblock \bibinfo{journal}{\emph{SIAM J. Numer. Anal.}} \bibinfo{volume}{40}, \bibinfo{number}{2} (\bibinfo{year}{2002}), \bibinfo{pages}{702--733}.
\newblock


\bibitem[perf({[n.\,d.]})]%
        {perftool}
perf \bibinfo{year}{[n.\,d.]}\natexlab{}.
\newblock \bibinfo{title}{perf: Linux profiling with performance counters}.
\newblock \bibinfo{howpublished}{\url{https://perf.wiki.kernel.org/index.php/Main_Page}}.
\newblock


\bibitem[Peter and Packer(2007)]%
        {peter2007understanding}
\bibfield{author}{\bibinfo{person}{H{\~A} Peter} {and} \bibinfo{person}{Frank Packer}.} \bibinfo{year}{2007}\natexlab{}.
\newblock \showarticletitle{Understanding asset prices: an overview}.
\newblock \bibinfo{journal}{\emph{Bis Papers}} (\bibinfo{year}{2007}).
\newblock


\bibitem[Peyr{\'e}(2011)]%
        {peyre2011numerical}
\bibfield{author}{\bibinfo{person}{Gabriel Peyr{\'e}}.} \bibinfo{year}{2011}\natexlab{}.
\newblock \showarticletitle{The numerical tours of signal processing-advanced computational signal and image processing}.
\newblock \bibinfo{journal}{\emph{IEEE Computing in Science and Engineering}} \bibinfo{volume}{13}, \bibinfo{number}{4} (\bibinfo{year}{2011}), \bibinfo{pages}{94--97}.
\newblock


\bibitem[Prathumwan and Trachoo(2020)]%
        {prathumwan2020solution}
\bibfield{author}{\bibinfo{person}{Din Prathumwan} {and} \bibinfo{person}{Kamonchat Trachoo}.} \bibinfo{year}{2020}\natexlab{}.
\newblock \showarticletitle{On the solution of two-dimensional fractional Black--Scholes equation for European put option}.
\newblock \bibinfo{journal}{\emph{Advances in Difference Equations}} \bibinfo{volume}{2020}, \bibinfo{number}{1} (\bibinfo{year}{2020}), \bibinfo{pages}{1--9}.
\newblock


\bibitem[Rappaz et~al\mbox{.}(2010)]%
        {rappaz2010numerical}
\bibfield{author}{\bibinfo{person}{Michel Rappaz}, \bibinfo{person}{Michel Bellet}, {and} \bibinfo{person}{Michel Deville}.} \bibinfo{year}{2010}\natexlab{}.
\newblock \bibinfo{booktitle}{\emph{Numerical Modeling in Materials Science and Engineering}}. Vol.~\bibinfo{volume}{32}.
\newblock \bibinfo{publisher}{Springer Science \& Business Media}.
\newblock


\bibitem[Rendleman(1979)]%
        {rendleman1979two}
\bibfield{author}{\bibinfo{person}{Richard~J. Rendleman}.} \bibinfo{year}{1979}\natexlab{}.
\newblock \showarticletitle{Two-state option pricing}.
\newblock \bibinfo{journal}{\emph{The Journal of Finance}} \bibinfo{volume}{34}, \bibinfo{number}{5} (\bibinfo{year}{1979}), \bibinfo{pages}{1093--1110}.
\newblock


\bibitem[Renson et~al\mbox{.}(2016)]%
        {renson2016numerical}
\bibfield{author}{\bibinfo{person}{Ludovic Renson}, \bibinfo{person}{Ga{\"e}tan Kerschen}, {and} \bibinfo{person}{Bruno Cochelin}.} \bibinfo{year}{2016}\natexlab{}.
\newblock \showarticletitle{Numerical computation of nonlinear normal modes in mechanical engineering}.
\newblock \bibinfo{journal}{\emph{Journal of Sound and Vibration}}  \bibinfo{volume}{364} (\bibinfo{year}{2016}), \bibinfo{pages}{177--206}.
\newblock


\bibitem[Richardson(1911)]%
        {richardson1911approximate}
\bibfield{author}{\bibinfo{person}{Lewis~F. Richardson}.} \bibinfo{year}{1911}\natexlab{}.
\newblock \showarticletitle{The approximate arithmetical solution by finite differences with an application to stresses in masonry dams}.
\newblock \bibinfo{journal}{\emph{Philosophical Transactions of the Royal Society of America}}  \bibinfo{volume}{210} (\bibinfo{year}{1911}), \bibinfo{pages}{307--357}.
\newblock


\bibitem[Richardson and Gaunt(1927)]%
        {richardson1927viii}
\bibfield{author}{\bibinfo{person}{Lewis~F. Richardson} {and} \bibinfo{person}{John~A. Gaunt}.} \bibinfo{year}{1927}\natexlab{}.
\newblock \showarticletitle{VIII. The deferred approach to the limit}.
\newblock \bibinfo{journal}{\emph{Philosophical Transactions of the Royal Society of London. Series A, containing papers of a mathematical or physical character}} \bibinfo{volume}{226}, \bibinfo{number}{636-646} (\bibinfo{year}{1927}), \bibinfo{pages}{299--361}.
\newblock


\bibitem[Robert(1981)]%
        {robert1981stable}
\bibfield{author}{\bibinfo{person}{Andr{\'e} Robert}.} \bibinfo{year}{1981}\natexlab{}.
\newblock \showarticletitle{A stable numerical integration scheme for the primitive meteorological equations}.
\newblock \bibinfo{journal}{\emph{Atmosphere-Ocean}} \bibinfo{volume}{19}, \bibinfo{number}{1} (\bibinfo{year}{1981}), \bibinfo{pages}{35--46}.
\newblock


\bibitem[Robert(1982)]%
        {robert1982semi}
\bibfield{author}{\bibinfo{person}{Andr{\'e} Robert}.} \bibinfo{year}{1982}\natexlab{}.
\newblock \showarticletitle{A semi-Lagrangian and semi-implicit numerical integration scheme for the primitive meteorological equations}.
\newblock \bibinfo{journal}{\emph{Journal of the Meteorological Society of Japan. Ser. II}} \bibinfo{volume}{60}, \bibinfo{number}{1} (\bibinfo{year}{1982}), \bibinfo{pages}{319--325}.
\newblock


\bibitem[Runggaldier(2003)]%
        {runggaldier2003jump}
\bibfield{author}{\bibinfo{person}{Wolfgang~J. Runggaldier}.} \bibinfo{year}{2003}\natexlab{}.
\newblock \showarticletitle{Jump-diffusion models}.
\newblock In \bibinfo{booktitle}{\emph{Handbook of heavy tailed distributions in finance}}. \bibinfo{publisher}{Elsevier}, \bibinfo{pages}{169--209}.
\newblock


\bibitem[Samuelson et~al\mbox{.}(2011)]%
        {samuelson2011louis}
\bibfield{author}{\bibinfo{person}{Paul~A. Samuelson}, \bibinfo{person}{Alison Etheridge}, \bibinfo{person}{Mark Davis}, {and} \bibinfo{person}{Louis Bachelier}.} \bibinfo{year}{2011}\natexlab{}.
\newblock \bibinfo{booktitle}{\emph{Louis Bachelier's Theory of Speculation: The Origins of Modern Finance}}.
\newblock \bibinfo{publisher}{Princeton University Press}.
\newblock


\bibitem[Sato et~al\mbox{.}(2010)]%
        {Katsuto2010}
\bibfield{author}{\bibinfo{person}{Katsuto Sato}, \bibinfo{person}{Hiroyuki Takizawa}, \bibinfo{person}{Kazuhiko Komatsu}, {and} \bibinfo{person}{Hiroaki Kobayashi}.} \bibinfo{year}{2010}\natexlab{}.
\newblock \showarticletitle{Automatic tuning of CUDA execution parameters for stencil processing}.
\newblock \bibinfo{journal}{\emph{Software Automatic Tuning: From Concepts to State-of-the-Art Results}} (\bibinfo{year}{2010}), \bibinfo{pages}{209--228}.
\newblock


\bibitem[Sewell(2005)]%
        {sewell2005numerical}
\bibfield{author}{\bibinfo{person}{Granville Sewell}.} \bibinfo{year}{2005}\natexlab{}.
\newblock \bibinfo{booktitle}{\emph{The numerical solution of ordinary and partial differential equations}}. Vol.~\bibinfo{volume}{75}.
\newblock \bibinfo{publisher}{John Wiley \& Sons}.
\newblock


\bibitem[Sharpe(1978)]%
        {sharpe1978}
\bibfield{author}{\bibinfo{person}{William~F. Sharpe}.} \bibinfo{year}{1978}\natexlab{}.
\newblock \bibinfo{booktitle}{\emph{Investments}}.
\newblock \bibinfo{publisher}{Prentice Hall}.
\newblock


\bibitem[Shreve(2005)]%
        {shreve2005stochastic}
\bibfield{author}{\bibinfo{person}{Steven Shreve}.} \bibinfo{year}{2005}\natexlab{}.
\newblock \bibinfo{booktitle}{\emph{Stochastic calculus for finance I: the binomial asset pricing model}}.
\newblock \bibinfo{publisher}{Springer Science \& Business Media}.
\newblock


\bibitem[Shreve(2004)]%
        {shreve2004stochastic}
\bibfield{author}{\bibinfo{person}{Steven~E. Shreve}.} \bibinfo{year}{2004}\natexlab{}.
\newblock \bibinfo{booktitle}{\emph{Stochastic calculus for finance II: Continuous-time models}}. Vol.~\bibinfo{volume}{11}.
\newblock \bibinfo{publisher}{Springer}.
\newblock


\bibitem[Smith(1976)]%
        {smith1976option}
\bibfield{author}{\bibinfo{person}{Clifford~W. Smith}.} \bibinfo{year}{1976}\natexlab{}.
\newblock \showarticletitle{Option pricing: A review}.
\newblock \bibinfo{journal}{\emph{Journal of Financial Economics}} \bibinfo{volume}{3}, \bibinfo{number}{1-2} (\bibinfo{year}{1976}), \bibinfo{pages}{3--51}.
\newblock


\bibitem[Snider and Banerjee(2010)]%
        {snider2010heterogeneous}
\bibfield{author}{\bibinfo{person}{Dale Snider} {and} \bibinfo{person}{Sibashis Banerjee}.} \bibinfo{year}{2010}\natexlab{}.
\newblock \showarticletitle{Heterogeneous gas chemistry in the CPFD Eulerian--Lagrangian numerical scheme (ozone decomposition)}.
\newblock \bibinfo{journal}{\emph{Powder Technology}} \bibinfo{volume}{199}, \bibinfo{number}{1} (\bibinfo{year}{2010}), \bibinfo{pages}{100--106}.
\newblock


\bibitem[Stampede(ede2)]%
        {Stampede2}
Stampede \bibinfo{year}{Stampede2}\natexlab{}.
\newblock \bibinfo{title}{The Stampede2 supercomputing cluster}.
\newblock \bibinfo{howpublished}{\url{https://www.tacc.utexas.edu/systems/stampede2}}.
\newblock


\bibitem[Stein and Stein(1991)]%
        {stein1991stock}
\bibfield{author}{\bibinfo{person}{Elias~M. Stein} {and} \bibinfo{person}{Jeremy~C. Stein}.} \bibinfo{year}{1991}\natexlab{}.
\newblock \showarticletitle{Stock price distributions with stochastic volatility: an analytic approach}.
\newblock \bibinfo{journal}{\emph{The review of financial studies}} \bibinfo{volume}{4}, \bibinfo{number}{4} (\bibinfo{year}{1991}), \bibinfo{pages}{727--752}.
\newblock


\bibitem[Strikwerda(2004)]%
        {strikwerda2004finite}
\bibfield{author}{\bibinfo{person}{John~C. Strikwerda}.} \bibinfo{year}{2004}\natexlab{}.
\newblock \bibinfo{booktitle}{\emph{Finite difference schemes and partial differential equations}}.
\newblock \bibinfo{publisher}{SIAM}.
\newblock


\bibitem[Szilard(2004)]%
        {szilard2004theories}
\bibfield{author}{\bibinfo{person}{Rudolph Szilard}.} \bibinfo{year}{2004}\natexlab{}.
\newblock \showarticletitle{Theories and applications of plate analysis: classical, numerical and engineering methods}.
\newblock \bibinfo{journal}{\emph{Appl. Mech. Rev.}} \bibinfo{volume}{57}, \bibinfo{number}{6} (\bibinfo{year}{2004}), \bibinfo{pages}{B32--B33}.
\newblock


\bibitem[Taflove and Hagness(2005)]%
        {taflove2005computational}
\bibfield{author}{\bibinfo{person}{Allen Taflove} {and} \bibinfo{person}{Susan~C. Hagness}.} \bibinfo{year}{2005}\natexlab{}.
\newblock \bibinfo{booktitle}{\emph{Computational electrodynamics: the finite-difference time-domain method}}.
\newblock \bibinfo{publisher}{Artech house}.
\newblock


\bibitem[Tang et~al\mbox{.}(2011)]%
        {tang2011pochoir}
\bibfield{author}{\bibinfo{person}{Yuan Tang}, \bibinfo{person}{Rezaul~A. Chowdhury}, \bibinfo{person}{Bradley~C. Kuszmaul}, \bibinfo{person}{Chi-Keung Luk}, {and} \bibinfo{person}{Charles~E. Leiserson}.} \bibinfo{year}{2011}\natexlab{}.
\newblock \showarticletitle{The pochoir stencil compiler}. In \bibinfo{booktitle}{\emph{Proceedings of the twenty-third annual ACM symposium on Parallelism in algorithms and architectures}}. \bibinfo{pages}{117--128}.
\newblock


\bibitem[Thijssen(2007)]%
        {thijssen2007computational}
\bibfield{author}{\bibinfo{person}{Jos Thijssen}.} \bibinfo{year}{2007}\natexlab{}.
\newblock \bibinfo{booktitle}{\emph{Computational Physics}}.
\newblock \bibinfo{publisher}{Cambridge University Press}.
\newblock


\bibitem[Thomas(2013)]%
        {thomas2013numerical}
\bibfield{author}{\bibinfo{person}{James~W. Thomas}.} \bibinfo{year}{2013}\natexlab{}.
\newblock \bibinfo{booktitle}{\emph{Numerical partial differential equations: finite difference methods}}. Vol.~\bibinfo{volume}{22}.
\newblock \bibinfo{publisher}{Springer Science \& Business Media}.
\newblock


\bibitem[Tsitsiklis and Van~R.(2001)]%
        {tsitsiklis2001regression}
\bibfield{author}{\bibinfo{person}{John~N. Tsitsiklis} {and} \bibinfo{person}{Benjamin Van~R.}} \bibinfo{year}{2001}\natexlab{}.
\newblock \showarticletitle{Regression methods for pricing complex American-style options}.
\newblock \bibinfo{journal}{\emph{IEEE Transactions on Neural Networks}} \bibinfo{volume}{12}, \bibinfo{number}{4} (\bibinfo{year}{2001}), \bibinfo{pages}{694--703}.
\newblock


\bibitem[Van~M.(1976)]%
        {van1976optimal}
\bibfield{author}{\bibinfo{person}{Pierre Van~M.}} \bibinfo{year}{1976}\natexlab{}.
\newblock \showarticletitle{On optimal stopping and free boundary problems}.
\newblock \bibinfo{journal}{\emph{Archive for Rational Mechanics and Analysis}} \bibinfo{volume}{60}, \bibinfo{number}{2} (\bibinfo{year}{1976}), \bibinfo{pages}{101--148}.
\newblock


\bibitem[Van~R.(2012)]%
        {van2012numerical}
\bibfield{author}{\bibinfo{person}{Ursula Van~R.}} \bibinfo{year}{2012}\natexlab{}.
\newblock \bibinfo{booktitle}{\emph{Numerical methods in computational electrodynamics: linear systems in practical applications}}. Vol.~\bibinfo{volume}{12}.
\newblock \bibinfo{publisher}{Springer Science \& Business Media}.
\newblock


\bibitem[Vese and Osher(2002)]%
        {vese2002numerical}
\bibfield{author}{\bibinfo{person}{Luminita~A. Vese} {and} \bibinfo{person}{Stanley~J. Osher}.} \bibinfo{year}{2002}\natexlab{}.
\newblock \showarticletitle{Numerical methods for p-harmonic flows and applications to image processing}.
\newblock \bibinfo{journal}{\emph{SIAM J. Numer. Anal.}} \bibinfo{volume}{40}, \bibinfo{number}{6} (\bibinfo{year}{2002}), \bibinfo{pages}{2085--2104}.
\newblock


\bibitem[Vesely(1994)]%
        {vesely1994computational}
\bibfield{author}{\bibinfo{person}{Franz~J. Vesely}.} \bibinfo{year}{1994}\natexlab{}.
\newblock \bibinfo{booktitle}{\emph{Computational Physics}}.
\newblock \bibinfo{publisher}{Springer}.
\newblock


\bibitem[Vijayaraghavan and Keith(1990)]%
        {vijayaraghavan1990efficient}
\bibfield{author}{\bibinfo{person}{D. Vijayaraghavan} {and} \bibinfo{person}{Theo Keith}.} \bibinfo{year}{1990}\natexlab{}.
\newblock \showarticletitle{An efficient, robust, and time accurate numerical scheme applied to a cavitation algorithm}.
\newblock  (\bibinfo{year}{1990}).
\newblock


\bibitem[Weaver et~al\mbox{.}(2013)]%
        {papilib}
\bibfield{author}{\bibinfo{person}{Vincent~M. Weaver}, \bibinfo{person}{Dan Terpstra}, \bibinfo{person}{Heike McCraw}, \bibinfo{person}{Matt Johnson}, \bibinfo{person}{Kiran Kasichayanula}, \bibinfo{person}{James Ralph}, \bibinfo{person}{John Nelson}, \bibinfo{person}{Phil Mucci}, \bibinfo{person}{Tushar Mohan}, {and} \bibinfo{person}{Shirley Moore}.} \bibinfo{year}{2013}\natexlab{}.
\newblock \showarticletitle{{PAPI} 5: Measuring power, energy, and the cloud}. In \bibinfo{booktitle}{\emph{2013 {IEEE} International Symposium on Performance Analysis of Systems and Software ({ISPASS})}}. \bibinfo{publisher}{{IEEE}}.
\newblock
\urldef\tempurl%
\url{https://doi.org/10.1109/ispass.2013.6557155}
\showDOI{\tempurl}


\bibitem[Weickert(2001)]%
        {weickert2000applications}
\bibfield{author}{\bibinfo{person}{Joachim Weickert}.} \bibinfo{year}{2001}\natexlab{}.
\newblock \showarticletitle{Applications of nonlinear diffusion in image processing and computer vision}.
\newblock \bibinfo{journal}{\emph{Acta Math. Univ. Comenianae}} \bibinfo{volume}{70}, \bibinfo{number}{1} (\bibinfo{year}{2001}), \bibinfo{pages}{33--50}.
\newblock


\bibitem[Wolf and Lam(1991)]%
        {Wolf1991}
\bibfield{author}{\bibinfo{person}{Michael~E. Wolf} {and} \bibinfo{person}{Monica~S. Lam}.} \bibinfo{year}{1991}\natexlab{}.
\newblock \showarticletitle{A data locality optimizing algorithm}. In \bibinfo{booktitle}{\emph{Proceedings of the ACM SIGPLAN Conference on Programming Language Design and Implementation}}. \bibinfo{pages}{30--44}.
\newblock


\bibitem[Wolf et~al\mbox{.}(1996)]%
        {Wolf1996}
\bibfield{author}{\bibinfo{person}{Michael~E. Wolf}, \bibinfo{person}{Dror~E. Maydan}, {and} \bibinfo{person}{Ding-Kai Chen}.} \bibinfo{year}{1996}\natexlab{}.
\newblock \showarticletitle{Combining loop transformations considering caches and scheduling}. In \bibinfo{booktitle}{\emph{Proceedings of the IEEE/ACM International Symposium on Microarchitecture}}. \bibinfo{pages}{274--286}.
\newblock


\bibitem[Wolfe(1987)]%
        {Wolfe1987}
\bibfield{author}{\bibinfo{person}{Michael~J. Wolfe}.} \bibinfo{year}{1987}\natexlab{}.
\newblock \showarticletitle{Iteration space tiling for memory hierarchies}.
\newblock \bibinfo{journal}{\emph{Parallel Processing for Scientific Computing}}  \bibinfo{volume}{357} (\bibinfo{year}{1987}), \bibinfo{pages}{361}.
\newblock


\bibitem[Wonnacott(2002)]%
        {Wonnacott2002}
\bibfield{author}{\bibinfo{person}{David Wonnacott}.} \bibinfo{year}{2002}\natexlab{}.
\newblock \showarticletitle{Achieving scalable locality with time skewing}.
\newblock \bibinfo{journal}{\emph{International Journal of Parallel Programming}} \bibinfo{volume}{30}, \bibinfo{number}{3} (\bibinfo{year}{2002}), \bibinfo{pages}{181--221}.
\newblock


\bibitem[Yuan and Agrawal(2002)]%
        {yuan2002numerical}
\bibfield{author}{\bibinfo{person}{Lixia Yuan} {and} \bibinfo{person}{Om~P. Agrawal}.} \bibinfo{year}{2002}\natexlab{}.
\newblock \showarticletitle{A numerical scheme for dynamic systems containing fractional derivatives}.
\newblock \bibinfo{journal}{\emph{Journal of Vibration and Acoustics}} \bibinfo{volume}{124}, \bibinfo{number}{2} (\bibinfo{year}{2002}), \bibinfo{pages}{321--324}.
\newblock


\bibitem[Zhang et~al\mbox{.}(2004)]%
        {zhang2004numerical}
\bibfield{author}{\bibinfo{person}{Jianfeng Zhang} {et~al\mbox{.}}} \bibinfo{year}{2004}\natexlab{}.
\newblock \showarticletitle{A numerical scheme for BSDEs}.
\newblock \bibinfo{journal}{\emph{Annals of Applied Probability}} \bibinfo{volume}{14}, \bibinfo{number}{1} (\bibinfo{year}{2004}), \bibinfo{pages}{459--488}.
\newblock


\bibitem[Zhang et~al\mbox{.}(2015)]%
        {Zhang2015A3D}
\bibfield{author}{\bibinfo{person}{Jiyuan Zhang}, \bibinfo{person}{Tze~Meng Low}, \bibinfo{person}{Qi Guo}, {and} \bibinfo{person}{Franz Franchetti}.} \bibinfo{year}{2015}\natexlab{}.
\newblock \showarticletitle{A 3 d-stacked memory manycore stencil accelerator system}.
\newblock  (\bibinfo{year}{2015}).
\newblock


\bibitem[Zhang et~al\mbox{.}(2006)]%
        {zhang2006weighted}
\bibfield{author}{\bibinfo{person}{Peng Zhang}, \bibinfo{person}{Sze~Chun Wong}, {and} \bibinfo{person}{Chi-Wang Shu}.} \bibinfo{year}{2006}\natexlab{}.
\newblock \showarticletitle{A weighted essentially non-oscillatory numerical scheme for a multi-class traffic flow model on an inhomogeneous highway}.
\newblock \bibinfo{journal}{\emph{J. Comput. Phys.}} \bibinfo{volume}{212}, \bibinfo{number}{2} (\bibinfo{year}{2006}), \bibinfo{pages}{739--756}.
\newblock


\bibitem[Zhu(2006)]%
        {zhu2006new}
\bibfield{author}{\bibinfo{person}{Song-Ping Zhu}.} \bibinfo{year}{2006}\natexlab{}.
\newblock \showarticletitle{A new analytical approximation formula for the optimal exercise boundary of American put options}.
\newblock \bibinfo{journal}{\emph{International Journal of Theoretical and Applied Finance}} \bibinfo{volume}{9}, \bibinfo{number}{07} (\bibinfo{year}{2006}), \bibinfo{pages}{1141--1177}.
\newblock


\bibitem[Zhu and He(2007)]%
        {zhu2007calculating}
\bibfield{author}{\bibinfo{person}{Song-Ping Zhu} {and} \bibinfo{person}{Zhi-Wei He}.} \bibinfo{year}{2007}\natexlab{}.
\newblock \showarticletitle{Calculating the early exercise boundary of American put options with an approximation formula}.
\newblock \bibinfo{journal}{\emph{International Journal of Theoretical and Applied Finance}} \bibinfo{volume}{10}, \bibinfo{number}{07} (\bibinfo{year}{2007}), \bibinfo{pages}{1203--1227}.
\newblock


\bibitem[Zhylyevskyy(2010)]%
        {zhylyevskyy2010fast}
\bibfield{author}{\bibinfo{person}{Oleksandr Zhylyevskyy}.} \bibinfo{year}{2010}\natexlab{}.
\newblock \showarticletitle{A fast Fourier transform technique for pricing American options under stochastic volatility}.
\newblock \bibinfo{journal}{\emph{Review of Derivatives Research}} \bibinfo{volume}{13}, \bibinfo{number}{1} (\bibinfo{year}{2010}), \bibinfo{pages}{1--24}.
\newblock


\bibitem[Zubair and Mukkamala(2008)]%
        {Zubair}
\bibfield{author}{\bibinfo{person}{Mohammad Zubair} {and} \bibinfo{person}{Ravi Mukkamala}.} \bibinfo{year}{2008}\natexlab{}.
\newblock \showarticletitle{High performance implementation of binomial option pricing}. In \bibinfo{booktitle}{\emph{Computational Science and Its Applications--ICCSA 2008: International Conference, Perugia, Italy, June 30--July 3, 2008, Proceedings, Part I 8}}. Springer, \bibinfo{pages}{852--866}.
\newblock


\end{thebibliography}
